\documentclass[12pt]{article}

\usepackage{verbatim,color,amssymb,bm}
\usepackage{amsmath,bbm}
\usepackage{amsthm}					
\usepackage{natbib}
\usepackage{multirow}
\usepackage{setspace}
\usepackage[mathscr]{euscript}
\usepackage{fancyhdr}
\usepackage{enumitem}
\usepackage{graphicx,grffile}
\usepackage{geometry}
\usepackage{xcolor}[dvipsnames]
\usepackage{physics,xfrac}
\usepackage[colorlinks=true,citecolor=blue,linkcolor=blue]{hyperref}
\usepackage{caption,subcaption}
\usepackage{soul}
\usepackage[linesnumbered,ruled,vlined]{algorithm2e}
\usepackage{algpseudocode}
\usepackage[export]{adjustbox}
\graphicspath{{plots/}}

\usepackage{multibib}
\newcites{latex}{References}

\usepackage{tabularx, booktabs}
\newcolumntype{Y}{>{\centering\arraybackslash}X}
\usepackage{lscape}

\usepackage{float}
\usepackage{lineno}

\newtheorem{lemma}{Lemma}
\newtheorem{theorem}{Theorem}

\newtheorem{corollary}{Corollary}
\newtheorem{definition}{Definition}

\def\K{{\cal K}}

\def\diag{\hbox{diag}}

\def\wh{\widehat}
\def\wt{\widetilde}

\def\diag{\hbox{diag}}

\newcommand{\mn}{{\mathrm{N}}}

\def\Beta{\mathrm{Beta}}

\def\Dir{\mathrm{Dir}}
\def\DP{\mathrm{DP}}
\newcommand{\DL}{\mathrm{DL}}
\newcommand{\NIW}{{\mathrm{NIW}}}

\def\Exp{\mathrm{Exp}}

\def\Ga{\mathrm{Ga}}

\newcommand\Lamb{Lamb }

\def\P_25_ICML{{\it Proceedings of the 25th international conference on Machine learning}}

\def\bse{\begin{eqnarray*}}
\def\ese{\end{eqnarray*}}
\def\be{\begin{eqnarray}}
\def\ee{\end{eqnarray}}
\def\bq{\begin{equation}}
\def\eq{\end{equation}}

\def\wh{\widehat}

\def\trans{^{\rm T}}
\newcommand{\transp}{\rm T}

\def\bone{{\mathbf 1}}

\def\b1e{\bm{e}}
\def\b1f{\bm{f}}

\def\by{\bm{y}}

\def\simind{\stackrel{\mbox{\scriptsize{ind}}}{\sim}}
\def\simiid{\stackrel{\mbox{\scriptsize{iid}}}{\sim}}

\newcommand{\half}{\sfrac{1}{2}}
\newcommand{\de}{\mathrm{d}}

\newcommand{\kn}{k_{n}}

\newcommand{\boldeta}{\mathcal{ \eta}}
\newcommand{\boldzeta}{ \zeta}
\newcommand{\Je}{\mathcal{J} }
\newcommand{\truncpi}{\tilde{\Pi}_p} 
\newcommand{\R}{\mathbb{R}}
\newcommand{\ke}{\mathcal{K}}
\newcommand{\p}{\mathbb{P}_0^p}
\newcommand{\ptheta}{\mathbb{P}_\vartheta^p}

\newcommand{\smin}{s_{\min}}
\newcommand{\smax}{s_{\max}}

\renewcommand\footnoterule{\kern-3pt \hrule \textwidth 2in \kern 2.6pt}

\usepackage[normalem]{ulem}

\def\boxit#1{\vbox{\hrule\hbox{\vrule\kern6pt \vbox{\kern6pt \textcolor{blue}{#1}\kern6pt}\kern6pt\vrule}\hrule}}

\def\authorfootnote#1{{\let\thefootnote\relax\footnotetext{#1}}}

\allowdisplaybreaks

\begin{document}
\thispagestyle{empty}
\baselineskip=28pt

\newcommand{\papertitle}{\LARGE{\bf Escaping the Curse of Dimensionality in\\
		\vskip-1ex
		 Bayesian Model-based Clustering}}
\begin{center}

\papertitle
		 \end{center}
\baselineskip=12pt
\vskip 12pt 

\newcommand{\authors}{\begin{center}
		Noirrit Kiran Chandra$^{a}$ (noirritchandra@gmail.com)\\
		Antonio Canale$^{b}$ (canale@stat.unipd.it)\\
		David B. Dunson$^{c}$ (dunson@duke.edu)\\

		\vskip 7mm
		$^{a}$Department of Mathematical Sciences, \\
		The University of Texas at Dallas, Richardson, TX, USA 
		\vskip 5pt
		$^{b}$Dipartimento di Scienze Statistiche\\
		 Universit\`a degli Studi di Padova, Padova, Italy\\
		\vskip 5pt
		$^{c}$Departments of Statistical Science and Mathematics\\
		 Duke University, Durham, NC, USA
\end{center}}
\authors

\vskip 12pt 
\begin{center}
{\Large{\bf Abstract}} 
\end{center}
\baselineskip=14pt
Bayesian mixture models are widely used for clustering of high-dimensional data with appropriate uncertainty quantification.
However, as the dimension of the observations increases, posterior inference often tends to favor too many or too few clusters. 
This article explains this behavior by studying the random partition posterior in a non-standard setting with a fixed sample size and increasing data dimensionality.  
We provide conditions under which the finite sample posterior tends to either assign every observation to a different cluster or all observations to the same cluster as the dimension grows. 
Interestingly, the conditions do not depend on the choice of clustering prior, as long as all possible partitions of observations into clusters have positive prior probabilities, and hold irrespective of the true data-generating model.
We then propose a class of latent mixtures for Bayesian clustering (Lamb) on a set of low-dimensional latent variables inducing a partition on the observed data.
The model is amenable to scalable posterior inference and we show that it can avoid the pitfalls of high-dimensionality under mild assumptions.
The proposed approach is shown to have good performance in simulation studies and an application to inferring cell types based on scRNAseq.

\vskip 20pt 
\baselineskip=12pt
\noindent\underline{\bf Key Words}: 
Big data; 
Clustering; 
Dirichlet process; 
Exchangeable partition probability function; 
High dimensional; 
Latent variables; 
Mixture model.

\par\medskip\noindent
\underline{\bf Short/Running Title}: Bayesian High-dimensional Clustering

\par\medskip\noindent
\underline{\bf Corresponding Author}: Noirrit Kiran Chandra (noirritchandra@gmail.com)

\pagenumbering{arabic}
\setcounter{page}{0}
\newlength{\gnat}
\setlength{\gnat}{26pt}
\baselineskip=\gnat
\vskip 4cm

\newpage

\section{Introduction}
\label{intro} 
High-dimensional data $y_i = (y_{i1}, \dots, y_{ip})\trans$ for $i = 1, \dots, n$, with $p \gg n$, have become commonplace, and there is routinely interest in clustering observations $\{1,\ldots,n\}$ into groups.  
As an illustrative application, we consider  single-cell RNA sequencing {(scRNASeq)} data; clustering of the cells based on their high-dimensional gene expression profiles produces potential cell types and provides information on heterogeneous cell populations of potential utility in disentangling carcinogenic processes.
RNAseq data is an exemplary setting in which $p$ is massive and clustering is crucial due to interest in inferring cell types.   Although there are a variety of alternatives in the literature \citep[see][for a review]{kiselev2019challenges}, we are particularly motivated to consider a Bayesian approach due to the potential for propagating uncertainty in inferring cell types. 
Additionally hierarchical Bayes models allow for borrowing of information in a principled manner in complicated scenarios.

Bayesian clustering is typically based on mixture models of the form: 
\begin{eqnarray}
	{y_{i} \stackrel{iid}{\sim} f,\quad f(y) = \sum_{h=1}^{k} \pi_{h} \mathcal{K}( y; \theta_{h})}, \label{eq:mix}
\end{eqnarray}
where $f(\cdot)$ is the marginal density of the data,
$k$ is the number of components, 
$\pi = (\pi_1,\ldots,\pi_k)\trans$ are probability weights, $\mathcal{K}(y; \theta_{h})$ is the density of the data within component $h$, {and the number of clusters in data $y_1,\ldots,y_n$ corresponds to the number of occupied components $\kn \le k$}. 
When $p$ is large and $y_i \in \R^p$, a typical approach chooses $\mathcal{K}(y; \theta_h)$ as a multivariate Gaussian density with a constrained and parsimonious  covariance \citep[see][for a review]{bouveyron2014model}. 
Examples include matrices that are
diagonal \citep{banfield1993model}, block diagonal \citep{galimberti2013using} or have a factor analytic representation \citep{ghahramani1996algorithm}.  

To avoid sensitivity to a pre-specified $k$, one can place a prior on $k$ to induce a mixture of finite mixture model \citep{miller2018, fruhwirth2021}.  
Alternatively, a Bayesian nonparametric approach lets $k=\infty$, which allows $k_n$ to increase without a bound as $n$ increases.  
Under a Dirichlet process \citep{ferguson1973bayesian} $\kn$ increases at a log rate in $n$, while for a Pitman-Yor process \citep{pitman1997two} the rate is a power law.

Notably, Bayesian approaches can be used to intrinsically regularize the model complexity, as discussed by  \citet{ockham_razor} exploiting the idea of a `Bayesian Ockham razor'. While in many circumstances relying on the Bayesian Ockham razor is sufficient to choose the appropriate compromise between extremes, e.g.  too many or too few clusters, in what follows we will argue that this is not the case in high-dimensional clustering.\label{pg:more_ockam} Indeed, when $p$ is very large, the posterior distribution of  $k_n$ can concentrate on large values \citep{celeuxetal:com_sol}; often the posterior mode of $\kn$ is even  equal to $n$ so that each subject is assigned to its own singleton cluster. 
\label{ch:smallpexample1}Consider, for example, the right panel of Figure \ref{fig:simpleboxplots} in the supplementary materials, which displays the distribution of the mean number of clusters in 100 replicates of a simple simulation example where we generate  samples of size $n=10$ from a $p=20$ variate normal distribution with mean zero and identity covariance. 
The boxplot, obtained running a standard Dirichlet process mixture, clearly shows how $\kn$ is concentrated near $n$ even for this moderate value of $p$.\label{pg:pg2}
\citet{celeuxetal:com_sol} conjectured that this aberrant behavior is mainly due to slow mixing of Markov chain Monte Carlo samplers.   
\citet{fruhwirth2006finite} combat this problem with a specific prior elicitation criterion; this can be successful for $p \approx 100$, but calibration of hyperparameters is a delicate issue and scaling to $p>1,000$ is problematic.  
Alternatively, one may attempt to cluster in lower dimensions via 
variable selection in clustering \citep{tadesse2005bayesian, kim2006variable} or by introducing both global and  variable-specific clustering indices for each subject, so that only certain variables inform global cluster allocation \citep{dunson2009nonparametric}.

However, we find these approaches complicated and to not address the fundamental question of what is causing the poor performance of Bayesian clustering for large $p$.  To fill this gap, we provide theory showing that, as $p \to \infty$ with $n$ fixed, the posterior can assign probability one to a trivial clustering - either with $\kn =1$ and all subjects in one cluster or with $\kn=n$ and every subject in a different cluster.  We further show that the conditions under which these degenerate limiting behaviors occur are satisfied for seemingly standard priors and multivariate Gaussian kernels. 
In a related result for classification, \citet{bickel2004} showed that when $p$ increases at a faster rate than $n$, the Fisher's linear discriminant rule is equivalent to randomly assigning future observations to the existing classes. 

Our result has no relationship with the literature studying  the posterior behavior of $k_n$  as $n \to \infty$ for nonparametric Bayes procedures \citep{miller2014,cai2020finite,ascolani2022clustering}.  Indeed, our result holds for finite $n$ regardless of the true data generating model, and has fundamentally different implications---in particular, that one needs to be extremely careful in specifying the kernel $\mathcal{K}(y; \theta)$ and prior for $\theta$ in the large $p$ context. Otherwise, the true posterior can produce degenerate clustering results that have nothing to do with true structure in the data. 

A key question is whether it is possible to define models that can circumvent this pitfall? We show that the answer is yes if 
clustering is conducted on the level of low-dimensional latent variables $\eta_i$ underlying $y_i$.  When the dimension of $\eta_i$ is small relative to $p$,
$y_i$ provides abundant information about the lower-dimensional $\eta_i$ even in low signal-to-noise settings in which each individual $y_{ij}$ contributes very little information on its own. Hence, the curse of dimensionality can be turned into a blessing.  This motivates a novel notion of a Bayesian oracle for clustering.  The oracle has knowledge of the latent $\eta_i$s and defines a Bayesian mixture model for clustering based on the $\eta_i$s; the resulting oracle clustering posterior is thus free of the curse of dimensionality.    We propose a particular latent mixture model structure, which can be shown to satisfy this oracle property and additionally leads to straightforward computation.

The article is organized as follows.
Section \ref{sec:pitfall} gives details on the limiting behavior of usual clustering methods based on (\ref{eq:mix}).  
Section \ref{sec:lamb} introduces our mixture model on the latent variable level with prior specifications and posterior computation strategies.
In Section \ref{sec:properties}, we introduce a Bayesian oracle clustering rule and show that our model achieves this oracle property as the dimension grows to infinity.
Section \ref{sec:simstudy} shows simulation studies illustrating how our proposed model learns the latent space with increasing dimensions and 
compares our method with some popular clustering methods.
{Section \ref{sec:realdata} considers an application to scRNASeq data, and Section \ref{sec:discussion} discusses the results.} 
Proofs of the main results are included in the Appendix, while additional simulation results, theorems and proofs are reported in the supplementary materials.

\section{Limiting Behavior of High-Dimensional Bayesian Clustering}
\label{sec:pitfall}

Under a general Bayesian framework, model \eqref{eq:mix} {becomes}  
\begin{equation}
	{y_{i} \sim f, \quad f(y) = \sum_{h\geq 1}\pi_h \K{y}{\theta_{h}} , \quad \theta_{h} \simiid P_0, \quad \{\pi_{h}\} \sim Q_{0}} ,  \label{eq:NPB}
\end{equation}
where {$\{\pi_h\} \sim Q_0 $ denotes a suitable prior for the mixture weights. Examples include stick-breaking \citep{sethuraman1994} constructions or a $k$-dimensional Dirichlet distribution with the dimension $k$ given a prior 
	following a mixture of finite mixtures (MFMs) approach.} 

Let $c_i \in \{1, \dots, \infty\}$ denote the cluster label for subject $i$ (for $i=1,\ldots,n$), with $\kn = \#\{c_{1},\ldots,c_{n}\}$ denoting the number of clusters represented in the sample. Conditionally on $c_i = h$, we can write  $y_i \mid c_i=h \sim \K{y_i}{\theta_{h}}$.  Assume that $n_j$ is the size of the $j$th cluster with $\sum_{j=1}^{k} n_j = n$. The posterior
probability of observing the partition $\Psi$ induced by the clusters $c_1,\ldots,c_n$ conditionally on the data $\by = \{y_{1}, \dots, y_{n}\}$ is \begin{equation}
	\Pi(\Psi \mid \by) = \frac{\Pi(\Psi)\times \prod_{h\geq 1} \int \prod_{i:c_i=h} \K{y_i}{\theta} \de  P_{0}(\theta)}{
		\sum_{\Psi' \in {\mathscr P}} {\Pi}(\Psi') 
		\times \prod_{h\geq 1} \int \prod_{i:c_{i}=h} \K{y_i}{\theta} \de P_0(\theta)},
	\label{eq:marginalposterior}
\end{equation}
where ${\mathscr P}$ is the space of all possible partitions of $n$ data points into clusters.
The numerator of \eqref{eq:marginalposterior} is the product of the prior probability of $\Psi$ multiplied by a product of the marginal likelihoods of the observations within each cluster.
The denominator is a normalizing constant consisting of 
an enormous sum over $\mathscr{P}$.
{Assuming exchangeability, the prior probability of any partition of $n$ subjects into $k_n$ groups depends only on $n_1,\ldots,n_{k_n}$ and $k_n$ through an exchangeable partition probability function (EPPF).  The latter is available in closed form for popular choices of $Q_{0}$, including the Dirichlet process, Pitman-Yor process and certain MFMs.}

The posterior \eqref{eq:marginalposterior} forms the basis for Bayesian inferences on clusterings in the data, while providing a characterization of uncertainty.  We are particularly interested in how this posterior behaves in the case in which  $y_i = (y_{i1},\ldots,y_{ip})\trans$ are high-dimensional so that $p$ is very large.  To study this behavior theoretically, we consider the limiting case as $p \to \infty$ while keeping $n$ fixed.  
This setting is quite appropriate in our motivating applications to genomics, as there is essentially no limit to the number of variables one can measure on each study subject, while the number of study subjects is often small to moderate.  

In such settings with enormous $p$ and modest $n$, we would like the true posterior distribution in \eqref{eq:marginalposterior} to provide a realistic characterization of clusters in the data.  
However,  this is commonly not the case and as $p$ increases the posterior distribution can have one of two trivial degenerate limits. 
In particular, depending on the choice of kernel density $\ke(\cdot; \theta)$  and the base measure $P_{0}$ for the $\theta_{h}$'s, the posterior assigns probability one to either the $k_n=1$ clustering that places all subjects in the same cluster or the $k_{n}=n$ clustering that places all subjects in different clusters. We derive sufficient conditions behind such aberrant behaviors as
formalized in the following theorem.

\begin{theorem}
	Let $y_1, \dots, y_n$ denote $p$-variate random vectors with joint probability measure $\p$.
	Let $\Psi$ denote the partition induced by the cluster labels $c_1,\ldots,c_n$, and let $c_1',\ldots,c_n'$ denote a new set of cluster labels obtained from $c_1,\ldots,c_n$ by merging an arbitrary pair of clusters, with $\Psi'$ the related partition. 
	Assume $Q_{0}(\pi_h>0\text{ for all }h=1,\dots,n)>0$. 
	If 
	\begin{equation*}
		\limsup_{p\rightarrow\infty} \frac{
			\prod_{h\geq1}\int  \prod_{i: c_{i}=h}  \K{y_{i}}{\theta} \de {P_{0}}(\theta)}{
			\prod_{h\geq1}\int  \prod_{i: c'_{i}=h}  \K{y_{i}}{\theta}
			\de {P_{0}} (\theta)}  = 0
	\end{equation*}
	in $\p$-probability,
	then $\lim_{p\rightarrow\infty} \Pi(c_{1}= \cdots = c_{n} \mid \by)=1$ in  $\p$-probability.
	Else if  
	\begin{equation*}
		\liminf_{p\rightarrow\infty} \frac{
			\prod_{h\geq1}\int  \prod_{i: c_{i}=h}  \K{y_{i}}{\theta} \de  {P_{0}} (\theta)}{
			\prod_{h\geq1} \int  \prod_{i: c'_{i}=h}  \K{y_{i}}{\theta}
			\de   {P_{0}} (\theta)}  = \infty
	\end{equation*}
	in  $\p$-probability, then	$\lim_{p\rightarrow\infty} \Pi(c_{1}\neq \cdots \neq c_{n} \mid \by )=1$ in  $\p$-probability.
	\label{th:negative_result}
\end{theorem}

The condition on $Q_{0}$ is equivalent to saying $\kn$ has positive prior mass on $1,\dots,n$, which is extremely mild and holds for essentially any prior in the literature, including the Dirichlet process, Pitman-Yor process and suitable MFMs that do not pre-specify $k < n$.\label{pg:Q_0_cond}
Changing the condition to $Q_{0}(\pi_h>0\text{ for all }h=1,\dots,k)>0$ with $k<n$, i.e. using a finite mixture model, leads to similar results. Specifically, if the first condition in Theorem \ref{th:negative_result} holds, then also for finite mixtures we will have a single occupied cluster comprising all samples.
If the opposite condition holds, instead, then all of the $k$ mixture components will be occupied. Both results are trivial modifications of the proof of Theorem \ref{th:negative_result}.\label{pg:finite_case}

Theorem \ref{th:negative_result} has disturbing implications in terms of the behavior of posterior distributions for Bayesian clustering in large $p$ settings. 
Notably, the theorem is stated for very general kernel density $\ke$ and base measure $P_{0}$, and the behavior is controlled by the induced marginal likelihoods obtained in integrating out the kernel parameter $\theta$ with respect to $P_0$.
Clearly it is the joint effect of $\ke$ and $P_{0}$ that leads to the two limiting results and thus it is not immediate to convert the statement of the theorem to simple conditions on $\ke$ and $P_{0}$. \label{pg:intuition_on_funcs}
However,  as we will discuss in detail, we can argue that these conditions are related to the two extreme situations of complex over-parametrized models 
having insufficiently informative priors and simpler models equipped with more informative priors.
To be more precise, consider the important and widely used special case corresponding to a location-scale mixture of multivariate Gaussian kernels: \begin{equation}
	y_{i} \simiid f, \quad f(y)  =  \sum_{h\geq 1} \pi_h \mn_p(y; \mu_{h}, \Sigma_{h}),\quad 
	(\mu_{h},\Sigma_{h}) \simiid P_{0},
	\label{eq:neg_model}
\end{equation}
where $\mn_{p}( \mu,\Sigma)$ denotes the $p$-dimensional multivariate normal density with mean $\mu$ and covariance matrix $\Sigma$.
We give two practical examples of Theorem \ref{th:negative_result} in
Corollary \ref{cor:neg_result2} and \ref{cor:neg_result}. 	
Let $ \lambda_{\min}(A)$ and  $ \lambda_{\max}(A)$ be the smallest and largest eigenvalues of a positive definite matrix $A$ and  $Y=[y_{1},\dots,y_{n}]\trans$ be the complete $n\times p$ data matrix.  Assume, for the true data generating distribution on the data $\by$,
\begin{enumerate}[label={(A\arabic*)}]
	\setcounter{enumi}{-1} 
	\item\label{cond_neg} $\liminf_{p\to\infty} \lambda_{\min}(Y Y\trans)/p>0$ in $\p$-probability and $\norm{y_i}^2\leq Kp$ for some $K>0$ in $\p$-probability. %
\end{enumerate}

Condition \ref{cond_neg}  is extremely mild ensuring that the data are non-atomic and is satisfied for any continuous distribution with finite second order moments.
Letting $\mathrm{IW}(\nu,\Lambda)$ denote an inverse-Wishart distribution with degrees of freedom $\nu$ and scale matrix $\Lambda$, we have the following:

\begin{corollary}
	\label{coroll1}
	Assume that the model \eqref{eq:neg_model} is used to cluster $\by$
	with $\Sigma_{h}\simiid \mathrm{IW}(\nu_{0},\Lambda_{0})$ and $\mu_{h} \mid \Sigma_{h} \simind \mn_{p}(\mu_{0},\kappa_{0}^{-1} \Sigma_{h})$,
	with $\norm{\mu_{0}}^2=O(p)$, $\kappa_{0}=O(1)$, $\nu_{0}=p+c$ for some fixed constant $c\geq 0$,  $\norm{\Lambda_{0}}_{2}=O(1)$ and $\norm{\Lambda_{0}}_{2}/\lambda_{\min}(\Lambda_{0})=O(1)$. 
	Under \ref{cond_neg} on the data $\by$, 
	${\Pi}(c_{1} = \cdots = c_{n} \mid \by )\rightarrow 1$ in  $\p$-probability.	
	\label{cor:neg_result2}
\end{corollary}

\begin{corollary}
	\label{coroll2}
	Assume that the model \eqref{eq:neg_model} is used to cluster $\by$
	with $\Sigma_h= \Sigma$ across all clusters, and let $\Sigma\sim \mathrm{IW}(\nu_{0},\Lambda_{0})$ and
	$\mu_{h}\mid \Sigma \simiid \mn_{p}(\mu_{0},\kappa_0^{-1} \Sigma)$,
	with $\norm{\mu_{0}}^2=O(p)$, $\kappa_{0}=O(1)$, $\nu_{0}>p-1$ such that $\lim_{p\to \infty}\nu_{0}/p>1$, and $\norm{\Lambda_{0}}_{2}=O(1)$ with $\norm{\Lambda_{0}}_{2}/\lambda_{\min}(\Lambda_{0})=O(1)$. 
	Under \ref{cond_neg} on the data $\by$,  ${\Pi}(c_{1} \neq  \cdots \neq c_{n} \mid  \by)\rightarrow1$ in  $\p$-probability.
	\label{cor:neg_result}
\end{corollary}

\label{cor_assumptions}
Bayesian model-based clustering routinely uses these setups for the kernel parameters and priors 
\citep{fruhwirth2019handbook}.
The conditions on $\mu_{0}$ and $\kappa_{0}$ ensure that the Euclidean norm of the prior mean grows with $p$ in the same order as the data $\{y_{i}\}$,
and the conditions on the scale matrix $\Lambda_{0}$ imply that the second moments of the location components are a priori bounded away from 0 while being finite; similar assumptions appear in \citet{yao2022bayesian} in a study on high-dimensional Gaussian location mixture models.
In terms of the degrees of freedom parameter $\nu_{0}$, 
in Corollary \ref{cor:neg_result2} the ratio $\nu_{0}/p$ is 1 in the limit inducing a heavy tailed prior predictive distribution, whereas in Corollary \ref{cor:neg_result} a thinner tailed prior predictive is induced.
Corollaries \ref{cor:neg_result2} and \ref{cor:neg_result} show that, for mixtures of Gaussians, we can obtain directly opposite aberrant limiting behavior of the posterior depending on the kernel and prior for the kernel parameters but not on the clustering prior $Q_{0}$.

Corollary \ref{cor:neg_result2} considers the case in which we allow flexible cluster-specific means and dispersion matrices, under typical conjugate multivariate normal $\mathrm{IW}$ priors.  \label{pg:parametrization}
This case can be viewed as a complex over-parametrized model
as $p$ increases and to combat this complexity the Bayesian Ockham razor \citep{ockham_razor} automatically assigns probability one to grouping all $n$ individuals into the same cluster effectively simplifying the model. At the other extreme, covered by Corollary \ref{cor:neg_result}, we assume an under-parametrized relatively simplistic model structure in which all the mixture components have a common covariance.  In this case, due perhaps to the relatively concentrated prior predictive distribution, there is not enough penalty for introducing new clusters, and all individuals are assigned to their own singleton cluster.
These results hold regardless of the true data-generating model, and in particular the true clustering structure.


These theoretical results demonstrate that in high dimensions it is crucial to choose a good compromise between parsimony and flexibility in Bayesian model-based clustering.  
Otherwise, the true posterior distribution of clusterings in the data can have effectively no relationship whatsoever with true clustering structure in the data. 
Although we focus on the limiting case as $p \to \infty$, we conjecture that this  behavior can `kick in' quickly as $p$ increases, based on intuition built through our proofs and through comprehensive simulation experiments. 

\section{Latent Factor Mixture}
\label{sec:lamb}
To overcome the problems discussed in Section~\ref{sec:pitfall}, we propose a general class of latent factor mixture models defined as
\vskip -4ex
\begin{equation}
y_{i}  \sim  f(y_i; \eta_i, \psi), \quad
\eta_{i} \sim  
\sum_{h=1}^{\infty} \pi_h \mathcal{K}( \eta_i; \theta_h), \label{eq:latent}
\end{equation}
where $\eta_i = (\eta_{i1},\ldots,\eta_{id})\trans$ are $d$-dimensional latent variables, $d< n$ is fixed and not growing with $p$, $f(\cdot; \eta_i, \psi)$ is the density of the observed data conditional on the latent variables and measurement parameters $\psi$ and $\mathcal{K}(\cdot; \theta)$ is a $d$-dimensional kernel density. 

Under \eqref{eq:latent}, the high dimensional data being collected are assumed to provide error-prone measurements of an unobserved lower-dimensional set of latent variables $\eta_{i}$ on subject $i$.  
As a canonical example, we focus on a linear Gaussian measurement model with a mixture of Gaussians for the latent factors: 
\vskip -4ex
\begin{equation}
y_{i} \sim \mn_{p}( \Lambda \eta_{i}, \Sigma),\quad
\eta_{i} \sim \sum_{h=1}^\infty \pi_{h} \mn_{d}(\mu_{h}, \Delta_{h}), \quad \{\pi_{h}\}\sim Q_{0},
\label{eq:milf}
\end{equation}
\vskip -1ex \noindent
where $\Sigma=\diag(\sigma_{1}^{2},\dots,\sigma_{p}^{2})$ is a $p\times p$  diagonal matrix, 
and $\Lambda$ is a $p \times d$ matrix of factor loadings. 
The key idea is to incorporate all the cluster-specific parameters at the latent data level instead of the observed data level to favor parsimony.  The latent variables are supported on a lower-dimensional hyperplane, and we map from this hyperplane to the observed data level through multiplication by a factor loadings matrix and then adding Gaussian noise. We could further simplify the model by assuming $\Sigma =\sigma^{2} I_p$ instead of $\Sigma$ diagonal;\label{pg:diagsigma} we find it appealing to allow the different $y_{ij}$s to have varying measurement error variances and hence focus mainly on the unconstrained diagonal case. We refer to model \eqref{eq:milf} as a LAtent Mixture for Bayesian (Lamb) clustering.
The model is highly flexible at the latent variable level, allowing differences across clusters in the mean through $\mu_h$ and the shape, size, and orientation through $\Delta_{h}$.

With different motivations, \citet{galimberti2009, baek2010, montanari2010}  proposed similar latent factor mixture models as \eqref{eq:milf} albeit with additional constraints. Moreover, they fixed the number of clusters, used EM algorithms for model fitting  and assessed goodness-of-fit via information criteria.\label{pg:refs}

The proposed Lamb model has fundamentally different implications from the popular mixture of factor analyzers of \citet{ghahramani1996algorithm}, which defines a mixture of multivariate Gaussians at the $p$-dimensional observed data level having cluster-specific means and covariance matrices, with the dimension of the covariances reduced via a factor model.  
In contrast, we are effectively learning a common affine space within which we can define a simple location-scale mixture of  Gaussians. 
Our approach not only massively reduces the effective number of parameters for large $p$, but also provides a successful compromise between the two extreme cases of Section \ref{sec:pitfall}.\label{pg:implication}

\subsection{Prior Specifications}
\label{subsec:prior_specifications}
In order to accommodate very high-dimensional data, with $p \gg n$, it is important to reduce the effective number of parameters in the $p \times d$ loadings matrix $\Lambda$.  There is a rich literature on sparse factor modeling using a variety of shrinkage or sparsity priors for $\Lambda$; for example, refer to \citet{bhattacharya2011sparse}
and the references therein. 
Although a wide variety of shrinkage priors for $\Lambda$ are appropriate, we focus on a Dirichlet-Laplace  prior \citep{dir_laplace}, as it is convenient both computationally and theoretically. 
On a $p$-dimensional vector $\theta$, the Dirichlet-Laplace prior with parameter $a$, denoted by $\DL(a)$, can be specified in the following hierarchical manner
\vskip -3ex
\begin{equation}
\theta_{j} \mid \phi,\tau \simind \mn(0,\psi_{j}\phi_{j}^{2}\tau^{2}),~~\psi_{j} \simiid \Exp(1/2),~~ \phi\sim \Dir(a,\ldots,a),~~ \tau \sim \Ga(pa,1/2), 
\label{eq:dir_laplace_prior}
\end{equation}
\vskip -1ex \noindent
where $\theta_{j}$ is the $j$-th element of $\theta$, $\phi$ is a vector of the same length as $\theta$, $\Exp(a)$ is an
exponential distribution with mean $1/a$, $\Dir(a_{1}, \dots, a_{p} )$ is the $p$-dimensional Dirichlet distribution, and $\Ga(a,b)$ is the gamma distribution with mean $a/b$ and variance $a/b^{2}$.
To impose shrinkage uniformly on its elements a priori, we let $\mbox{vec}(\Lambda)\sim \DL(a)$ where $\mbox{vec}(\Lambda)$ denotes the vectorization of $\Lambda$. \label{pg:vect}
We then choose inverse-gamma priors for the residual variances, 	
$\sigma_{j}^{-2} \simiid \Ga(a_{\sigma},b_{\sigma})$.

For the prior $Q_{0}$ on the cluster weights $\{\pi_{h}\}$, for convenience in computation, we use a stick-breaking prior \citep{ishwaran2001gibbs} derived from a Dirichlet process, which has concentration parameter $\alpha$ impacting the induced prior on the number of clusters.  
To allow greater data adaptivity, we choose a $\Ga(a_{\alpha},b_{\alpha})$ prior for $\alpha$.  
We assign the cluster-specific means and covariances $\{ \mu_h, \Delta_h \}$ independent multivariate normal inverse-Wishart priors with location $\mu_{0}$, precision parameter $\kappa_{0}$, inverse scale matrix $\Delta_{0}$ and degrees of freedom $\nu_{0}$.
Our hierarchical Bayesian model for the $\eta_i$s can be equivalently represented as 
\vskip -2.5ex
\begin{equation}
\eta_{i}\mid \mu_{i},\Delta_{i} \simind \mn_{d} (\mu_{i},\Delta_{i}), \quad \mu_{i},\Delta_{i}\mid G\simiid G, \quad G\sim \DP(\alpha,G_{0}), \quad \alpha \sim \Ga(a_{\alpha},b_{\alpha}),\label{eq:usage_model}
\end{equation}
\vskip -.75ex \noindent
where $G_{0}=\NIW(\mu_{0},\Delta_{0},\kappa_{0},\nu_{0})$.
The gamma prior on the concentration parameter $\alpha$ is commonly adopted in many applications motivated by \citet{escobar1995bayesian}. The role of this hyperprior and the elicitation of its hyperparameters has been carefully studied by \citet{fruhwirth2019here}, and \citet{ascolani2022clustering}
recently showed the prior to have a crucial impact on consistency in estimating the number of clusters.

In practice, the latent variable dimension $d$ is unknown. 
Potentially we could put a prior on $d$ and implement a reversible-jump type \citep{rjmcmc} Markov chain Monte Carlo (MCMC) algorithm, which may lead to inefficient and expensive computation. 
Instead we adopt a principal component analysis (PCA) based empirical Bayes type approach \citep{bai2008factorest} to set $d$ to a large value learned from the data and let the prior shrink the extra columns on $\Lambda$.\label{pg:est_d}
We use the augmented implicitly restarted Lanczos bidiagonalization algorithm \citep{irlba} to obtain approximate singular values and eigenvectors, and choose the smallest $\wh{d}$ 
explaining at least $95\%$ of the variability in the data.  
This strategy substantially simplifies the computation.
The left and right singular values are used to initialize the $\Lambda$ and $\eta_{i}$'s in our MCMC implementation.
We initialize our cluster membership indicators using $k$-means.

For all the simulation experiments of the next section and the application, we choose $\mu_{0}=0$ and $\Delta_{0}=\xi^{2} I_{d}$ for a scalar $\xi^{2}>0$.
To specify weakly informative priors,
we set $\xi^{2}=20$, $\kappa_{0}=0.001$, $\nu_{0}=\wh{d}+50$, $a_{\alpha}=b_{\alpha}=0.1$ as the hyper-parameters of the DP mixture prior;
$a_{\sigma}=1$, $b_{\sigma}=0.3$ as the hyper-parameters of the prior on the residual variances.
We set $a=0.5$ as the Dirichlet-Laplace parameter following the recommendation of \citet{dir_laplace}.

\subsection{Posterior Sampling}
\label{sec:posteriorcomputation}

For posterior computation we {use a} Gibbs sampler defined by the following steps.
\begin{description}[leftmargin=0in]
\item[Step 1] Letting $\lambda_j\trans$ denote the $j$th row of $\Lambda$, $\eta=[\eta_{1},\dots,\eta_{n}]\trans$, $D_j = \tau^2\mbox{diag}(\psi_{j1}\phi_{j1}^2, \dots, \psi_{jd}\phi_{jd}^2)$ and $y^{(j)} = (y_{1j},\dots,y_{nj})\trans$, for $j=1, \dots, p$ sample 
\vskip-2ex
\begin{equation*}
(\lambda_j \mid -)  \sim \mn_{d}\left\{ (D_j^{-1} + \sigma_j^{-2} \eta\trans\eta)^{-1}\eta\trans \sigma_j^{-2}y^{(j)}, \,\,(D_j^{-1} + \sigma_j^{-2} \eta\trans\eta)^{-1} \right\}.
\end{equation*}

\item[Step 2] Update the 
$\Delta_{h}$'s from  the inverse-Wishart distributions 
$\mathrm{IW}\left( \wh{\psi}_{h}, \wh{\nu}_{h} \right)$ where 
\vskip-2ex
\begin{equation*}
\bar{\eta}_h=\textstyle{\frac{1}{n_h}\sum_{i:c_i=h}\eta_{i}}, \quad \wh{\nu}_{h}=\nu_{0}+n_{h}, 
\end{equation*}
\vskip-2ex
\begin{equation*}
\textstyle{\wh{ \psi}_{h}=\xi^2 I_{d}+\sum_{i:c_{i}=h}(\eta_{i}-\bar{\eta}_{h})(\eta_{i}-\bar{\eta}_{h})\trans+ \frac{\kappa_{0} n_{h}}{\kappa_{0} +n_{h}} \bar{\eta}_{h} \bar{\eta}_{h}\trans.}
\end{equation*}
\vskip-1.5ex			
Due to conjugacy, the location parameters $\mu_{h}$'s can be integrated out of the model.

\item[Step 3] Sample the latent factors, for $i = 1,\dots, n$, from 
\begin{equation*}
(\eta_i \mid -) \sim \mn_{d}\left\{\Omega_h\rho_h,\Omega_h +\Omega_h (\wh{\kappa}_{h,-i}\Delta_h)^{-1}\Omega_h \right\},
\end{equation*}
where 
$ n_{h,-i}=\sum_{j\neq i}\mathbbm{1} (c_{j}=h)$, $\wh{\kappa}_{h,-i}=\kappa_{0}+n_{h,-i}$,
$\bar{\eta}_{h,-i}=\frac{1}{n_{h,-i}}\sum_{j:c_j=h,j\neq i}\eta_{i}$, $\wh{\mu}_{h,-i}=\frac{n_{h,-i}\bar{\eta}_{h,-i} }{n_{h,-i}+\kappa_{0}}$,
$\rho_{h} = \Lambda\trans \Sigma^{-1} Y_i + \Delta_h^{-1} \wh{\mu}_{h,-i}$ and $\Omega_h^{-1} = \Lambda\trans \Sigma^{-1} \Lambda + \Delta_h^{-1}.$

\item[Step 4] Sample the cluster indicator variables  $c_1,\dots,c_n$ with  probabilities
\begin{equation}
\Pi(c_{i}=h\mid -)\propto \begin{cases}
n_{h,-i}\int \mn_d(\eta_i;  \mu_{h}, \Delta_h ) \de\Pi(\mu_h, \Delta_h\mid c_{-i}, \eta_{-i}) \text{ for } h\in c_{-i},\\
\alpha \int \mn_d(\eta_i;  \mu_h, \Delta_h ) \de\Pi(\mu_h, \Delta_h ) \text{ for }h \notin c_{-i}.
\end{cases}\label{eq:gibbs_eta}
\end{equation}
where $\eta_{-i}=\{\eta_{j}:j \neq i\}$ and $c_{-i}=\{c_{j}:j \neq i\}$. 
Due to conjugacy the above integrals are analytically available.

\item[Step 5] Let $r$ be the number of unique $c_{i}$'s.
Following \citet{west1992hyperparameter}, 
first generate $\varphi\sim \Beta(\alpha+1,n)$, 
evaluate $\pi/(1-\pi)=(a_{\alpha}+r-1)/\left\{n(b_{\alpha}-\log\varphi)\right\}$ and generate
\begin{equation*}
\alpha \mid \varphi, r \sim \begin{cases}
\Ga(\alpha+r,b_{\alpha}-\log\varphi)\text{ with probability }\pi, \\
\Ga (\alpha+r-1,b_{\alpha}-\log\varphi )\text{ with probability }1-\pi.
\end{cases}
\end{equation*}

\item[Step 6] For $j=1, \dots, p$ sample $\sigma_j^2$ from $\Ga \left\{a_{\sigma} + n/2, b_{\sigma} +  \sum_{i=1}^n (y_{ij} - \lambda\trans_j \eta_{i})^{2}/2\right\}$.	

\item[Step 7] 
Update the hyper-parameters of the Dirichlet-Laplace prior  through: 
\begin{enumerate}[label=(\roman*)]
\item For $j=1, \dots, p$ and $h=1,\dots d$ sample $\wt{\psi}_{jh}$ independently from an inverse-Gaussian $\mbox{iG}(\tau\phi_{jh}/\abs{\lambda_{jh}},1)$ distribution  and set $\psi_{jh}=1/\wt{\psi}_{jh}$.
\item Sample the full conditional posterior distribution of $\tau$  from a generalized inverse Gaussian $\mbox{giG}\{dp(1-a),1, 2\sum_{j,h}\abs{\lambda_{jh} }/\phi_{jh} \}$ distribution.
\item To sample $\phi \mid \Lambda$, draw $T_{jh}$ independently with $T_{jh} \sim \mbox{giG}(a - 1, 1, 2\abs{\lambda_{jh}} )$ and set $\phi_{jh}  = T_{jh}/T$ with $T=\sum_{jh} T_{jh}$.		
\end{enumerate}
\end{description}

This simple Gibbs sampler sometimes gets stuck in local modes;  a key bottleneck is the exploration Step 4.  
Therefore, we adopt the split-merge MCMC procedure proposed by 
\citet{split_merge}; the authors note that the Gibbs sampler is useful in moving singleton samples between clusters while the split-merge algorithm makes major changes. 
Hence, we randomly switch between Gibbs and split-merge updates. 
The split-merge algorithm makes smart proposals by performing restricted Gibbs scans of the same form as in \eqref{eq:gibbs_eta}.

From the posterior samples of $c_{i}$'s, we compute summaries following  \citet{wade2018bayesian}. 	
Our point estimate is the partition visited by the MCMC sampler that minimizes the posterior expectation of the Binder loss \citep{binderloss} exploiting the posterior similarity matrix obtained from the different sampled partitions.

The sampling algorithm can be easily modified for other priors on $\Lambda$ having a conditionally Gaussian representation, with Step 7 modified accordingly.  For example, we could use horseshoe \citep{carvalho2009handling}, increasing shrinkage priors \citep{bhattacharya2011sparse,legramanti2019bayesian,schiavon}, or the fast factor analysis prior \citep{rovckova2016fast}.
Similarly, alternative priors for $\{ \pi_h \}$, such as  \citet{pitman1997two} or  \citet{miller2018}, can be adopted with minor modifications in Steps 4 and 5.

\section{Properties of the Latent Mixture for Bayesian Clustering Method}
\label{sec:properties}
\subsection{Bayes Oracle Clustering Rule}
\label{subsec:properties}
{We first define a Bayes oracle clustering rule where the observed data follow the distribution in model \eqref{eq:latent}, that is, the high dimensional $y_{i}$'s provide error-prone measurements on unobserved lower-dimensional latent variables $\eta_{i}$'s on subject $i$, and
we assume the oracle has knowledge of the exact values of the latent variables $\{\eta_{0i}\}$, where $\eta_{0i}$'s are $d$-dimensional latent vectors.
Given this knowledge, the oracle can define any Bayesian  mixture model to induce a posterior clustering of the data, which is not affected by the high-dimensionality of the problem.  
This leads to the distribution over the space of partitions in the following definition.

\begin{definition}
\label{def:oracle}
Let ${\boldeta}_0 = \{\eta_{01},\ldots,\eta_{0n}\}$ be the true values of the unobserved latent variables corresponding to each data point.
The following mixture model is assumed to cluster $\eta_{0}$
\begin{equation*}
\eta_{0i}\sim {\sum_{h=1}^{\infty} \pi_h \mathcal{K}( \eta_{0i}; \theta_h), \quad \{\pi_{h}\}\sim Q_{0}, \quad \theta_{h}\simiid G_{0}}.
\end{equation*}
Then the oracle probability of clustering is defined as
\begin{equation}
\Pi(\Psi \mid  {\boldeta}_{0} ) = 
\frac{{\Pi}(\Psi) \times 
	\int\prod_{h\geq 1}  \prod_{i:c_{i}=h} \K{\eta_{0i}}{\theta_h} \de G_{0}(\theta_{h})}{
	\sum_{\Psi' \in {\mathscr P}} {\Pi}(\Psi') 
	\times \int \prod_{h\geq 1} \prod_{i:c_{i}'=h} \K{\eta_{0i}}{\theta_{h}} \de G_{0}(\theta_{h})}.  \label{eq:marginalposterior_eta}
\end{equation}
\end{definition}

Probability \eqref{eq:marginalposterior_eta} expresses the oracles' uncertainty in clustering if the clustering model could have been applied on the true latent factors.  This is a gold standard in being free of the curse of dimensionality through using the oracles' knowledge of the true  latent variables, but we make no claims about the relationship between the oracle posterior and any `true' clustering.
Under the framework of Section \ref{sec:lamb}, the high-dimensional measurements on each subject provide information on these latent variables, with the clustering done on the latent variable level.  
Ideally, we would get closer to the oracle partition probability under the proposed method as $p$ increases, turning the curse of dimensionality into a blessing.  We show that this is indeed the case in Section \ref{subsec:suff_cond}.

{To this end,} we assume the oracle uses a location mixture of Gaussians with a common covariance matrix. 
We assume the following mixture distribution on $\eta_{0i}$'s, independent non-informative Jeffreys prior for the common covariance and arbitrary
prior $Q_{0}$ on the mixture probabilities:
\begin{equation}
\eta_{i} \simiid \sum_{h=1}^\infty \pi_h \mn_d(\mu_h, \Delta),\quad
\mu_{h}\mid \Delta\simiid \mn_d(0,\kappa_0^{-1}\Delta),\quad
\Delta \propto \abs{\Delta}^{-\frac{d+1}{2}},\quad
\{\pi_h \}\sim Q_{0}.
\label{eq:prior_on_eta}
\end{equation}	
For $d < n$, the oracle rule is well defined for the Jeffreys prior on $\Delta$.
Note that the marginal Jeffrey's prior is free of any hyperparameter. 

\label{ref:loc_scale_invariance}


\subsection{Assumptions on Data and Prior Specifications}
\label{subsec:assumptions}
In this section, we show that the posterior probability on the space of partitions induced by the proposed model converges to the oracle probability as $p\to \infty$ in expectation under appropriate conditions on the data generating process and the prior.   
We assume that the residual error variances $\sigma_{j}^{2}$'s are the same having true common value $\sigma_{0}^{2}$ for all $j=1,\dots,p$.
Our result is based on the following assumptions on $\p$, the true data-generating distribution of $y_{1},\dots, y_{n}$:
\begin{enumerate}[label={(C\arabic*)}]
\item\label{ass0} $y_{i}\simind \mn_{p}(\Lambda_{0}\eta_{0i}, \sigma_{0}^{2} I_{p})$, for each  $i=1,\ldots,n$; 
\item \label{ass1}\label{ass2}{$\lim_{p\rightarrow\infty} \norm{\frac{1}{p}\Lambda_{0}\trans \Lambda_{0}-M}_{2}=0 $ where $M$ is a $d\times d$ positive-definite matrix}; 
\item \label{ass3} $\sigma^{2}_L<\sigma^{2}_0<\sigma^{2}_U$ where $\sigma^{2}_L$ and $\sigma^{2}_U$ are known constants;
\item \label{ass4} $\norm{\eta_{0i}}=O(1)$ for each $i=1,\ldots,n$.
\end{enumerate}
Condition \ref{ass0} corresponds to the conditional likelihood of $y_{i}$ given $\eta_{i}$ being correctly specified and the data containing increasing information on the latent factors as $p$ increases. This increasing information assumption is extremely mild; indeed, each individual $y_{ij}$ can be very noisy and provide minimal information about $\eta_i$ and there will still be a build up of information across $j=1,\ldots,p$ as long as the additional variables are not completely uncorrelated with the target latent factors. In fact, we have a build up of information even when a proportion of the factor loadings are exactly zero, the factor loadings are very small relative to the residual variance, and the residuals are heavy-tailed.
We illustrate this empirically with a simple simulation study in Section \ref{sm subsec:degen_clustering} of the supplementary materials.
Condition \ref{ass1} ensures that $\Lambda_{0}$ is not \textit{ill-conditioned} and its spectral norm does not increase too fast with respect to $p$ since the highest and lowest eigenvalues of $\Lambda_{0}\trans\Lambda_{0}$ grow in $O(p)$.
\label{pg:ass_C2}
Related but much stronger conditions appear in the factor modeling \citep{FAN_factor,fan_factor2} and massive covariance estimation literature \citep{pati2014}. 
We allow the columns of $\Lambda_{0}$ to be non-orthogonal with varying average squared values which is expected in high-dimensional studies.  Condition \ref{ass3} bounds the variance of the observed $y_{i}$s and  \ref{ass4} is a weak assumption ensuring that the latent variables do not depend on $n$ or $p$. 
Additionally, we assume that the latent dimension $d$ is known.

Although we use a stick-breaking prior on the mixture probabilities $\{\pi_{h}\}$ in Section \ref{subsec:prior_specifications},
we derive our results for an arbitrary prior $Q_{0}$ for wider applicability. We assume the inverse-gamma prior on residual variance $\sigma^{2}$ to be restricted to the compact set $[\sigma_{L}^{2},\sigma_{U}^{2}]$.

\subsection{Main Results}
\label{subsec:suff_cond}
In Lemma \ref{lemma:suff_cond} we derive sufficient conditions for the posterior probability on the space of partitions to converge to the oracle probability for $p\to \infty$. 
\begin{lemma}
\label{lemma:suff_cond}
Let ${ \boldeta}=[\eta_{1},\dots, \eta_{n} ]\trans$,
${ \zeta}^{(p)}=\left[\zeta^{(p)}_{1},\dots,\zeta^{(p)}_{n}\right]\trans=(\sqrt{p\log p})^{-1}(\Lambda\trans\Lambda)^{\half}{\boldeta}$ and, 	for any $\delta>0$,
$B_{p,\delta}=\bigcap_{i=1}^{n} \{\Lambda, \eta_i: (\sqrt{p\log p})^{-1}  \norm{ \Lambda\eta_i-\Lambda_0\eta_{0i}}<\delta \}$.
Assume  for any $\delta>0$ 
\begin{equation}
\Pi (\bar{B}_{p,\delta} \mid  \by)\rightarrow 0 \quad \p\text{-a.s.} \label{eq:suff_eq}
\end{equation}
where $\bar{B}_{p,\delta}$ is the complement of ${B}_{p,\delta}$.
Let $E(\cdot\mid\by)$ denote expectation with respect to the posterior distribution of the parameters given data $\by$ and $\Pi(\Psi \mid { \zeta^{(p)}})$ be the conditional probability of partition $\Psi$ with ${\eta}_0$ replaced by ${\zeta}^{(p)}$ in \eqref{eq:marginalposterior_eta}.
Then, 
$
\lim_{p\rightarrow\infty} E\left\{\Pi(\Psi \mid { \zeta^{(p)}}) \mid \by \right\} =  \Pi(\Psi \mid {{\eta}_{0}})
$. 
\end{lemma}

In the following theorem, we show that condition \eqref{eq:suff_eq} holds for \Lamb and hence we avoid the large $p$ pitfall.
The proof is in the supplementary materials.
\begin{theorem}
\label{th:main_result}
Let $B_{p,\delta}$ be as defined in Lemma \ref{lemma:suff_cond} and $\bar{B}_{p,\delta}$ be its complement set.
Then, under \ref{ass0}-\ref{ass4} and  model \eqref{eq:milf}, $\Pi (\bar{B}_{p,\delta} \mid  \by )\rightarrow 0$ $ \p$-a.s. for any $\delta>0$.
\end{theorem}

Theorem \ref{th:main_result} implies that our model learns the latent factors more accurately with increasing $p$.
In addition to the proof of Theorem \ref{th:main_result}, this result is further illustrated empirically  via a simple simulation experiment reported in Section \ref{subsec:simulation_latent} of the supplementary materials.


The oracle has a slightly simpler model specification than \eqref{eq:usage_model} assuming  common covariances across components. This simplification is done to make the associated theory more tractable, but the simplified location mixture case is rich enough to provide a nice test case for assessing how the proposed approach can escape the curse of dimensionality.

\label{ref:diff_limit}

As conditions \ref{ass0}-\ref{ass4} imply \ref{cond_neg},
the clustering models in Corollaries \ref{cor:neg_result2} and \ref{cor:neg_result} would still lead to the two extreme partitions. 
The Lamb model, in learning the low-dimensional latent space with increasing dimensions, escapes these pitfalls.



\section{Simulation Study}
\label{sec:simstudy}

\label{sec:simulation}
We perform a simulation study to analyze the performance of \Lamb in clustering high dimensional data. 
The sampler introduced in Section \ref{sec:posteriorcomputation} is available from the GitHub page of the first author.  
We compare with a Dirichlet process mixture of Gaussian model with diagonal covariance matrix  implemented in \texttt{R} package \texttt{BNPmix} \citep{bnpmix}\label{pg:failed_other_packages}, a nonparametric mixture of infinite factor analyzers implemented in \texttt{R} package \texttt{IMIFA} \citep{imifa}, and a pragmatic two-stage approach (PCA-KM) that performs an approximate sparse  principal component analysis of the high dimensional data to reduce dimensionality from $p$ to $\wh{d}$---with $\wh{d}$ the minimum number of components explaining at least $95\%$ of the variability as discussed in Section \ref{subsec:prior_specifications}---and then applies $k$-means on the principal components, with $k$ chosen by maximizing the average silhouette width \citep{rousseeuw1987silhouettes}. This same approach is used to choose $\wh{d}$ in implementing Lamb.

\begin{figure}[H]
	\centering
	\includegraphics[width=.95\textwidth]{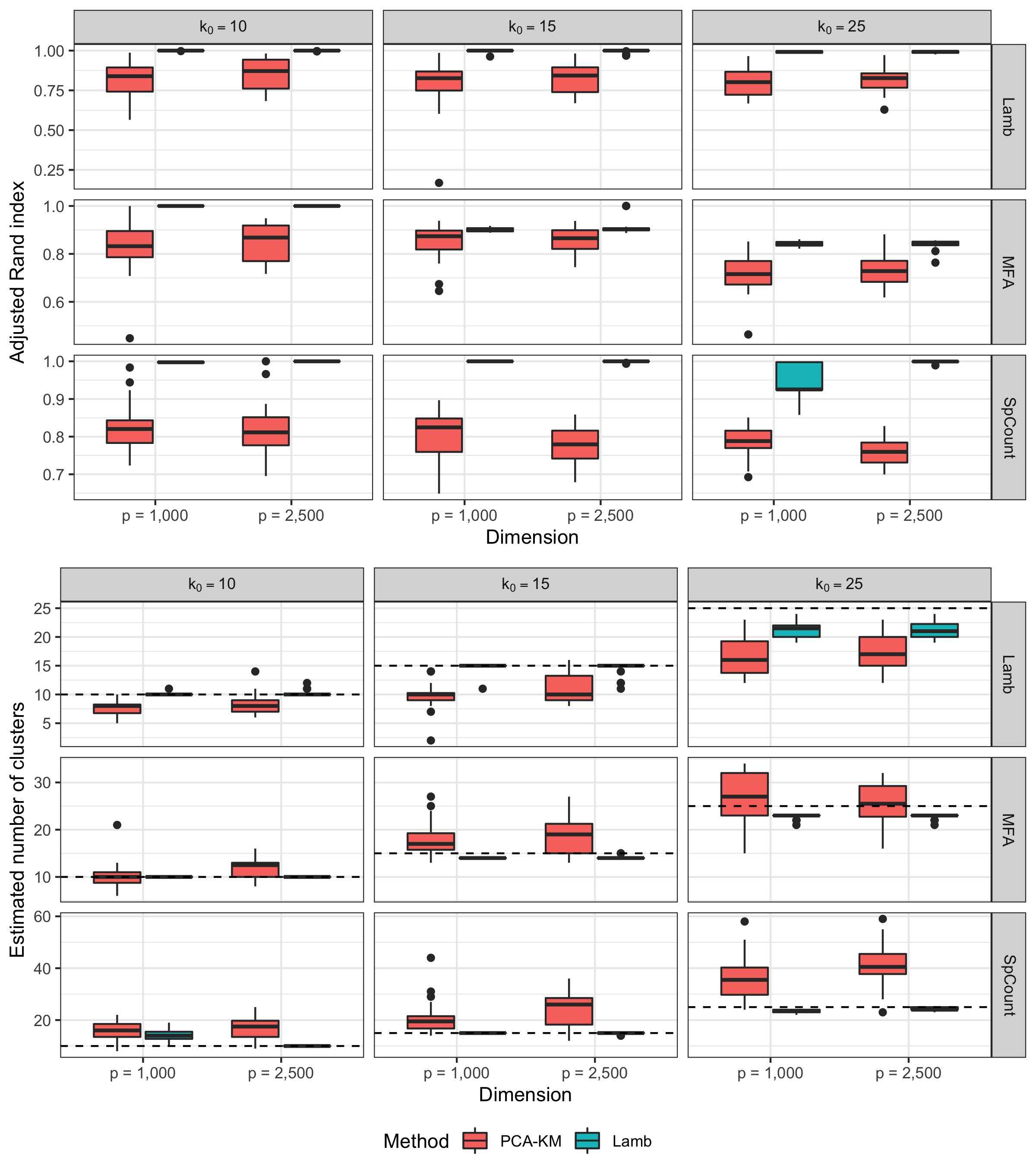}
	\vskip-3ex
	\caption{Comparison between our proposed Lamb and the two-stage PCA-KM approach:
		Distributions of the adjusted Rand indices (upper plot) and estimated number of clusters (lower plot) in 20 replicated experiments. 
		Horizontal dashed lines denote the true number of clusters. 
		The simulation scenarios, reported in each row, are labeled as Lamb for the model of Section \ref{sec:lamb}, MFA for mixture of factor analyzers and SpCount for the $\log$ transformed zero inflated sparse Poisson counts.}
	\label{fig:comparison}
	\vskip-4ex
\end{figure}

For the high-dimensional simulation settings we considered, both the mixture of Gaussians and the mixture of factor analyzers  showed high instability,  including software crashing for memory issues, lack of convergence, and extremely long running times. For these reasons we report a comparison with PCA-KM approach only.  To test the accuracy of the estimated clustering relative to the true clustering, we compute the adjusted Rand index \citep{rand1971objective}.   

We generated data under: [1] Lamb, [2] mixture of sparse factor analyzers (MFA), and [3] mixture of log transformed zero inflated sparse Poisson counts (SpCount)
[1]-[2] have latent dimension 20, while for [3] the data are discrete and highly non-Gaussian within clusters mimicking the data of 
Section \ref{sec:realdata}.
Details  are provided in Section \ref{SM_sec:simulation_details} of the supplementary materials.  

We vary true number of clusters 
$k_{0} \in \{10, 15, 25\}$, 
with the first $\lfloor2k_{0}/3\rfloor$ `main' clusters having the same probability and the remaining ones having together the same probability of a single main cluster.
For example if $k_{0} = 25$, we set 16 main clusters with probability $1/17$ each and 9 minor clusters of equal weights, whose total probability sums to $1/17$.  
This is a highly challenging case, as many methods struggle unless there are a small number of close to equal weight clusters that are well separated. 
The dimension $p$ varies in $p=\{1{,}000,~ 2{,}500\}$  while the sample size $n$ is $n=2{,}000$. 
Data visualization plots using \citet{umap} are in Section \ref{subsec:simulations_sm} of the supplementary materials. For each configuration, we perform 20 independent replications.
We run our sampler for $6{,}000$ iterations discarding the first $1{,}000$ as burn in and taking one draw every five to reduce autocorrelation. 
Prior elicitation follows the default specification of Section \ref{subsec:prior_specifications}. 
On average, $6{,}000$ iterations under these settings took between 40 and 50 minutes on a iMac with 4.2 GHz Quad-Core Intel Core i7 processor and 32GB DDR4 RAM.

Figure \ref{fig:comparison} reports  the distribution of the 20 replicates of the adjusted Rand index and mean estimated number of clusters. 
Our proposed Lamb is uniformly superior in each scenario obtaining high adjusted Rand indices, accurate clustering results, and less variability across replicates.  
In the MFA scenario, Lamb yields relatively lower Rand index for $k_{0} = 25$. 
This is {not unusual} due to model misspecification and the large number of clusters.

The Lamb results do not vary much across the simulation replicates 
because the oracle posterior is quite concentrated at the true clustering.
Since the dimension $p$ is in the thousands, the asymptotic results derived in Section \ref{sec:properties} kicked in resulting in narrow posterior credible intervals. To understand the performance of our proposed method in smaller sample sizes, we include additional simulation results with $n=500$ in Section \ref{sm subsec:smallsamp} of the supplementary materials.

Furthermore, Section \ref{sm subsec:degen_clustering} in the supplementary materials reports two simple simulation experiments showing that the degenerate clustering behavior discussed in Section \ref{sec:pitfall} is evident even in moderate dimensions of $p=20$.

\section{Application to ScRNASeq Cell Line Dataset}
\label{sec:realdata}

In this section, we analyze the GSE81861 cell line dataset \citep{cellline} to illustrate the proposed method. The dataset profiles 630 cells from 7 cell lines using the Fluidigm based single cell RNA-seq protocol \citep{fluidigm}. 
The dataset includes 83 A549 cells, 65 H1437 cells, 55 HCT116 cells, 23 IMR90 cells, 96 K562 cells, 134 GM12878 cells, 
174 H1 cells 
and $57{,}241$ genes. 
The cell types are known and hence the data provide a useful benchmark to assess performance in clustering high-dimensional data.

Following standard practice in single cell data analysis, we apply data pre-processing.
Cells with low read counts are discarded, as we lack reliable gene expression measurements for these cells, and 
data are normalized following \citet{scran_normalization}.  
We remove non-informative genes using M3Drop \citep{m3drop}. 
After this pre-processing phase, we obtain a final dataset with $n=531$ cells and $p=7{,}666$ genes. 

Applying our empirical Bayes approach, we estimate the latent dimension as $\wh{d}=19$.\label{pg:rna_d}
We implement \Lamb using our default prior, collecting $10,000$ iterations after a burn-in of $5,000$ and keeping one draw in five.
As comparison, we apply the two stage procedure of the previous section and the popular Seurat \citep{seurat2018} pipeline which performs  quality  control, normalization, and selects informative genes that exhibit high variation across the cells. \label{pg:highly_variable}

Graphical representations of the different clustering results are shown in Figure \ref{fig:realdata} via UMAP projections \citep{umap}.
Our proposed Lamb, the two stage approach, and Seurat achieve adjusted Rand indices of 0.977, 0.734 and 0.805 when compared to the true cluster-configuration and yield 12, 10, and 8 clusters, respectively. 
Seurat is reasonably accurate but splits the H1 cell-type into two clusters, while the two-stage approach is dramatically worse.

\begin{figure}[ht]
\centering
\includegraphics[width=\linewidth]{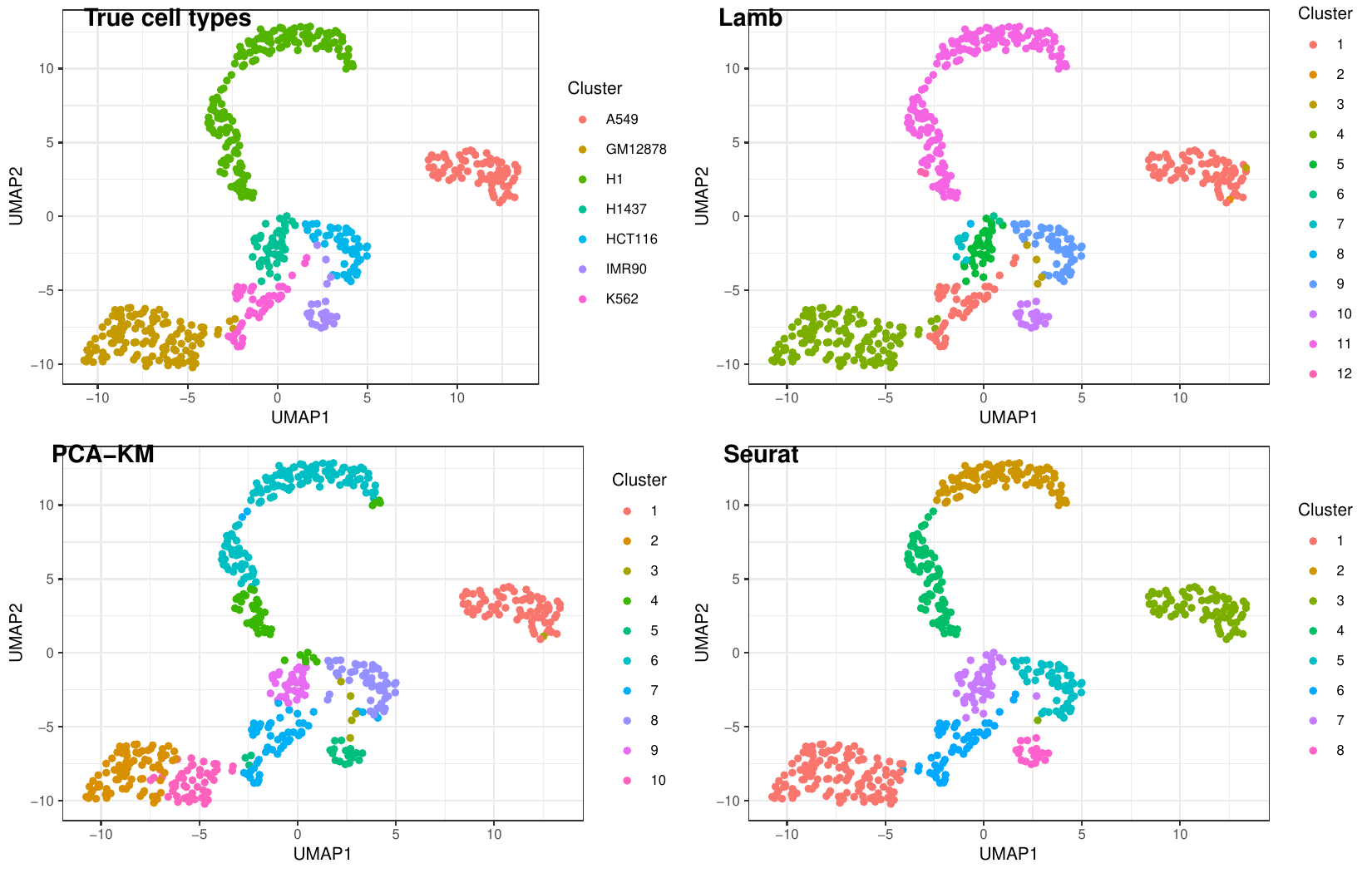}
\vskip -3ex
\caption{UMAP plots of the cell line dataset: Clusterings corresponding to the true cell-types, Lamb estimate, PCA-KM estimate and Seurat estimate are plotted in clockwise manner.
Different panels use different color legends.}
\label{fig:realdata}
\end{figure}

\begin{figure}[ht]
\centering
\includegraphics[trim={0cm 1.5cm .9cm 0.1cm},clip,height=0.5\textwidth,valign=c]{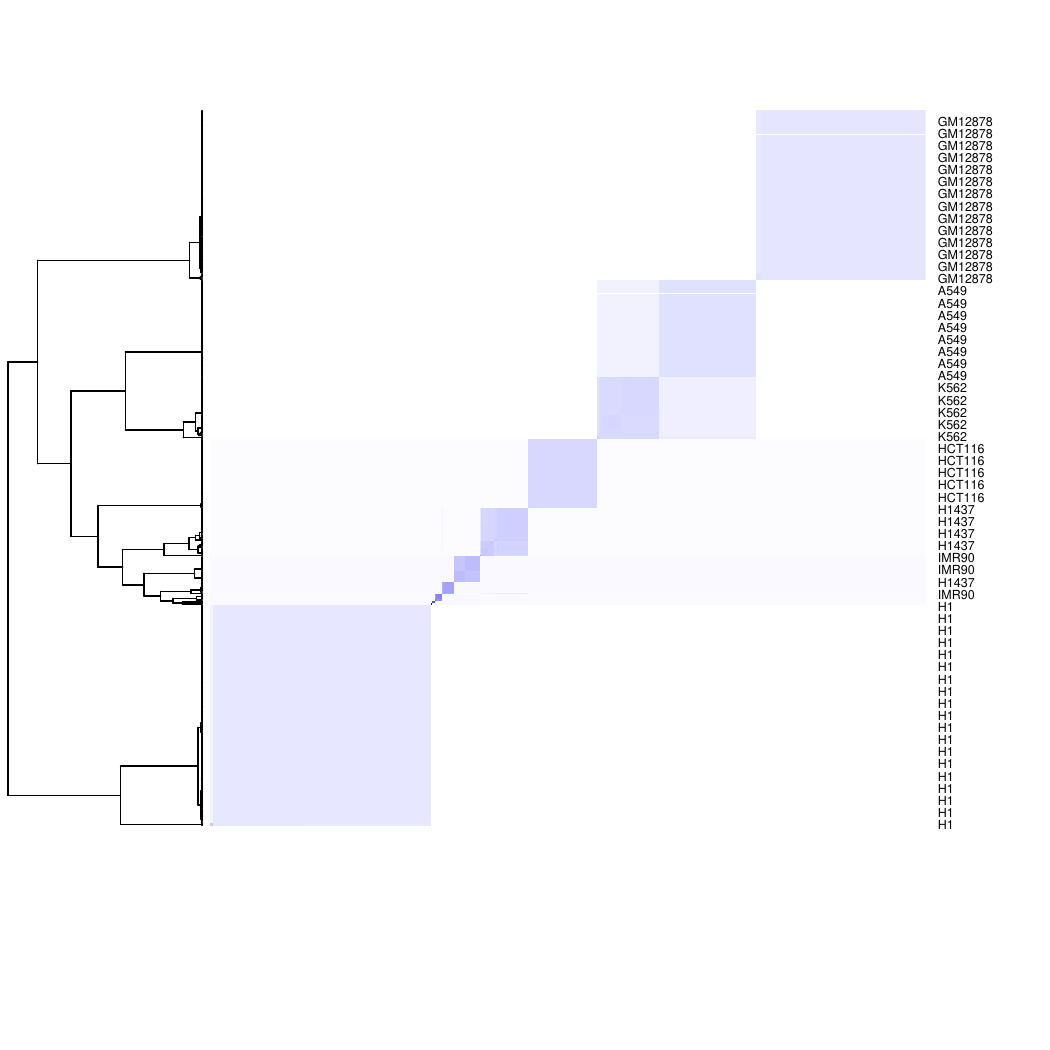}
\includegraphics[trim={0cm 1.5cm .9cm 0.1cm},clip,height=0.4\textwidth,valign=c]{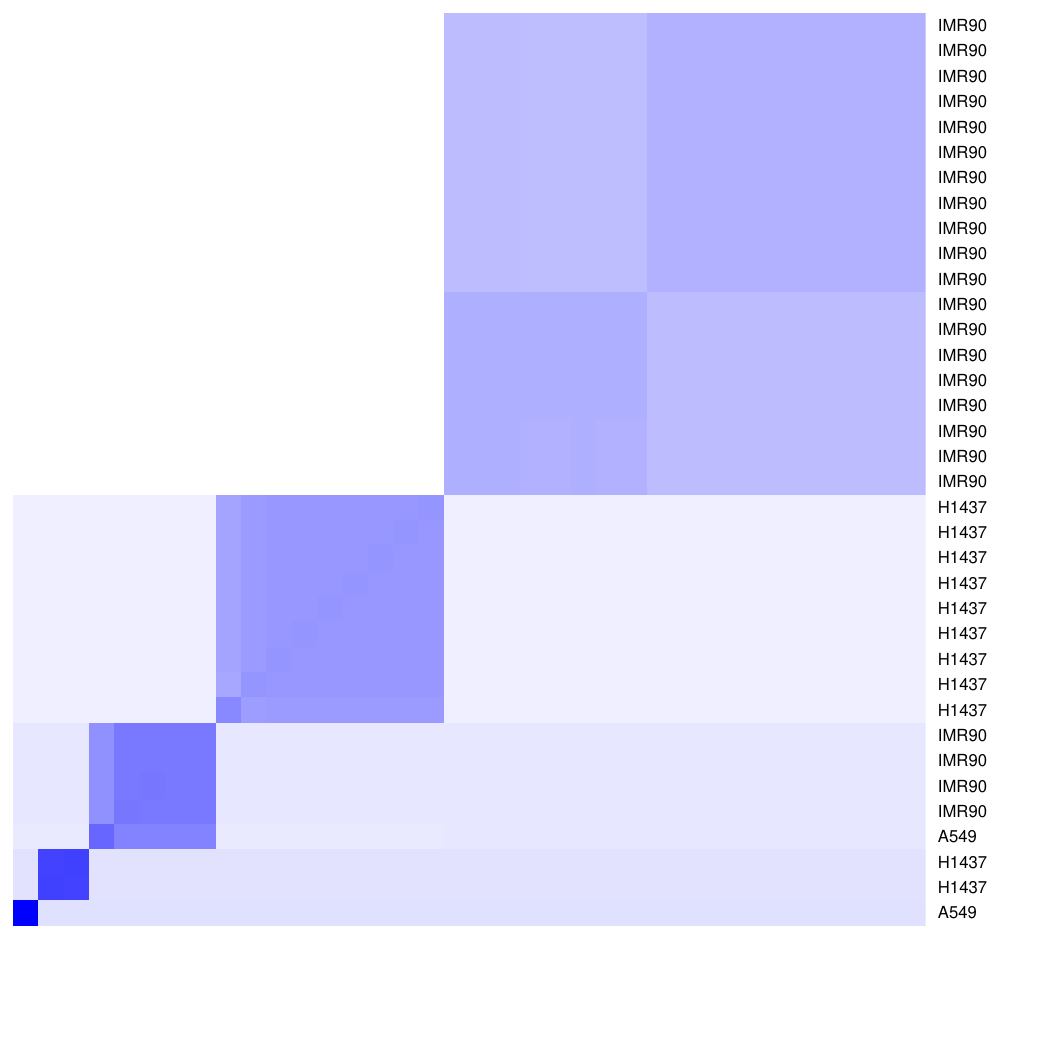}
\vskip -5ex
\caption{Posterior similarity matrix obtained from the Markov chain Monte Carlo samples of the Lamb method: Left panel reports the similarity matrix for the full cell line dataset along with the dendrogram obtained using complete linkage; row names report the true cluster names; right panel  zooms the center of the left panel.}
\label{fig:psmplot}	
\end{figure}

An appealing aspect of our approach is posterior uncertainty quantification. The 95\% credible interval for the adjusted Rand index is $[0.900, 0.985]$ and the posterior probability of having between 11 and 13 clusters is 0.98. This suggests that the posterior distribution is highly concentrated, which is consistent with our simulations.  
The posterior similarity matrix reported in the first panel of Figure \ref{fig:psmplot}---also reporting the related dendrogram obtained by using complete linkage---clearly shows that the majority of the observations have a high posterior probability of being assigned to a specific cluster and negligible probability of being assigned to artifactual clusters.\label{pg:artifactual}
Figure \ref{fig:psmplot} also shows micro clusters leading to over-estimation of the number of cell types.  
Two cells of cluster A549 are put in singleton clusters. 
Similarly cluster IMR90 is divided into two clusters of size 4 and 19 with negligible posterior probability of being merged.  
Finally cluster H1437 is split into four clusters with the main one comprising 35 of 47 observations and the smallest one comprising just one observation.  
Such micro-clusters have negligible impact for practical inference since Lamb does recover the original clustering configurations for most cell-types as reflected by the high adjusted Rand index with the true cell-types.
Single-cell experiments are subject to high technical noise 
\citep{singlecell_tech_noise}
which is not possible to completely remove in pre-processing steps.
Such noise can potentially induce differences between cells that may not have any biological significance,
for example, the cells in IMR90 (split into the clusters 3 and 10, see the top panel of Figure \ref{fig:realdata} for details)  exhibit a substantial amount of variability although they are biologically of the same type.

\section{Discussion}
\label{sec:discussion}

Part of the appeal of Bayesian methods is the intrinsic penalty for model complexity or `Bayesian Ockham razor'  \citep{ockham_razor}, 
which comes through integrating the likelihood over the prior in obtaining the marginal likelihood.  If one adds unnecessary parameters, then the likelihood is integrated over a larger region, which tends to reduce the marginal likelihood.   In clustering problems, one relies on the Bayesian Ockham razor to choose the appropriate compromise between the two extremes of too many clusters and over-fitting and too few clusters and under-fitting.  Often in low-dimensional problems, this razor is effective and one obtains a posterior providing a reasonable representation of uncertainty in clustering data into groups of relatively similar observations.  However, a key contribution of this article is showing that this is fundamentally not the case in high-dimensional problems, and one can obtain nonsensical results using seemingly reasonable priors.  

Perhaps our most interesting result is the degenerate behavior in the $p \to \infty$ case for the true posterior on clusterings, regardless of the true data generating model.
This negative result provided motivation for our latent factor mixture model, which addresses the large $p$ pitfall by clustering on the latent variable level.  
Using a low rank factorization with appropriate shrinkage priors, the method can also handle realistic high-dimensional problems.
Another interesting theoretical result is our notion of a Bayesian oracle for clustering; to our knowledge, there is not a similar concept in the literature.
We show that our proposed Lamb attains the oracle with increasing dimensions.

Several interesting projects stem from the proposed work, which is a first step towards addressing pitfalls of Bayesian approaches to high-dimensional clustering.
One important thread is designing faster MCMC algorithms for massive sample size exploiting parallel and distributed computing; for example, running MCMC for different subsets of the variables in parallel and combining the results.
Some recent works in the 
literature discuss related approaches \citep{yangni2020,song2020distributed} but without considering the pitfalls that arise in high-dimensional data clustering.
Another thread is to develop fast approximate inference algorithms that avoid MCMC, such as variational Bayes.  
In addition, it is of substantial interest to generalize the proposed approach to handle more complex data structures; for example, involving data that are not real-valued vectors and allowing for kernel misspecification \citep{miller2019robust}.
In our settings $d$ and $n$ are fixed and not growing with $p$.
The study of situations in which $p$, $d$ and $n$ jointly increase---at some rate---would be a very interesting theoretical extension of our results.

\section*{Supplementary Materials}
Proofs of additional theoretical results, simulation studies and MCMC convergence diagnostics are provided in the supplementary materials.

\section*{Appendix}

\setcounter{equation}{0}
\renewcommand{\theequation}{A.\arabic{equation}}
\renewcommand{\thelemma}{A.\arabic{lemma}}

\subsection*{Proofs of Section \ref{sec:pitfall}}

\begin{proof}[Theorem \ref{th:negative_result}]
Consider the ratio of posterior probabilities: 
\begin{equation} 
\frac{ \Pi(\Psi \mid \by) }{\Pi(\Psi'\mid \by)  }.
\label{eq:ratio}
\end{equation}
If this ratio converges to zero for all $c_1,\ldots,c_n$ in $\p$-probability as $p \to \infty$, then any partition nested into another partition is more likely \textsl{a posteriori} implying $\Pi(c_1=\cdots = c_n \mid \by) = 1$ in $\p$-probability so that all subjects are grouped in the same cluster with probability one.  Conversely if the ratio  converges to $+\infty$, then $\Pi(c_1\neq \cdots \neq c_n \mid \by) = 1$  in $\p$-probability and each subject is assigned to their own cluster with probability one.  

Without loss of generality, assume that $c_1, \dots, c_n$ define $k_n$ clusters of sizes $n_1, \dots, n_{k_n}$ and that  $c_i'=c_i$ for $c_i \in \{1,\ldots,{k_n}-2\}$ and $c_i'=k_n-1$ for $c_i \in \{k_n-1,k_n\}$, with $n_1', \dots, n'_{k_n'-1}$ the cluster sizes under the partition induced by the $c_i'$. In general, ratio \eqref{eq:ratio} can be expressed as
\begin{equation}
\frac{
\Pi (\Psi)}{
\Pi(\Psi') } \times 
\frac{ \prod_{h=1}^{k_n} \int \prod_{i:c_i=h}  \K{y_i}{\theta} \de P_0(\theta)}{
\prod_{h=1}^{k_n-1} \int \prod_{i:c'_i=h}  \K{y_i}{\theta} \de P_0(\theta) 
}.\label{eq:ratio2}
\end{equation}
The left hand side of \eqref{eq:ratio2} can be expressed as the ratio between the EPPFs. Since by assumption there is a positive prior probability for any partition in $\mathscr P$, this ratio is finite and does not depend on $p$ or the data generating distribution.
Thus, by induction and under the assumptions on the right factor of \eqref{eq:ratio2} we conclude the proof.
\end{proof}

\begin{proof}[Corollary \ref{cor:neg_result2}]
Define $c_1,\ldots,c_n$ and  $c'_1,\ldots,c'_n$ consistently with the proof of Theorem \ref{th:negative_result}. 
Then, consider the ratio of the marginal likelihoods
\begin{equation}
\frac{
\prod_{h=1}^{\kn} \int \prod_{i:c_i=h}  
\mn_p\left(y_i; \mu_{h} ,\Sigma_{h}\right) 
\mn_p(\mu_h; \mu_0, \kappa_0^{-1}\Sigma_{h}) 
IW(\Sigma_{h}; \nu_0, \Lambda_{0}) 
\, \de( \mu_{h},\Sigma_{h})} {
\prod_{h=1}^{\kn-1} \int \prod_{i:c'_i=h} 
\mn_p\left(y_i; \mu_{h} ,\Sigma_{h}\right) 
\mn_p(\mu_h; \mu_0, \kappa_0^{-1}\Sigma_{h}) 
IW(\Sigma_{h}; \nu_0, \Lambda_{0}) 
\, \de( \mu_{h},\Sigma_{h})
}.\label{eq:ratio2_neg}
\end{equation}	
The numerator of \eqref{eq:ratio2_neg} is 
$$\prod_{h=1}^{\kn} \left\{\frac{1}{\pi^{n_{h} p/2}} 
\frac{\Gamma_p(\frac{\nu_0 +n_{h}}{2})}{\Gamma_p(\frac{\nu_0}{2})} 
\left( \frac{\kappa_0}{\kappa_0+n_{h}}\right)^{\frac{p}{2}} \frac{\abs{\Lambda_{0}}^{\frac{\nu_0}{2}}} {\abs{\Lambda_0 + S_{h}^{\Psi} + \frac{n_{h}\kappa_{0}}{\kappa_0+n_{h}}(\bar{y}_{h}^{\Psi} -\mu_0)(\bar{y}_{h}^{\Psi} -\mu_0)\trans}^{\frac{\nu_0+n_{h}}{2}}} \right\},$$
with 
$\bar{y}_h^{\Psi}=n_{h}^{-1}\sum_{i:c_{i}=h}y_{i}$,
$S^{\Psi}_{h}=\sum_{i:c_i=h} \left(y_{i}-\bar{y}^{\Psi}_h\right)\left(y_{i}-\bar{y}^{\Psi}_h\right)\trans$, and $\Gamma_{p}(\cdot)$ being the multivariate gamma function.
Obtaining a corresponding expression for the denominator, the ratio  \eqref{eq:ratio2_neg} becomes
\begin{multline}
\frac{\Gamma_{p}\left(\frac{\nu_{0}+n_{\kn-1} }{2}\right) \Gamma_{p}\left(\frac{\nu_{0}+n_{\kn} }{2}\right)} {\Gamma_{p}\left(\frac{\nu_{0}+n'_{\kn-1} }{2}\right)\Gamma_{p}\left(\frac{\nu_{0} }{2}\right) }
\times \left\{ 
\frac{\kappa_{0} (\kappa_{0}+n'_{\kn-1})}{
(\kappa_{0}+n_{\kn-1})(  \kappa_{0}+n_{\kn} ) } \right\}^{p/2} \\
\times \frac{\abs{\Lambda_{0}}^{\frac{\nu_{0}}{2}} \abs{\Lambda_{0} + S_{\kn-1}^{\Psi'} + \frac{n'_{\kn-1}\kappa_{0}}{\kappa_0+n'_{\kn-1}}(\bar{y}_{\kn-1}^{\Psi'} -\mu_0)(\bar{y}_{\kn-1}^{\Psi'} -\mu_0)\trans}^{\frac{\nu_0+n'_{\kn-1}}{2}} } {\prod_{h=\kn-1}^{\kn}  \abs{\Lambda_0 + S_{h}^{\Psi} + \frac{n_{h}\kappa_{0}}{\kappa_0+n_{h}}(\bar{y}_{h}^{\Psi} -\mu_0)(\bar{y}_{h}^{\Psi} -\mu_0)\trans}^{\frac{\nu_0+n_{h}}{2}} }.
\label{eq:ratio111}
\end{multline}
We first study the limit of the  first factor of  \eqref{eq:ratio111}. 
From Lemma \ref{lemma:multgamma_ratio},
we have
\begin{equation*}
\lim_{p\to\infty}\frac{1}{p} \left\{ \log\frac{\Gamma_{p}\left(\frac{\nu_{0}+n_{\kn} }{2}\right) } {\Gamma_{p}\left(\frac{\nu_{0}+n'_{\kn-1} }{2}\right) }  +  \log \frac{\Gamma_{p}\left(\frac{\nu_{0}+n_{\kn-1} }{2}\right) }{ \Gamma_{p}\left(\frac{\nu_{0} }{2}\right) } \right\} =0.
\end{equation*}
We now study the limit of the remaining part of \eqref{eq:ratio111}.  Note that, if we replace each observation $y_{i}$ with $\wt{y}_{i} =\Lambda_{0}^{-\half} (y_{i}-\mu_{0} )$, assumption \ref{cond_neg} is still valid for $\wt{y}_{i}$'s. 	
Moreover, $\abs{\Lambda_{0}}$ terms get canceled out from \eqref{eq:ratio111}.
Hence, without loss of generality we can assume $\mu_{0}=0$ and $\Lambda_{0}=I_{p}$.	
Without loss of generality we can also assume that $y_{1+\sum_{j=1}^{h-1}n_{j}},\dots, y_{\sum_{j=1}^{h} n_{j} }$ are in cluster $h$. We define 
$$Y^{\Psi}_{(h)}= \left[y_{1+\sum_{j=1}^{h-1}n_{j}},\dots, y_{\sum_{j=1}^{h} n_{j} } \right]\trans,$$ to be the sub-data matrix corresponding to the $h$-th cluster in partition $\Psi$.
Exploiting lower rank factorization results on matrix determinants, we have
\begin{align*}
&\abs{I_{p}+  S_h^{\Psi} +\frac{n_h\kappa_0}{n_h+\kappa_0}\bar{y}_h^{\Psi}\bar{y}_h^{\Psi\trans}}=  
\abs{I_{p}+Y^{\Psi\trans}_{(h)} Y^{\Psi}_{(h)} - \frac{1}{n_{h}+\kappa_{0}} Y^{\Psi\trans}_{(h)} \bone_{n_{h}}\bone_{n_{h}}\trans Y^{\Psi }_{(h)}}\\
=& \abs{ 1- \frac{1}{n_{h}+\kappa_{0}}  \bone_{n_{h}}\trans  Y^{\Psi}_{(h)} \left\{I_{p}+Y^{\Psi\trans}_{(h)} Y^{\Psi}_{(h)} \right\}^{-1}Y^{\Psi\trans}_{(h)} \bone _{n_{h}}  } \abs{I_{p}+Y_{(h)}^{\Psi\trans} Y_{(h)}^{\Psi}},
\end{align*}
{where the symbol $\abs{A}$ or $\abs{a}$ is to be interpreted as the determinant of the matrix $A$ or the absolute value of the scalar $a$, respectively.}
Then, the second factor of  \eqref{eq:ratio111} simplifies to 
\begin{equation*}
\frac{
\left\{ \abs{ 1- \frac{1}{n'_{\kn-1}+\kappa_{0}}  \bone_{n'_{\kn-1}}\trans  Y_{(\kn-1)}^{\Psi'} \left(I_{p}+Y_{(\kn-1)}^{\Psi^{\prime \transp}} Y_{(\kn-1)}^{\Psi' } \right)^{-1}Y_{(\kn-1)}^{\Psi^{\prime \transp}} \bone _{n'_{\kn-1}}  } \abs{I_{p}+Y_{(\kn-1)}^{\Psi^{\prime \transp}} Y_{(\kn-1)}^{\Psi' } } \right\}^{\frac{\nu_{0}+n'_{k_{n}-1}}{2}} } {
\prod_{h=\kn-1}^{\kn} \left\{ \abs{ 1- \frac{1}{n_{h}+\kappa_{0}}  \bone_{n_{h}}\trans  Y_{(h)}^{\Psi} \left( I_{p}+Y_{(h)}^{\Psi\trans} Y_{(h)}^{\Psi } \right)^{-1}Y_{(h)}^{\Psi\trans} \bone _{n_{h}}  } \abs{I_{p}+Y_{(h)}^{\Psi\trans} Y_{(h)}^{\Psi } } \right\} ^{\frac{\nu_{0}+n_{h}}{2}}		 
}.
\end{equation*}
Using Lemma \ref{lemma:submatrix},
$\lim_{p\to\infty} \norm{Y_{(h)} \left\{I_{p}+Y_{(h)}\trans Y_{(h)}\right\}^{-1}Y_{(h)}\trans-I_{n_{h}}}_{2}=0$ in $\p$-probability and 
\begin{equation}
\lim_{p\to\infty}\abs{ 1- \frac{1}{n_{h}+\kappa_{0}}  \bone_{n_{h}}\trans  Y_{(h)} \left\{I_{p}+Y_{(h)}\trans Y_{(h)}\right\}^{-1}Y_{(h)}\trans \bone _{n_{h}}  }=\frac{\kappa_{0}}{\kappa_{0}+n_{h}}\quad \mbox{in $\p$-probability}.
\label{eq:determinant_limit}
\end{equation}
Taking the logarithm of the second and third factor of \eqref{eq:ratio111} and rearranging it using the previous result  
\begin{multline}
\log\frac{\abs{ 1- \frac{1}{n'_{\kn-1}+\kappa_{0}}  \bone_{n'_{\kn-1}}\trans  Y_{(\kn-1)}^{\Psi'} \left\{I_{p}+Y_{(\kn-1)}^{\Psi^{\prime \transp}} Y_{(\kn-1)}^{\Psi' } \right\}^{-1}Y_{(\kn-1)}^{\Psi^{\prime \transp}} \bone _{n'_{\kn-1}}  }^{\frac{\nu_{0}+n'_{\kn-1}}{2}} }{\prod_{h=\kn-1}^{\kn} \abs{ 1- \frac{1}{n_{h}+\kappa_{0}}  \bone_{n_{h}}\trans  Y_{(h)}^{\Psi} \left\{I_{p}+Y_{(h)}^{\Psi\trans} Y_{(h)}^{\Psi } \right\}^{-1}Y_{(h)}^{\Psi\trans} \bone _{n_{h}}  }^{\frac{\nu_{0}+n_{h}}{2}} }\\			
+\log \left\{ 
\frac{\kappa_{0} (\kappa_{0}+n'_{\kn-1})}{
(\kappa_{0}+n_{\kn-1})(  \kappa_{0}+n_{\kn} ) } \right\}^{p/2} 
+ \log \frac{\abs{I_{p}+Y_{(\kn-1)}^{\Psi^{\prime \transp}} Y_{(\kn-1)}^{\Psi' } }^{\frac{\nu_{0}+n'_{\kn-1}}{2}} } {\prod_{h=\kn-1}^{\kn} \abs{I_{p}+Y_{(h)}^{\Psi\trans} Y_{(h)}^{\Psi } } ^{\frac{\nu_{0}+n_{h}}{2}}}.
\label{eq:2nd3rd}
\end{multline}
Since $\nu_{0}=p+c$, in conjunction with \eqref{eq:determinant_limit} we have sum of the limits of the first and second terms in \eqref{eq:2nd3rd} is 0 in $\p$-probability.
We finally study the last summand of  \eqref{eq:determinant_limit} and particularly
\begin{equation}
\lim_{p\to \infty} \frac{1}{p}\log \frac{\abs{I_{p}+Y_{(\kn-1)}^{\Psi^{\prime \transp}} Y_{(\kn-1)}^{\Psi' } }^{\frac{\nu_{0}+n'_{\kn-1}}{2}} } {\prod_{h=\kn-1}^{\kn} \abs{I_{p}+Y_{(h)}^{\Psi\trans} Y_{(h)}^{\Psi } } ^{\frac{\nu_{0}+n_{h}}{2}}}.
\label{eq:3rd}
\end{equation}
Since $\abs{I_{p}+Y_{(h)}^{\Psi\trans} Y_{(h)}^{\Psi } }=\abs{I_{n_{h}}+Y_{(h)}^{\Psi} Y_{(h)}^{\Psi \trans} }$ for any partition $\Psi$, 
following the same arguments of \eqref{ll1} from Lemma \ref{lem:cov_term} in the supplementary materials, we have  
$$
\abs{I_{p}+Y_{(\kn-1)}^{\Psi^{\prime \transp}} Y_{(\kn-1)}^{\Psi' } }= \prod_{h=\kn-1}^{\kn} \abs{I_{p}+Y_{(h)}^{\Psi\trans} Y_{(h)}^{\Psi } } ^{\frac{\nu_{0}+n_{h}}{2}} \times \abs{I_{n_{\kn-1}} -ZZ\trans}
$$ 
where $Z= \{I_{n_{\kn-1}}+ Y_{(\kn-1)}^{\Psi} Y_{(\kn-1)}^{\Psi\trans} \}^{-\half} Y_{(\kn-1)}^{\Psi} Y_{(\kn)}^{\Psi\trans} \{I_{n_{\kn}}+ Y_{(\kn)}^{\Psi} Y_{(\kn)}^{\Psi\trans} \}^{-\half}$, and  \eqref{eq:3rd} reduces to 
$$\frac{n_{\kn}}{2p}\log \abs{I_{n_{\kn-1}}+ Y_{(\kn-1)}^{\Psi} Y_{(\kn-1)}^{\Psi\trans}} +\frac{n_{\kn-1}}{2p}\log \abs{I_{n_{\kn}}+ Y_{(\kn)}^{\Psi} Y_{(\kn)}^{\Psi\trans}}+ \frac{\nu_{0}+n'_{\kn-1}}{2p}\log \abs{I_{n_{\kn-1}} -ZZ\trans}.$$		
From \ref{cond_neg}, it can be deduced that the limits of the first two terms in the last expression are 0 in $\p$-probability as $p\to\infty$. 
Invoking Lemma \ref{lem:cov_term} and the fact that $\nu_{0}=p+c$, we have 
$$\limsup_{p\rightarrow\infty} \frac{\nu_{0}+n'_{\kn-1}}{2p}\log \abs{I_{n_{\kn-1}} -ZZ\trans}<0,$$ 
and henceforth \eqref{eq:3rd} is negative.	This leads to
\begin{equation*}
\limsup_{p\to \infty}\log \frac{
\prod_{h=1}^{\kn} \int \prod_{i:c_i=h}  
\mn_p\left(y_i; \mu_{h} ,\Sigma_{h}\right) 
\mn_p(\mu_h; \mu_0, \kappa_0^{-1}\Sigma_{h}) 
IW(\Sigma_{h}; \nu_0, \Lambda_{0}) 
\, \de( \mu_{h},\Sigma_{h})} {
\prod_{h=1}^{\kn-1} \int \prod_{i:c'_i=h} 
\mn_p\left(y_i; \mu_{h} ,\Sigma_{h}\right) 
\mn_p(\mu_h; \mu_0, \kappa_0^{-1}\Sigma_{h}) 
IW(\Sigma_{h}; \nu_0, \Lambda_{0}) 
\, \de( \mu_{h},\Sigma_{h})
}= 0,
\end{equation*}
and hence for $p \to \infty$ all the data points are cluster together in $\p$-probability  thanks to Theorem \ref{th:negative_result}.
\end{proof}

\begin{proof}[Corollary \ref{cor:neg_result}]	
Define $c_{1},\dots,c_{n}$ and  $c'_{1},\dots,c'_{n}$ consistently with the proof of Theorem \ref{th:negative_result}. 
Then, consider the ratio of the marginal likelihoods
\begin{equation}
\frac{
\int \prod_{h=1}^{k_n} \int \prod_{i:c_i=h}  
\mn_p\left(y_i; \mu_h ,\Sigma\right) 
\mn_p(\mu_h; \mu_0, \kappa_0^{-1}\Sigma) \de \mu_h
IW(\Sigma; \nu_0, \Lambda_0) 
\, \de \Sigma
}{
\int \prod_{h=1}^{k_n-1} \int \prod_{i:c'_i=h} 
\mn_p\left(y_i; \mu_h ,\Sigma\right) 
\mn_p(\mu_h; \mu_0, \kappa_0^{-1}\Sigma) \de \mu_h
IW(\Sigma; \nu_0, \Lambda_0) 
\, \de \Sigma
}.\label{eq:ratio1_neg}
\end{equation}
The numerator of \eqref{eq:ratio1_neg} is 
$$\prod_{h=1}^{k_n} \left(\frac{\kappa_0}{n_h+\kappa_0}\right)^{\frac{p}{2}} 	\abs{\Lambda_0+\sum_{h=1}^{k_n} \left\{S_h^{\Psi} +\frac{n_h\kappa_0}{n_h+\kappa_0}(\bar{y}_h^{\Psi}-\mu_0)(\bar{y}_h^{\Psi}-\mu_0)\trans\right\}} ^{-\frac{\nu_0 + n}{2}} \hspace{-0.2cm}
\pi^{-\frac{np}{2}} \frac{\Gamma_p(\frac{\nu_0 + n}{2})}{\Gamma_p(\nu_0/2)} \abs{\Lambda_0}^{\frac{\nu_0}{2}},$$
where $\bar{y}_h^{\Psi}=\frac{1}{n_h}\sum_{i:c_i=h} y_{i}$ and $S^{\Psi}_h=\sum_{i:c_{i}=h} \left(y_{i}-\bar{y}^{\Psi}_h\right)\left(y_{i}-\bar{y}^{\Psi}_h\right)\trans$.
Hence, obtaining a corresponding expression for the denominator, ratio \eqref{eq:ratio1_neg} becomes 
\begin{align}
\left\{\frac{\kappa_0(\kappa_0+n'_{k_n-1})}{(\kappa_0+n_{k_n-1})(\kappa_0+n_{k_n})} \right\}^{\frac{p}{2}}
\left[
\frac{
\abs{\Lambda_0+\sum_{h=1}^{k_n-1} \left\{S_h^{\Psi'} +\frac{n'_h\kappa_0}{n'_h+\kappa_0}(\bar{y}_h^{\Psi'}-\mu_0)(\bar{y}_h^{\Psi'}-\mu_0)\trans\right\}}
}{
\abs{\Lambda_0+\sum_{h=1}^{k_n} \left\{S_h^{\Psi} +\frac{n_h\kappa_0}{n_h+\kappa_0}(\bar{y}_h^{\Psi}-\mu_0)(\bar{y}_h^{\Psi}-\mu_0)\trans\right\}} 
}\right]^{\frac{\nu_0 + n}{2}}.
\label{eq:marglikcor2}
\end{align}
First note that for $n_{k_n}, n_{k_n-1}\geq1$
\begin{equation}
\frac{\kappa_0(\kappa_0+n'_{k_n-1})}{(\kappa_0+n_{k_n-1})(\kappa_0+n_{k_n})}< 1.\label{eqqq1}
\end{equation}
Similar to Corollary \ref{cor:neg_result2}, we can assume without loss of generality $\mu_0$ to be a $p$-dimensional vector of zero and  $\Lambda_{0}=I_{p}$. 
Note that,
\begin{align}
\sum_{h=1}^{\kn}\left(	S_h^{\Psi} +\frac{n_h\kappa_0}{n_h+\kappa_0}\bar{y}_h^{\Psi}\bar{y}_h^{\Psi\trans} \right) &=\sum_{i=1}^{n} y_{i} y_{i}\trans -\sum_{h=1}^{k} \frac{n_{h}^{2}}{n_{h}+\kappa_{0}} \bar{y}_h^{\Psi}\bar{y}_h^{\Psi \trans}\notag\\
&=\sum_{i=1}^{n} y_{i} y_{i}\trans  -\sum_{h=1}^{k} \frac{1}{n_{h}+\kappa_{0}} \left(\sum_{i:c_{i}=h}y_{i}\right)\left(\sum_{i:c_{i}=h}y_{i}\right)\trans. \label{eq:matrix_determinant_lemma}
\end{align}
Also, without loss of generality we can assume $\left\{y_{1+\sum_{j=1}^{h-1}n_{j}},\dots, y_{\sum_{j=1}^{h} n_{j} } \right\}$ are in cluster $h$ of $\Psi$   similarly to Corollary \ref{cor:neg_result2}.
Then $\sum_{h=1}^{\kn} \left(S_{h}^{\Psi} +\frac{n_h\kappa_0}{n_h+\kappa_0}\bar{y}_h^{\Psi}\bar{y}_h^{\Psi \trans}\right) =Y\trans\left(I_{n}-\Je^{\Psi}_{n}\right)Y$, where 
$\Je^{\Psi}_{n}=$ $\diag\left(\frac{J_{n_{1}}}{n_{1}+\kappa_{0}},\dots,\frac{J_{n_{\kn}}}{n_{\kn}+\kappa_{0}} \right)$ is an $n\times n$ order block diagonal matrix and $J_{r}$ is the $r\times r$ order square matrix with all elements being 1. Clearly, $\Je^{\Psi}_{n}$ is a positive semi-definite matrix of rank $\kn$. 
Henceforth, exploiting the lower rank factorization structure, each determinant in \eqref{eq:marglikcor2} can be simplified as 	
\begin{align*}
\abs{I_{p}+\sum_{h=1}^{\kn} \left( S_h^{\Psi} +\frac{n_h\kappa_0}{n_h+\kappa_0}\bar{y}_h^{\Psi}\bar{y}_h^{\Psi \trans}\right)} & = \abs{I_{p}+Y\trans\left(I_{n}-\Je^{\Psi}_{n}\right)Y}\\
& = 
\abs{I_{p}+Y\trans Y}\abs{I_{n}- \Je^{\Psi\half}_{n}Y(I_{p}+Y\trans Y)^{-1}Y\trans \Je^{\Psi\half}_{n} }.
\end{align*}
Hence \eqref{eq:marglikcor2} reduces to 
\begin{align}
\left\{\frac{\kappa_0(\kappa_0+n'_{k_n-1})}{(\kappa_0+n_{k_n-1})(\kappa_0+n_{k_n})} \right\}^{p/2}
\left\{
\frac{
\abs{I_{n}- \Je^{\Psi'\half}_{n}Y(I_{p}+Y\trans Y)^{-1}Y\trans \Je^{\Psi'\half}_{n} }
}{
\abs{I_{n}- \Je^{\Psi\half}_{n}Y(I_{p}+Y\trans Y)^{-1}Y\trans \Je^{\Psi\half}_{n} } 
}\right\}^{\frac{\nu_0 + n}{2}}.\label{eq3}
\end{align}
From Lemma \ref{lemma:matrixlimit} in the supplementary materials, $\lim_{p\to\infty}\norm{ Y(I_{p}+Y\trans Y)^{-1}Y\trans -I_{n}}_{2}=0$ in $\p$-probability. 
Therefore, from the construction of $\Je^{\Psi}_{n}$,
\begin{equation}
\lim_{p\to\infty}\abs{I_{n}- \Je^{\Psi\half}_{n}Y(I_{p}+Y\trans Y)^{-1}Y\trans \Je^{\Psi\half}_{n} }=\abs{I_{n}- \Je^{\Psi}_{n} }=\prod_{h=1}^{k_{n}} \abs{I_{n_{h}}-\frac{1}{n_{h}+\kappa_{0}}J_{n_{h}}},\label{eq:eq1}
\end{equation}
in $\p$-probability. 	Notably $J_{r}={1}_{r} {1}_{r}\trans $ where ${1}_{r}$ is the $r$-dimensional vector of ones, implying that  $\abs{I_{r}-\frac{1}{n_{r}+\kappa_{0}}J_{r}}=\frac{\kappa_{0}}{n_{r}+\kappa_{0}}$ for any positive integer $r$. 
Substituting this in \eqref{eq:eq1}, we have
\begin{equation*}
\lim_{p\to\infty}\abs{I_{n}- \Je^{\Psi\half}_{n}Y(I_{p}+Y\trans Y)^{-1}Y\trans \Je^{\Psi\half}_{n} }=\prod_{h=1}^{k_{n}}\frac{\kappa_{0}}{n_{h}+\kappa_{0}},\quad \mbox{{in $\p$-probability}}
\end{equation*}
and therefore, 
$$\lim_{p\to\infty} \frac{
\abs{I_{n}- \Je^{\Psi'\half}_{n}Y(I_{p}+Y\trans Y)^{-1}Y\trans \Je^{\Psi'\half}_{n} }
}{
\abs{I_{n}- \Je^{\Psi\half}_{n}Y(I_{p}+Y\trans Y)^{-1}Y\trans \Je^{\Psi\half}_{n} } 
}= \frac{(\kappa_0+n_{k_n-1})(\kappa_0+n_{k_n})}{\kappa_0(\kappa_0+n'_{k_n-1})}\quad  \mbox{in $\p$-probability.} $$
Thus if we take the log of \eqref{eq3} multiplied by $p^{-1}$ and study its limit we have
\begin{align*}
\liminf_{p\rightarrow\infty}&\left\{ \frac{1}{2}\log \frac{\kappa_0(\kappa_0+n'_{k_n-1})} {(\kappa_0+n_{k_n-1})(\kappa_0+n_{k_n})} +\frac{n+\nu_0}{2p}\log \frac{
\abs{I_{n}- \Je^{\Psi'\half}_{n}Y(I_{p}+Y\trans Y)^{-1}Y\trans \Je^{\Psi'\half}_{n} }
}{
\abs{I_{n}- \Je^{\Psi\half}_{n}Y(I_{p}+Y\trans Y)^{-1}Y\trans \Je^{\Psi\half}_{n} } 
} \right\}\\
=& \frac{1}{2} \log \frac{\kappa_0(\kappa_0+n'_{k_n-1})} {(\kappa_0+n_{k_n-1})(\kappa_0+n_{k_n})}\times \left(1-\limsup_{p\rightarrow\infty}\frac{n+\nu_{0}}{p} \right)> 0.
\end{align*}   
Since $n$ is fixed with $p$, the above limit follows from \eqref{eqqq1} and the assumption on $\nu_0$.
Thus we have
$$\liminf_{p\rightarrow\infty}	
\frac{
\int \prod_{h=1}^{k_n} \int \prod_{i:c_i=h}  
\mn_p\left(y_i; \mu_h ,\Sigma\right) 
\mn_p(\mu_h; \mu_0, \kappa_0^{-1}\Sigma) \de \mu_h
IW(\Sigma; \nu_0, \Lambda_0) 
\, \de \Sigma
}{ 
\int \prod_{h=1}^{k_n-1} \int \prod_{i:c'_i=h} 
\mn_p\left(y_i; \mu_h ,\Sigma\right) 
\mn_p(\mu_h; \mu_0, \kappa_0^{-1}\Sigma) \de \mu_h
IW(\Sigma; \nu_0, \Lambda_0) 
\, \de \Sigma
}=\infty,$$
and hence for $p \to \infty$ each data point is clustered separately  in $\p$-probability thanks to Theorem~\ref{th:negative_result}.
\end{proof}

\subsection*{Proofs of Section \ref{sec:properties}}

\begin{proof}[Lemma \ref{lemma:suff_cond}]
Let 
${\boldzeta}^{(p)}_0= (\sqrt{p\log p})^{-1}(\Lambda\trans \Lambda)^{-\half}\Lambda\trans \Lambda_{0}{\boldeta}_0$, 
then
$\Pi(\Psi \mid \boldeta_{0}) = \Pi(\Psi \mid \boldzeta^{(p)}_{0})$. 
Then, 
\begin{equation}
\frac{1}{\sqrt{p\log p}} \norm{\boldzeta^{(p)}_{i}-\boldzeta^{(p)}_{0i}} \leq  \norm{(\Lambda\trans\Lambda)^{-\half}\Lambda }_{2} \times \frac{1}{\sqrt{p\log p}}\norm{\Lambda\eta_{i}-\Lambda_{0}\eta_{0i}} \leq \frac{1}{\sqrt{p\log p}}\norm{\Lambda\eta_{i}-\Lambda_{0}\eta_{0i}}.\label{eq:bound}
\end{equation}	
From \eqref{eq:prior_on_eta} we see that the numerator in the right hand side of \eqref{eq:marginalposterior_eta} can be simplified as 
\begin{equation}
C\times \Pi (\Psi)\times \prod_{h=1}^{k_n} \left(\frac{\kappa_0}{n_h+\kappa_0}\right)^{\frac{d}{2}} \times  \abs{\sum_{h=1}^{k_n}\left\{S^h_{\eta_0}+\frac{n_h}{n_h+1} \bar{\eta}_{0}^{h} \bar{\eta}_{0}^{h\trans}  \right\}}^{-\frac{n}{2}},\label{eq:or_oracle}
\end{equation}	
where $n_h=\sum_{i=1}^n I(c_i=h)$, $\bar{\eta}^h_0=\frac{1}{n_h} \sum_{i:c_i=h}\eta_{0i}$, $S_{\eta_0}^h=\sum_{i:c_i=h}(\eta_{0i} -\bar{\eta}_0^h)(\eta_{0i} -\bar{\eta}_0^h)\trans $ and $C$ is a positive quantity constant across all $\Psi' \in {\mathscr P}$. 
Hence it is clear that $\Pi(\Psi \mid {\eta})$ is a continuous function of ${\eta}$.
Since the function is bounded (being a probability function), the continuity is also uniform.
Also note that, for the particular choice of Gaussian kernel and base measure in \eqref{eq:prior_on_eta}, the oracle partition probability \eqref{eq:marginalposterior_eta}  is unchanged if ${\eta}$ is multiplied by a full-rank square matrix and therefore $\Pi(\Psi\mid{\zeta}_{0}^{(p)})=\Pi(\Psi \mid {\eta}_{0})$. 	
Therefore, for any $\epsilon>0$ there exists $\delta>0$ such that
$\norm{\boldzeta_0^{(p)}-\boldzeta^{(p)}}<\delta$ implies that $\abs{\Pi(\Psi\mid\boldzeta^{(p)})- \Pi(\Psi\mid\boldzeta_{0}^{(p)}) }=\abs{\Pi(\Psi\mid\boldzeta^{(p)})-\Pi(\Psi \mid \boldeta_{0} )}<\epsilon$.
Again,
\begin{multline}
E\left\{\abs{\Pi(\Psi  \mid {\boldzeta^{(p)}}) - \Pi(\Psi \mid \boldeta_{0} )} \bigg{\lvert}~ \by  \right\}=
E\left\{\abs{\Pi(\Psi  \mid {\boldzeta^{(p)}}) - \Pi(\Psi \mid \boldzeta^{(p)}_{0} )} \bigg{\lvert}~ B_{p,\delta}, \by  \right\} \Pi(B_{p,\delta} \mid \by)\\
+E\left\{\abs{\Pi(\Psi  \mid {\boldzeta^{(p)}}) - \Pi(\Psi \mid \boldeta_{0} )} \bigg{\lvert}~ \bar{B}_{p,\delta}, \by  \right\} \Pi(\bar{B}_{p,\delta}\mid \by).
\label{eq:eqn_suff}
\end{multline}
Due to continuity, $\delta$ can be chosen sufficiently small such that the term inside the first expectation in the right hand side of \eqref{eq:eqn_suff} is smaller than arbitrarily small $\epsilon>0$.
Now for any $\delta>0$, the second term in the right hand side of \eqref{eq:eqn_suff} goes to 0 as  $\Pi (\bar{B}_{p,\delta}\mid \by)\to 0$ as $p\to\infty$ by assumption.
Therefore, for arbitrarily small $\epsilon>0$, $E\left\{\abs{\Pi(\Psi  \mid {\boldzeta^{(p)}}) - \Pi(\Psi \mid \boldeta_{0} )}~ \big{\lvert} ~ \by  \right\}<\epsilon$
for large enough $p$.
Hence the proof.
%
\end{proof}

\baselineskip=14pt
\section*{Supplementary Materials}
Proofs of additional theoretical results and simulation studies, and MCMC convergence diagnostics are provided in the supplementary materials.


\section*{Acknowledgments}
This work was partially funded by grants R01-ES027498 and R01-ES028804  from the National Institute of Environmental Health Sciences of the United States Institutes of National Health and by the University of Padova under the STARS Grant.


\baselineskip=14pt
\bibliographystyle{natbib}
\bibliography{refs}


\clearpage\pagebreak\newpage
\newgeometry{textheight=9in, textwidth=6.5in,}
\pagestyle{fancy}
\fancyhf{}
\rhead{\bfseries\thepage}
\lhead{\bfseries SUPPLEMENTARY MATERIALS}

\baselineskip 20pt
\begin{center}
{\LARGE{Supplementary Materials for\\}} 
\papertitle
\end{center}

\setcounter{equation}{0}
\setcounter{page}{1}
\setcounter{table}{1}
\setcounter{figure}{0}
\setcounter{section}{0}
\numberwithin{table}{section}
\renewcommand{\theequation}{S.\arabic{equation}}
\renewcommand{\thesubsection}{S.\arabic{section}.\arabic{subsection}}
\renewcommand{\thesection}{S.\arabic{section}}
\renewcommand{\thepage}{S.\arabic{page}}
\renewcommand{\thetable}{S.\arabic{table}}
\renewcommand{\thefigure}{S.\arabic{figure}}
\baselineskip=15pt

\vspace{0cm}

\authors

\vskip 10mm

Supplementary materials present 
proofs of additional theoretical results, some figures additional to the simulation studies and MCMC convergence diagnostics are provided in the main paper.

\newpage
\baselineskip=16pt
\section{Additional Theoretical Results}	

In the supplementary materials, we denote by $\norm{x}$ the Euclidean norm of a vector $x$ and by $\norm{X}_2$ the spectral norm of a matrix $X$. 
The smallest and largest eigenvalues of the matrix $(X\trans X)^{\frac{1}{2}}$ are denoted by $\smin(X)$ and $\smax(X)$, respectively. For a positive-definite matrix $X$, $\lambda_{\min}(X)$ and $\lambda_{\max}(X)$ denote the smallest and largest eigenvalues, respectively. 

\begin{lemma}
	\label{lemma:multgamma_ratio}
	Let $\Gamma_{p}(\cdot)$ be the multivariate gamma function, $\nu_{0}=p+c$ for some constant $c\geq0$, and $\ell$ and $m$ be (not varying with $p$) non-negative integers.
	Then,		
	$\lim_{p\to\infty}\frac{1}{p}\log \frac{\Gamma_p(\frac{\nu_{0} +\ell}{2})} {\Gamma_p(\frac{\nu_0+m}{2})}=0$.
\end{lemma}
\begin{proof}
	Without loss of generality assuming $\ell>m$, we have
	\begin{equation}
		\frac{\Gamma_p(\frac{\nu_0 +\ell}{2})}{\Gamma_p(\frac{\nu_0+m}{2})}=\prod_{j=1}^p\frac{\Gamma\left(\frac{\nu_0+\ell-j+1}{2} \right) }{\Gamma\left(\frac{\nu_0+m-j+1}{2} \right) }=\frac{\prod_{j=m+1}^\ell \Gamma\left(\frac{\nu_0+j}{2} \right)}{\prod_{j=m}^\ell \Gamma\left(\frac{\nu_0+j-p}{2} \right)}.
		\label{eq:mult_gam_ratio}
	\end{equation}
	Note that the denominator term in the extreme right hand of \eqref{eq:mult_gam_ratio} does not depend on $p$ as $\nu_0-p$ is constant from assumption.
	Applying Stirling's approximation on the numerator we get
	\begin{align*}
		\frac{\Gamma_p(\frac{\nu_0 +\ell}{2})}{\Gamma_p(\frac{\nu_0+m}{2})}&= \frac{1}{\prod_{j=m}^\ell \Gamma\left(\frac{\nu_0+j-p}{2} \right)} \times \prod_{j=m+1}^{\ell} \left\{\sqrt{2\pi\frac{\nu_0+j-1}{2}}\left(\frac{\nu_0+j-1}{2e}\right)^{\frac{\nu_0+j-1}{2}}E_{j}\right\}\\
		&= \frac{1}{\prod_{j=m}^\ell \Gamma\left(\frac{\nu_0+j-p}{2} \right)} \times \prod_{j=m+1}^{\ell} \left\{\sqrt{2\pi e} \left(\frac{\nu_0}{2e} \right) ^{\frac{\nu_0+j+1}{2}}\left(1+\frac{j-1} {\nu_0} \right) ^{\frac{\nu_0+j+1}{2}} E_{j}\right\},
	\end{align*}
	where $E_{j}=O(\log p)$ arising from the Stirling's approximation formulae.
	Using the result $\lim_{x\to \infty}(1+c/x)^{x}=e^c$, it can be seen that
	\begin{multline*}
		\lim_{p\to\infty} \prod_{j=m+1}^{\ell} \left\{\sqrt{2\pi e} \left(\frac{\nu_0}{2e} \right) ^{\frac{\nu_0+j+1}{2}}\left(1+\frac{j-1} {\nu_0} \right) ^{\frac{\nu_0+j+1}{2}} \right\}\\
		= (2\pi e)^{\frac{\ell-m}{2}}\times \left(\frac{\nu_0}{2e} \right) ^{\frac{1}{4}(\ell-m) (2\nu_0+m+\ell+3)}\times e^{\frac{1}{2}(\ell-m)(\ell+m-1)},
	\end{multline*}
	which is a finite quantity.
	Hence the proof. 
\end{proof}

\begin{lemma}
	\label{lemma:matrixlimit}
	For any $n\times p$ order matrix $Y$ satisfying \ref{cond_neg}, $\lim_{p\to\infty}\norm{ Y(I_{p}+Y\trans Y)^{-1}Y\trans -I_{n}}_{2}=0$ in $\p$-probability.
\end{lemma}
\begin{proof}
	Letting $Y=UDV$, the singular value decomposition of $Y$, we have $Y(I_{p}+Y\trans Y)^{-1}Y\trans =U\diag\left(\frac{d_{1}^{2}}{1+d_{1}^{2}},\dots, \frac{d_{n}^{2}}{1+d_{n}^{2}}\right)U\trans $ where $d_{1},\dots,d_{n}$ are the singular values of $Y$ in descending order. 
	From \ref{cond_neg} we have $\liminf_{p\to\infty}\frac{1}{p}d_{i}^{2}>0$, which further implies that $\liminf_{p\to\infty}\frac{d_{i}}{1+d_{i}}\to1$ for all $i=1,\dots,n$. 
	As $\frac{d_{i}}{1+d_{i}}\leq1$, $\lim_{p\to\infty}\norm{ Y(I_{p}+Y\trans Y)^{-1}Y\trans -I_{n}}_{2}=0$ in $\p$-probability.
\end{proof}

\begin{lemma}
	\label{lemma:submatrix}
	Let $\wt{Y}$ be an $\wt{n}\times p$ order matrix, formed by arbitrarily selecting $\wt{n}$ rows from $Y=[y_{1},\dots,y_{n}]\trans$ where $1 \leq \wt{n}\leq n$. If $Y$ satisfies \ref{cond_neg}, then $\lim_{p\to\infty}\norm{ \wt{Y}(I_{p}+\wt{Y}\trans \wt{Y})^{-1}\wt{Y}\trans -I_{\wt{n}}}_{2}=0$ in $\p$-probability.
\end{lemma}
\begin{proof}
	Letting $Y=UDV$ the singular value decomposition of $Y$, we have $\wt{Y}=\wt{U}DV$ where $\wt{U}$ is formed by the corresponding rows of $Y$ which were used to form $\wt{Y}$.
	Using \citet[Lemma 1.1(iii) from the Suppplementary section]{pati2014}, we have $\smin(\wt{Y}\trans)\geq \smin(V\trans) \smin(D\trans)  \smin(\wt{U}\trans)=\smin(Y\trans) $.
	Since $\smin(\wt{U}\trans)=\smin({U}\trans)=1$, $\smin(\wt{Y}\trans)\geq \smin(Y\trans) $.
	Therefore, $\wt{Y}$ also satisfies \ref{cond_neg} if we substitute $Y=\wt{Y}$.
	Consequently applying Lemma \ref{lemma:matrixlimit}, we conclude the proof.
\end{proof}

\begin{lemma}
	\label{lem:cov_term}
	Let $Y$ be an $n\times p$ order matrix satisfying \ref{cond_neg}.
	Let $Y_{i}=[y_{j_{i,1}},\dots,  y_{j_{i,n_{i}}} ]\trans$, $i=1,2$ be an arbitrary partiton of the data-matrix into two sub-matrices such that $n_{1}+n_{2}=n$.
	Then $\limsup_{p\rightarrow\infty} \smax(Z)<1$ where $Z=(I_{n_{1}}+Y_{1}Y_{1}\trans)^{-\half}Y_{1}Y_{2}\trans (I_{n_{2}}+Y_{2}Y_{2}\trans)^{-\half}$ in $\p$-probability.
\end{lemma}
\begin{proof}
	From \ref{cond_neg} we have $\norm{YY\trans}_{2}=O(p)$ and $\liminf \lambda_{\min}(YY\trans)/p>0$, which implies that
	\begin{equation}
		0<\liminf_{p\rightarrow\infty} \abs{(I_{n}+YY\trans)/p}\leq \limsup_{p\rightarrow\infty} \abs{(I_{n}+YY\trans)/p} =O(1)\text{ in }\p\text{-probability}. \label{eq:lim_on_det}
	\end{equation} 
	Following the proof of Lemma \ref{lemma:submatrix}, we see that $Y_{i}$ also satisfies \ref{cond_neg}, and therefore \eqref{eq:lim_on_det} also holds if $Y$ is replaced with $Y_{i}$ for $i=1,2$.
	Noting that $I_{n}+YY\trans= \begin{bmatrix}
		I_{n_{1}}+ Y_{1} Y_{1}\trans & Y_{1}Y_{2}\trans\\
		Y_{2}Y_{1}\trans & I_{n_{2}}+ Y_{2} Y_{2}\trans
	\end{bmatrix}$
	and using matrix factorization results, we have
	\begin{equation}
		\abs{\frac{1}{p} (I_{n}+YY\trans)}= \abs{\frac{1}{p} (I_{n_{1}}+Y_{1}Y_{1}\trans)} \abs{\frac{1}{p} (I_{n_{2}}+Y_{2}Y_{2}\trans)} 
		\abs{I_{n_1}-ZZ\trans}. \label{ll1}
	\end{equation}
	Again in $\p$-probability,
	\begin{equation}
		\limsup_{p\rightarrow\infty}	\smax^{2}(Z)\leq \limsup_{p\rightarrow\infty} \norm{ Y_{1}\trans (I_{n_{1}}+Y_{1}Y_{1}\trans)^{-1}Y_{1}}_{2} \norm{ Y_{2}\trans (I_{n_{2}}+Y_{2}Y_{2}\trans)^{-1}Y_{2}}_{2}\leq 1. \label{ll}
	\end{equation}
	For \eqref{eq:lim_on_det} to hold, all the terms in the RHS of \eqref{ll1} must be bounded away from 0. 
	As $\abs{I_{n_{1}}-ZZ\trans}=\prod_{j=1}^{n_{1}}\abs{1- s_{j}^{2}(Z)}$, the inequality on \eqref{ll} must be strict.
	Thus, we conclude the proof.
\end{proof}

\begin{lemma}
	\label{th:DL_eig_bound}
	For prior \eqref{eq:dir_laplace_prior} $\lim_{p\rightarrow\infty} \frac{1}{p} \lambda_{\min}(\Lambda\trans\Lambda)=\lim_{p\rightarrow\infty} \frac{1}{p}\lambda_{\max}(\Lambda\trans\Lambda)=v_1$ for some $v_1>0$ $\Pi$-a.s.
\end{lemma}
\begin{proof}
	From \citet[Eqn (9) and Section 2.4]{dir_laplace} we get that $\Lambda\trans \Lambda=\tau^2  T\trans T$ where the $(i,j)$-th element of $T$ is $t_{ij}=e_{ij}\phi_{ij}$ with $e_{ij}\overset{iid}{\sim}\mbox{DE}(1)$ where $\mbox{DE}(b)$ is the double exponential distribution with median 0 and variance $2b^{2}$.
	Additionally $\phi\sim \mbox{Dir}(a,\ldots,a)$ and $\tau \sim \mbox{Ga}(pda,\sfrac{1}{2} )$.  

		Now by the strong law of large numbers $\norm{\frac{1}{p}\tilde{T} \trans\tilde{T}-v_1 I_d } _F\rightarrow0$ as $p\to\infty$ where $\norm{\cdot}_F$ is the Frobenius norm of a matrix and $v_1=\text{Var}(e_{ij} \gamma_{ij})$. Hence, for any $i=1,\dots,p$ $\lim_{p\rightarrow\infty} \lambda_i( \tilde{T} \trans\tilde{T})p^{-1}=v_1$. Also $\lim_{p\rightarrow\infty} \tau/(pd)=E(\tau_{ij})$ which implies that $\lim_{p\rightarrow\infty}\left(\tau/\Gamma \right)^2 =1$ $\Pi$-a.s.
		Hence the proof.
	\end{proof}
	

	\section{Proof of Theorem \ref{th:main_result} and Associated Results}
	\label{sm subsec:consistency_theorem}
	To prove Theorem \ref{th:main_result}, we consider an adaptation of Theorem 6.39 in \citet{ghosal_book} where instead of having an increasing sample size, we assume an increasing data dimension with fixed sample size.
	This notion is consistent with the idea that  more and more variables are measured on each study subject. 
	We introduce the following notation.
	Let $\vartheta=(\Lambda, {\boldeta},\sigma)$ with ${  \boldeta}=\left[\eta_1,\dots,\eta_n\right]\trans$ and $\vartheta \in \Theta_p$.
	Let $\ptheta$ and $\p$ be the joint distributions of the data $y_1,\dots,y_n$ given $\vartheta$ and $\vartheta_{0}$,  respectively, with $\vartheta_ 0 = (\Lambda_{0},{ \boldeta}_{0},\sigma_{0})$.  
	We also denote the expectation of a function $g$ with respect to $\p$ and $\ptheta$ by $\p g$ and $\ptheta g$ respectively.
	Let $p_{0}^{p}$ and $p_{\vartheta}^{p}$ be the  densities of $\p$ and $\ptheta$ with respect to the Lebesgue measure. Finally, define the Kullback-Leibler (KL) divergence and the $r$-th order positive KL-variation between $p_{0}^p$ and $p_\vartheta^p$, respectively, as $KL(\p,\ptheta)= \int \log \frac{p_{0}^p}{p_\vartheta^p} \de\p$ and $V_{r}^+(\p,\ptheta)=\int \left\{  \left(\log \frac{p_{0}^p}{p_\vartheta^p} -KL\right) ^{+}\right\} ^{r}  \de\p$,
	where $f^{+}$ denotes the positive part of a function $f$.
	
	\begin{theorem}\label{th:2}
		If for some $r \geq 2$, $c > 0$ there exist measurable sets $B_p \subset \Theta_p$ with $\liminf \Pi(B_p)>0$, 
		\begin{enumerate}[label={(\Roman*)}]
			\item\label{cond1} 
			$\sup_{\vartheta\in B_p} \frac{1}{p} KL(\p,\ptheta) \leq c$ and $\sup_{\vartheta\in B_p} \frac{1}{p^r} V_r^+(\p,\ptheta) \to 0$.
			\item\label{cond2} For sets $ \tilde{\Theta}_p \subset \Theta_p$ there exists a sequence of test functions $\phi_p$ such that  $\phi_p\to0$ $\p$-a.s. and $\int_{\tilde{\Theta}_p} \ptheta (1-\phi_p)\de\Pi(\vartheta)\leq e^{-Cp} $ for some $C>0$.
			\item\label{cond3} Letting $A_p=\bigg\{\vartheta\in\Theta_p:\frac{1}{p}\int\Big(\log \frac{p_{0}^p}{p_\vartheta^p}-KL(\p,\ptheta)\Big)\de\tilde{\Pi}_p(\vartheta) <\tilde\epsilon \bigg\}$, 
			with $\truncpi$ the renormalized restriction of $\Pi$ to set $B_p$, for any $\tilde\epsilon>0$, $\mathbbm{1} (\bar{A}_{p}) \to 0 ~ \p$-a.s.
		\end{enumerate}		
		Then $\Pi( \tilde{\Theta}_p\mid  \by) \to 0$ $\p$-a.s.
	\end{theorem}	
	
	Condition \ref{cond1} ensures that the assumed model is not too far from the true data-generating model. 
	Condition \ref{cond2} controls the variability of the log-likelihood around its mean. 
	In the Lamb model, the number of parameters grows with $p$ and hence 
	the assumption on $V_r^+$ is instrumental. 
	The conditions on $\phi_{p}$ ensure the existence of a sequence of consistent test functions for $H_{0}:\mathbb{P}=\p$ in which type-II error diminishes to 0 exponentially fast in the critical region. 
	Condition \ref{cond3} is a technical condition required to bound the numerator of $\Pi( \tilde{\Theta}_{p}\mid  \by)$.	
	The proof of this theorem follows along the lines of the proof of Theorem 6.39 of \citet{ghosal_book}.
	
	Theorem \ref{th:2} is a general result stating sufficient conditions for posterior consistency as $p\to \infty$. 
	We 	use this theorem to prove Theorem \ref{th:main_result}.
	
	\vspace{.2in}

	\begin{proof}[Theorem \ref{th:main_result}]
		We verify the conditions \ref{cond1}-\ref{cond3} from Theorem \ref{th:2}.
		Theorems \ref{th:KL_support} and \ref{th:ghoshal_cond} jointly imply that 
		for the Lamb model there exist a sequence of sets $B_p$ such that conditions
		\ref{cond1} and \ref{cond3} are satisfied for any $c>0$.
		Theorem \ref{th:test_function} ensures the existence of a sequence of test functions satisfying \ref{cond2}, and finally Theorem \ref{th:last_cond} proves \ref{cond3}.
		Hence the proof.		
	\end{proof}

	{
	}
	
	\begin{theorem}
		\label{th:KL_support}
		For any $\epsilon>0$ define $B_p^\epsilon=\left\{\Theta: p^{-1}KL(\p, \ptheta) \leq \epsilon   \right\}$. Then, under the settings of Section \ref{sec:properties}, $\liminf \Pi(B_p^\epsilon)>0$.
	\end{theorem}
	\begin{proof}
		Let $P_{0}$ and $P$ be $p$-variate multivariate normal distributions with $P\overset{}{=}\mn_p(\mu,\Sigma)$ and $P_{0}\overset{}{=}\mn_p(\mu_{0},\Sigma_{0})$. Then their Kullback-Leibler divergence is $KL(P_{0}, P) = \frac{1}{2} \{ \log\frac{\abs{\Sigma}}{\abs{\Sigma_{0}}}+ \tr\left(\Sigma^{-1}\Sigma_{0} \right) +(\mu-\mu_{0})\trans\Sigma^{-1}(\mu-\mu_{0})-p \}$,		
		which, under  the settings of Section \ref{sec:properties}, simplifies to
		\begin{equation}
			KL(\p , \ptheta) = \frac{1}{2} \left\{ np\log\frac{\sigma^2}{\sigma_{0}^2}+ np \left(\frac{\sigma_{0}^2}{\sigma^2}-1 \right) + \frac{1}{\sigma^2}\sum_{i=1}^{n} \norm{\mu_{i}-\mu_{0i}}^{2} \right\},
		\end{equation}
		where $\mu_i=\Lambda\eta_i$ and $\mu_{0i}=\Lambda_{0}\eta_{0i}$.
		Now, 
		\begin{align*}
			\Pi\left\{ p^{-1}KL(\p , \ptheta)< \epsilon \right\} & =    \Pi \left\{
			n\log\frac{\sigma^2}{\sigma_{0}^2}+ n \left(\frac{\sigma_{0}^2}{\sigma^2}-1 \right) + \frac{1}{p \sigma^2}\sum_{i=1}^{n} \norm{\mu_{i}-\mu_{0i}}^{2}  < \epsilon
			\right\}\\
			& \geq \Pi \left\{\log\frac{\sigma^2}{\sigma_{0}^2}+  \left(\frac{\sigma_{0}^2}{\sigma^2}-1 \right)\leq \frac{\epsilon}{2n},\frac{1}{\sigma^2}\sum_{i=1}^{n} \norm{\mu_{i}-\mu_{0i}}^{2} < \frac{p\epsilon}{2} \right\}.
		\end{align*}
		Note that for any $x>0$, $\log x\leq x-1$ and therefore  $\log\frac{\sigma^2}{\sigma_{0}^2}+  \left(\frac{\sigma_{0}^2}{\sigma^2}-1\right)\leq \left(\frac{\sigma_{0}}{\sigma}-\frac{\sigma}{\sigma_{0}} \right)^2$ implying that
		\begin{align*}
			\Pi\left\{ p^{-1} KL(\p , \ptheta) < \epsilon \right\} &
			\geq 
			\Pi \left\{
			\left(\frac{\sigma_{0}}{\sigma}-\frac{\sigma}{\sigma_{0}} \right)^2\leq \frac{\epsilon}{2n} ,  \frac{1}{\sigma^2}\sum_{i=1}^{n} \norm{\mu_{i}-\mu_{0i}}^{2}< \frac{p\epsilon}{2}
			\right\} \\
			&\geq \Pi\left\{\left(\frac{\sigma_{0}}{\sigma}-\frac{\sigma}{\sigma_{0}} \right)^2\leq\frac{\epsilon}{2n} \right\} \Pi\left( \sum_{i=1}^{n} \norm{\mu_{i}-\mu_{0i}}^{2}< \sigma_L \frac{p\epsilon}{2} \right) ,
		\end{align*}
		where the second inequality holds thanks to condition \ref{ass3}. 
		The first factor above is positive under our proposed prior on $\sigma$. Now consider the second factor and note that for each $i = 1, \dots, n$, 
		$\norm{\mu_{i}-\mu_{0i}}^{2}= \norm{\Lambda (\eta_i- (\Lambda\trans\Lambda)^{-1}\Lambda\trans\Lambda_{0}\eta_{0i} )}^2 + \eta_{0i}\trans (\Lambda_{0}\trans\Lambda_{0}- \Lambda_{0}\trans\Lambda(\Lambda\trans\Lambda)^{-1}\Lambda\trans\Lambda_{0} )\eta_{0i}$.		
		By the triangle inequality
		\begin{equation}
			\frac{1}{p}\norm{ \Lambda_{0}\trans\Lambda_{0}- \Lambda_{0}\trans\Lambda(\Lambda\trans\Lambda)^{-1}\Lambda\trans\Lambda_{0} }_2\leq  \norm{\frac{1}{p}\Lambda_{0}\trans\Lambda_{0}- M}_2  +\norm{M-\frac{1}{p} \Lambda_{0}\trans\Lambda(\Lambda\trans\Lambda)^{-1}\Lambda\trans\Lambda_{0}}_2. \label{eq:norm_bound}
		\end{equation}
		The first term on the right hand side of \eqref{eq:norm_bound}  goes to 0 as $p\to\infty$ by \ref{ass1}. Let us define the matrix $B= (\Lambda\trans\Lambda)^{-\sfrac{1}{ 2}}\Lambda\trans$, 
		with $\norm {B}_2=1$ and $\tilde\Lambda_{0}=\Lambda_{0} M^{-\sfrac{1}{2}}$ where $M^{\sfrac{1}{2}}$ is the Cholesky factor of $M$. From \citet[Theorem 5.39]{vershynin_2012} it follows that for any $0<\epsilon<1$ and large enough $p$, $1 -\epsilon \leq \norm {\frac{1}{\sqrt{p}}\tilde\Lambda_{0}}_2 \leq 1 +\epsilon$.
		Again, from Lemma 1.1 of the Supplement section of \citet{pati2014} we have that $1 -\epsilon\leq \norm {\frac{1}{\sqrt{p}}B\tilde\Lambda_{0}}_2 \leq 1 +\epsilon,\quad 
		1 -\epsilon \leq\frac{1}{\sqrt{p}} s_{\min} (B\tilde\Lambda_{0}) \leq 1 +\epsilon$.  
		Therefore $\lim_{p\rightarrow\infty}\lambda_i( \tilde\Lambda_{0}\trans\Lambda(\Lambda\trans\Lambda)^{-1}\Lambda\trans\tilde\Lambda_{0})p^{-1}=1$ for all $i=1,\dots,d$. Now $\norm{M-\frac{1}{p} \Lambda_{0}\trans\Lambda(\Lambda\trans\Lambda)^{-1}\Lambda\trans\Lambda_{0}}_2 =\norm{M}_2 \norm{I_d-\frac{1}{p} \tilde\Lambda_{0}\trans\Lambda(\Lambda\trans\Lambda)^{-1}\Lambda\trans\tilde\Lambda_{0}}_2$ and therefore the second term on the right hand side of \eqref{eq:norm_bound} goes to 0 as $p\to\infty$.  Subsequently we have
		$\lim_{p\rightarrow\infty} \frac{1}{p}\eta_{0i}\trans (\Lambda_{0}\trans\Lambda_{0}- \Lambda_{0}\trans\Lambda(\Lambda\trans\Lambda)^{-1}\Lambda\trans\Lambda_{0} )\eta_{0i}=0$ for all $i=1,\dots,n$. Now \ref{ass1} and Lemma \ref{th:DL_eig_bound} jointly imply that 
		$\norm {(\Lambda\trans\Lambda)^{-1}\Lambda\trans\Lambda_{0}}_2=O(1)$ $\Pi$-a.s. Therefore, for standard normal priors on the latent variables, 
		$\liminf_{p\rightarrow\infty} \Pi(  \sum_{i=1}^{n} \norm{\mu_i-\mu_{0i}}^{2}< \sigma_{L} p\epsilon) > 0$.
		From the permanence of KL-property of mixture priors  \citep[Proposition 6.28]{ghosal_book} we can conclude that the right hand side is also positive.	\end{proof}

	\begin{theorem}
		\label{th:ghoshal_cond}
		On the set $B_p^\epsilon$ defined in Theorem \ref{th:KL_support},  we have $V_r^+(\p,\ptheta)=o(p^r)$ for $r=2$.
	\end{theorem}
	\begin{proof}
		For $r=2$, $V_r^+(\p,\ptheta) \leq \int \log^2 \frac{p_{0}^p}{p_\vartheta^p} d\p  - \left\{\int  \log\frac{p_{0}^p}{p_\vartheta^p} d\p \right\}^2$.		
		Now conditionally on  $\vartheta \in \vartheta$, the observations $y_1,\dots,y_n$ are independent. Therefore,
		\begin{equation}
			V_r^+(\p,\ptheta) \leq\sum_{j=1}^n \left[\int \left\{\log\frac{p_{0j}(y_j)}{p_{\vartheta_j}(y_j)}\right\}^{2} p_{0j}(y_j) \de y_j - \left\{\int  \log\frac{p_{0j}(y_j)}{p_{\vartheta_j}(y_j)} p_{0j}(y_j)\de y_j\right\}^2 \right]  \label{eq:KL_var_sum}
		\end{equation}
		where $p_{0j}(y_j)=\prod_{i=1}^p \mn(y_{ji}; \mu_{0ji}, \sigma_{0}^2 )$ and $p_{\vartheta_j}(y_j)=\prod_{i=1}^p \mn(y_{ji}; \mu_{ji}, \sigma^2 )$ with $\mu_{0j}= $ and $\mu_{j}= \Lambda\eta_{j}$. We first show the result for a particular term inside the summation of \eqref{eq:KL_var_sum}. Since $\norm{\eta_{0i}}=O(1)$ and $n$ is fixed, the result will readily follow afterwards. For simplicity, we drop the suffix $j$ from the terms of \eqref{eq:KL_var_sum} henceforth.
		Consider,
		\begin{multline*}
			\left\{\log\frac{p_{0}(y_i)}{p_\vartheta(y_i)}\right\}^{2}  =\left[\log\frac{\sigma}{\sigma_{0}}-\frac{1}{2}\left\{\left( \frac{y_i-\mu_{0i}}{\sigma_{0}}\right)^2-\left(\frac{y_i-\mu_{i}}{\sigma} \right)^2\right\} \right]^2\\
			=\frac{1}{4}\left\{\left( \frac{y_i-\mu_{0i}}{\sigma_{0}}\right)^2 -\left(\frac{y_i-\mu_{i}}{\sigma} \right)^2\right\}^2 +\log^2 \frac{\sigma}{\sigma_{0}}- \left\{\left( \frac{y_i-\mu_{0i}}{\sigma_{0}}\right)^2 -\left(\frac{y_i-\mu_{i}}{\sigma} \right)^2\right\}  \log\frac{\sigma}{\sigma_{0}}.
		\end{multline*}
		Note that,
		\begin{multline*}
			\left\{\left( \frac{y_i-\mu_{0i}}{\sigma_{0}}\right)^2 -\left(\frac{y_i-\mu_{i}}{\sigma} \right)^2\right\}^2 =\left\{z_i^2\left(1-\frac{\sigma_{0}^2}{\sigma^2}\right) -2z_i(\mu_{0i}-\mu_i)\frac{\sigma_{0}}{\sigma}+ \left( \frac{\mu_i-\mu_{0i}}{\sigma}\right)^2\right\}^2\\
			=z_i^4\left(1-\frac{\sigma_{0}^2}{\sigma^2}\right)^2+4z_i^2\sigma_{0}^2\left(\frac{\mu_{0i}-\mu_i}{\sigma}\right)^2 + \left( \frac{\mu_i-\mu_{0i}}{\sigma}\right)^4 - 2z_i^3\left(1-\frac{\sigma_{0}^2}{\sigma^2}\right)\frac{\sigma_{0}}{\sigma}(\mu_{0i}-\mu_i)\\- 2z_i\sigma_{0} \left(\frac{\mu_{0i}-\mu_i}{\sigma} \right)^3  +2z_i^2\left(\frac{\mu_{0i}-\mu_i}{\sigma} \right)^2\left(1-\frac{\sigma_{0}^2}{\sigma^2}\right)
		\end{multline*}
		where $z_i=(y_i-\mu_{0i})/\sigma_{0}$ and $z_i\simiid \mn(0,1)$. Therefore,
		\begin{equation*}
			E_{y_i}\left\{\left( \frac{y_i-\mu_{0i}}{\sigma_{0}}\right)^2 -\left(\frac{y_i-\mu_{i}}{\sigma} \right)^2\right\}=\left(1-\frac{\sigma_{0}^2}{\sigma^2}\right)+ \left( \frac{\mu_i-\mu_{0i}}{\sigma}\right)^2 \text{ and}
		\end{equation*}
		\begin{multline*}
			E_{y_i}\left\{\left( \frac{y_i-\mu_{0i}}{\sigma_{0}}\right)^2 -\left(\frac{y_i-\mu_{i}}{\sigma} \right)^2\right\}^2=3\left(1-\frac{\sigma_{0}^2}{\sigma^2}\right)^2+4\sigma_{0}^2\left(\frac{\mu_{0i}-\mu_i}{\sigma}\right)^2+ \left( \frac{\mu_i-\mu_{0i}}{\sigma}\right)^4 \\
			+2\left(\frac{\mu_{0i}-\mu_i}{\sigma} \right)^2\left(1-\frac{\sigma_{0}^2}{\sigma^2}\right).
		\end{multline*}
		Hence,
		\begin{multline*}
			\int \left\{ \log\frac{p_{0}(y_i)}{p_\vartheta(y_i)}\right\}^2p_{0}(y_i)\de y_i=\left(\frac{\mu_{0i}-\mu_i}{\sigma}\right)^2 \times \left\{\sigma_{0}^2+\frac{1}{2}\left(1-\frac{\sigma_{0}^2}{\sigma^2}\right)-\log\frac{\sigma}{\sigma_{0}} \right\}\\-\log\frac{\sigma}{\sigma_{0}} \left(1-\frac{\sigma_{0}^2}{\sigma^2}\right)+\frac{1}{4}\left(\frac{\mu_{0i}-\mu_i}{\sigma}\right)^4+\frac{3}{4}\left(1-\frac{\sigma_{0}^2}{\sigma^2}\right)^2+\log^2\frac{\sigma}{\sigma_{0}}
		\end{multline*}
		\begin{multline*}
			\left\{ \int  \log\frac{p_{0}(y_i)}{p_\vartheta(y_i)}p_{0}(y_i)\de y_i\right\}^2= \left\{\log \frac{\sigma}{\sigma_{0}} + \frac{\sigma_{0}^2 + (\mu_{0i} - \mu_i)^2}{2 \sigma^2} - \frac{1}{2}\right\}^2=\log^2 \frac{\sigma}{\sigma_{0}}+\frac{1}{4} \left(1-\frac{\sigma_{0}^2}{\sigma^2}\right)^2\\
			+\frac{1}{4}\left(\frac{\mu_{0i}-\mu_i}{\sigma}\right)^4+\left(\frac{\mu_{0i}-\mu_i}{\sigma}\right)^2 \times \left\{\log\frac{\sigma}{\sigma_{0}} -\frac{1}{2}\left(1-\frac{\sigma_{0}^2}{\sigma^2}\right)\right\}-\log\frac{\sigma}{\sigma_{0}} \left(1-\frac{\sigma_{0}^2}{\sigma^2}\right),
		\end{multline*}
		leading to
		\begin{multline}
			V_r^+(\p,\ptheta)\leq \sum_{i=1}^p\left[ \int \left\{ \log\frac{p_{0}(y_i)}{p_\vartheta(y_i)}\right\}^2p_{0}(y_i)\de y_i-   \left\{ \int  \log\frac{p_{0}(y_i)}{p_\vartheta(y_i)}p_{0}(y_i)\de y_i\right\}^2\right]\\=\frac{p}{2}\left(1-\frac{\sigma_{0}^2}{\sigma^2}\right)^2+ \left\{\sigma_{0}^2-2\log\frac{\sigma}{\sigma_{0}} +\left(1-\frac{\sigma_{0}^2}{\sigma^2}\right)\right\}\times \sum_{i=1}^p \left(\frac{\mu_{0i}-\mu_i}{\sigma}\right)^2.  \label{eq:KL_var_pow}
		\end{multline}
		Note that
		\begin{align}
			\sum_{i=1}^p (\mu_{0i}-\mu_i)^2=\sum_{i=1}^p \left( \lambda_{0i}\trans\eta_{0}-\lambda_{i}\trans\eta \right)^2=\eta_{0}\trans \Lambda_{0}\trans\Lambda_{0}\eta_{0}+\eta\trans \Lambda\trans\Lambda\eta-2\eta_{0}\trans \Lambda_{0}\trans\Lambda\eta. \label{eq:KL_var1}
		\end{align}
		Now $\eta_{0}\trans \Lambda_{0}\trans\Lambda_{0}\eta_{0}\leq \norm{\Lambda_{0}}_2^2  \norm{\eta_{0}}^2$ and therefore, by conditions \ref{ass1} and \ref{ass4}, $\eta_{0}\trans \Lambda_{0}\trans\Lambda_{0}\eta_{0}=O(p) $.
		Also from Lemma \ref{th:DL_eig_bound}, $\frac{1}{p}\norm{\Lambda}_2^2 \leq c$ for large enough $p$ and some $c>0$  and therefore $\eta\trans\Lambda\trans\Lambda\eta \leq \norm{\Lambda}_2^2 \norm{\eta}^2= \norm{\eta}^2 O(p)$. From the proof of Theorem \ref{th:KL_support} we can see that in the set $B_p^\epsilon$, $\norm{\eta}$ is bounded.
		We have shown that the highest powers in \eqref{eq:KL_var1} and thus in \eqref{eq:KL_var_pow} are almost surely bounded  by $p$ for large enough $p$. Hence the proof. 
	\end{proof}

	\begin{theorem}
		\label{th:test_function}
		Let us define the test function 
		$\phi_{p}=\mathbbm{1}\left\{\abs{\frac{1}{\sqrt{np}\sigma_{0}} \norm{ \sum_{i=1}^n (y_{i}-\Lambda_{0}\eta_{0i})}-1}   > \tau \right\}$
		to test the following hypothesis
		$H_{0}: y_{1},\dots, y_{n}\sim \p$ versus $H_1:H_{0}$ is false where $\tau$ is a positive real number. 
		Define the set 
		$\tilde{\Theta}_p=\bar{B}_{p,\delta }$. 
		Then there exists a constant $C>0$ such that $\phi_{p}\rightarrow0$ $\p$-a.s. and $\int_{\tilde{\Theta}_p} \ptheta(1-\phi_{p})\de\Pi(\vartheta)\leq e^{-Cp}$.
	\end{theorem}
	\begin{proof}
		Let us define $\mu_i=\Lambda\eta_i$ and $\mu_{0i}=\Lambda\eta_{0i}$. Then under $H_{0}$, $\frac{1}{\sqrt{n}\sigma_{0}} \sum_{i=1}^n(y_i-\Lambda\eta_{0i}) \sim \mn_p(0,I_p )$ and therefore $\frac{1}{\sqrt{np}\sigma_{0}} \sum_{i=1}^n(y_i-\Lambda\eta_{0i}) \overset{d}{=} \omega/\sqrt{p}$ where $\omega\sim \mn_p(0,I_p)$. Then from \citet[Theorem 2.1]{Rudelson13} for some $c>0$ and any $\tau_{np}>0$  $\p\phi_{p}=\Pr\left(\, \abs{\frac{1}{\sqrt{p}}\norm{ \omega}-1 }>\tau_{np} \right) \leq 2\exp\left(-pc\tau^2_{np} \right)$.
		Since $\sum_{p=1}^\infty \p\phi_{p}<\infty$, by Borel-Cantelli lemma $\phi_p\to0$ $\p$-a.s.
		
		Notably when $H_{0}$ is not true i.e. under $\ptheta$, $Y_i\overset{d}{=}\sigma\varphi_i+\Lambda\eta_i$ where $\varphi_i\simiid\mn_p(0,I_p)$ for some $\vartheta\neq(\Lambda_{0},\eta_{0},\sigma_{0})$ and therefore under $\ptheta$
		\begin{align}
			&\ptheta(1-\phi_p)\leq \Pr\left\{\frac{1}{\sqrt{pn}\sigma_{0}}\norm{\sum_{i=1}^n(\sigma\varphi_i+\Lambda\eta_i-\Lambda_{0}\eta_{0i})}<1+\tau_{np} \right\} \notag\\
			\leq& \Pr\left\{\frac{1}{\sqrt{pn}\sigma_{0}}\sum_{i=1}^n\norm{\Lambda\eta_i-\Lambda_{0}\eta_{0i}}-1-\tau_{np}-\frac{\sigma}{\sigma_{0}}\leq \frac{\sigma}{\sigma_{0}}\left(\frac{1}{\sqrt{np}}\sum_{i=1}^n\norm{\varphi_i} -1\right)  \right\}.\label{eq:prob}
		\end{align}
		Notably for $\vartheta\in\tilde \Theta_{p}$, $\frac{1}{\sqrt{pn}\sigma_{0}}\sum_{i=1}^n\norm{\Lambda\eta_i-\Lambda_{0}\eta_{0i}}$ is unbounded above for increasing $p$ and $\frac{\sigma}{\sigma_{0}}$ is bounded thanks to \ref{ass3}.
		Letting $C_p=\frac{1}{\sqrt{pn}\sigma_{0}}\sum_{i=1}^n\norm{\Lambda\eta_i-\Lambda_{0}\eta_{0i}}-1-\tau_{np}-\frac{\sigma}{\sigma_{0}}$ we have $\liminf_{p\rightarrow\infty}C_p>0$. Therefore, from \citet[Theorem 2.1]{Rudelson13},  we have for $\vartheta\in\tilde \Theta_{p}$,
		$\ptheta(1-\phi_{p})\leq 2\exp\left(-pnc C^2_{p} \right)$. Hence the proof.
	\end{proof}

	\begin{theorem}
		\label{th:last_cond}
		Let $\truncpi$ be the renormalized restriction of $\Pi$ to the set $B_p^\epsilon$ defined in Theorem \ref{th:KL_support}. Then	$\mathbbm{1}\{\bar{A}_p\} \to 0 ~ \p$-a.s.
	\end{theorem}
	\begin{proof}
		If we can show that $\sum_{p=1}^\infty \p(\bar{A}_p)<\infty$, then by Borel-Cantelli lemma $\p[\limsup \bar{A}_p]=0$ and henceforth $\mathbbm{1}\{\bar{A}_p\} \to 0 ~ \p$-a.s. Now
		\begin{align*}
			\p(\bar{A}_p) 
			=\p\left[\frac{1}{p}\int \hspace{-0.1cm} \sum_{i=1}^n\left\{
			\frac{1}{\sigma^2}\norm{y_i\hspace{-0.1cm} - \hspace{-0.1cm}\mu_i}^2 -
			\frac{1}{\sigma_{0}^2} \norm{y_i\hspace{-0.1cm} - \hspace{-0.1cm}\mu_{0i}}^2 -
			\frac{1}{\sigma^2}\norm{\mu_i\hspace{-0.1cm} - \hspace{-0.1cm}\mu_{0i}}^2  -
			p\left(\frac{\sigma^2_{0}}{\sigma^2}\hspace{-0.1cm} - \hspace{-0.1cm}1\right) \right\}\de\truncpi >2\epsilon\right].
		\end{align*}
		Notably under $\p$, $Y_i\overset{d}{=}\sigma_{0}\varphi_i+\mu_{0i}$ where $\varphi_i\simiid\mn_p(0,I_p)$. Therefore	
		\begin{align}
			\p(\bar{A}_p)
			=&\Pr\left[\frac{1}{p}\int\sum_{i=1}^n\left\{\left(\frac{\sigma^2_{0}}{\sigma^2}-1\right)(\norm{\varphi_i}^2-p) +2 \frac{\sigma_{0}}{\sigma^2}\varphi_i\trans(\mu_i-\mu_{0i})   \right\}\de\truncpi>2\tilde\epsilon\right]\nonumber\\
			\leq&\Pr\left[\frac{1}{p}\sum_{i=1}^n (\norm{\varphi_i}^2-p) \int\left(\frac{\sigma^2_{0}}{\sigma^2}-1\right) \de\truncpi>\tilde\epsilon\right]+\notag \\
			& \phantom{ABCABCABCABCABC} \Pr\left[\frac{2}{p}\int\sum_{i=1}^n\left\{ \frac{\sigma_{0}}{\sigma^2}\varphi_i\trans(\mu_i-\mu_{0i})   \right\}\de\truncpi>\tilde\epsilon\right].\label{eq:bclemma_parts}
		\end{align}
		Let us consider the first term of \eqref{eq:bclemma_parts}. Notably 
		\begin{equation}
			\Pr\left[\frac{1}{p}\sum_{i=1}^n (\norm{\varphi_i}^2-p) \int\left(\frac{\sigma^2_{0}}{\sigma^2}-1\right) \de\truncpi>\tilde\epsilon\right]\leq \Pr\left[\frac{1}{p}\abs{\sum_{i=1}^n (\norm{\varphi_i}^2-p)}  \int\abs{\frac{\sigma^2_{0}}{\sigma^2}-1} \de\truncpi>\tilde\epsilon\right].\label{eq:1stpart}
		\end{equation}
		From \ref{ass3} we have that $\sigma$ lies in a compact interval. Hence the integral in the right hand side of \eqref{eq:1stpart} is bounded above by some positive constant, say $C_{\sigma,1}$. Therefore,
		\begin{align*}
			\Pr\left[\frac{1}{p}\sum_{i=1}^n (\norm{\varphi_i}^2-p) \int\left(\frac{\sigma^2_{0}}{\sigma^2}-1\right) \de\truncpi>\tilde\epsilon\right]
			\leq \Pr\left[\frac{1}{p}\abs{\sum_{i=1}^n (\norm{\varphi_i}^2-p)} >\frac{\tilde\epsilon}{C_{\sigma,1}}\right]\leq 2 e^{-pC_{\sigma,2}}
		\end{align*}
		for some positive constant $C_{\sigma,2}>0$. The second inequality in the above equation follows from \citet[Theorem 2.1]{Rudelson13}. Clearly
		\begin{equation}
			\sum_{p=1}^\infty \Pr\left[\frac{1}{p}\sum_{i=1}^n (\norm{\varphi_i}^2-p) \int\left(\frac{\sigma^2_{0}}{\sigma^2}-1\right) \de\truncpi>\tilde\epsilon\right]<\infty.\label{eq:1stpart_sum}
		\end{equation}
		Now we consider the second term of \eqref{eq:bclemma_parts}. As $\varphi_i=(\varphi_{i1},\dots,\varphi_{ip})\trans$ (similarly $\mu_i$ and $\mu_{0i}$ are also $p$-dimensional vectors) we can write
		\begin{align*}
			\Pr\left[\frac{2}{p}\int\sum_{i=1}^n\left\{ \frac{\sigma_{0}}{\sigma^2}\varphi_i\trans(\mu_i-\mu_{0i})   \right\}\de\truncpi>\tilde\epsilon\right]&=\Pr\left[\frac{2}{p}\sum_{i=1}^n\sum_{j=1}^p\varphi_{ij}\int \left\{ \frac{\sigma_{0}}{\sigma^2}(\mu_{ij}-\mu_{0ij})   \right\}\de\truncpi>\tilde\epsilon \right]\\
			&\leq \exp\left[-\frac{p^2\tilde\epsilon^2}{4\sigma_{0}^2\sum_{i=1}^n\sum_{j=1}^pE^2_{\truncpi} \left\{ \frac{1}{\sigma^2}(\mu_{ij}-\mu_{0ij})   \right\}}\right],
		\end{align*}
		where $E_{\truncpi}$ denotes the expectation with respect to the probability measure $\truncpi$. The above inequality follows from sub-Gaussian concentration bounds. Now
		\begin{align}
			\sum_{i=1}^n\sum_{j=1}^p E^2_{\truncpi} \left\{ \frac{1}{\sigma^2}(\mu_{ij}-\mu_{0ij}) \right\} &\leq \sum_{i=1}^n  E_{\truncpi}\frac{1}{\sigma^4}  \norm{ \mu_{i}-\mu_{0i} }^2\mbox{ (by Jensen's inequality)}\nonumber\\
			=&E_{\truncpi}\frac{1}{\sigma^4} \sum_{i=1}^n  \times E_{\truncpi}  \norm{ \mu_{i}-\mu_{0i} }^2. \label{eq:equality}
		\end{align}
		Since we consider independent priors on $\sigma,\Lambda$ and $\eta_i$, \eqref{eq:equality} follows from its preceding step. Note that on the set $B_p^\epsilon$ 	
		\vskip-4ex	
		\begin{equation}
			n\log\frac{\sigma^2}{\sigma_{0}^2}+ n \left(\frac{\sigma_{0}^2}{\sigma^2}-1 \right) + \frac{1}{p\sigma^2}\sum_{i=1}^{n} \norm{\mu_i-\mu_{0i}}^2 <2\epsilon.\label{eq:B_p_set}
		\end{equation}
		\vskip-1ex
		\noindent From the inequality $\log x<x-1$ we see that $n\log\frac{\sigma^2}{\sigma_{0}^2}+ n \left(\frac{\sigma_{0}^2}{\sigma^2}-1 \right)>0$. Therefore for $\vartheta\in B_p^\epsilon$, in conjunction of \eqref{eq:B_p_set} and \ref{ass3} we have
		$\frac{1}{p}\sum_{i=1}^{n} \norm{\mu_i-\mu_{0i}}^2 <2\epsilon\sigma_U^2\Rightarrow \frac{1}{p} \sum_{i=1}^n  E_{\truncpi}  \norm{ \mu_{i}-\mu_{0i} }^2<2\epsilon\sigma_U^2$. Also thanks to  \ref{ass3} $E_{\truncpi}\frac{1}{\sigma^4}$ is bounded above. Hence the term in \eqref{eq:equality} is bounded above and consequently
		\vskip-2ex
		\begin{equation}
			\sum_{p=1}^\infty\Pr\left[\frac{2}{p}\int\sum_{i=1}^n\left\{ \frac{\sigma_{0}}{\sigma^2}\varphi_i\trans(\mu_i-\mu_{0i})   \right\}\de\truncpi>\tilde\epsilon\right]<\infty. \label{eq:2ndpart_sum}
		\end{equation}
		\vskip-.5ex
		\noindent Combining \eqref{eq:1stpart_sum} and \eqref{eq:2ndpart_sum} we conclude that
		$\sum_{p=1}^\infty\p(\bar{A}_p)<\infty.$ Hence the proof.	
	\end{proof}
	
	\section{Details on Simulation Studies}
	\label{SM_sec:simulation_details}
	In this section, we discuss the data-generation strategies in three simulation scenarios: [1] Lamb, [2] mixture of sparse factor analyzers (MFA), and [3] mixture of log transformed zero inflated Poisson counts (SpCount) considered in Section \ref{sec:simulation} of the main manuscript.
	The observed $p$-dimensional data are $y_{1},\dots,y_{n}$,
	$k_{0}$ is the number mixture components in the simulation truth
	and $\pi_{1},\dots, \pi_{k_{0}}$ are the mixture probabilities attached to each cluster such that $\sum_{h=1}^{k_{0}} \pi_{h}=1$.
	\begin{description}	[leftmargin=0pt]
		\item[Lamb] We let the observed data 
		\vskip-3ex
		\begin{equation*}
			y_{i}= \Lambda \eta_{i} +\epsilon_{i}, \quad \eta_{i}\simiid \textstyle{ \sum_{h=1}^{k_{0}} \pi_{h} \mn_{d} ( \mu_{h},  \Delta_{h} )},\quad \epsilon_{i} \simiid\mn_{p}(0, \Sigma),
		\end{equation*}
		\vskip-2ex
		
		where $\Lambda$ is a $p\times d$ order sparse matrix with many entries equal to zero, 
		$\mu_{h} \in \R^{d} $, $\Delta_{h}$ is a  $d\times d$ positive definite matrix, for all $h=1,\dots, k_{0}$
		and $\Sigma$ is a $p\times p$ order diagonal matrix with positive entries.
		
		\item[Mixture of sparse factor analyzers (MFA)] We let the observed data 
		\vskip-3ex
		\begin{equation*}
			y_{i} \simiid \textstyle{ \sum_{h=1}^{k_{0}} \pi_{h} \mn_{p} (  \mu_{h},  \Lambda_{h}\Lambda_{h}\trans +\Sigma_{h})},
		\end{equation*}
		\vskip-2ex
		where $\Lambda_{h}$ is a $p\times d$ order sparse matrix with many entries equal to zero, 
		$\Sigma_{h}$ is a $p\times p$ diagonal matrix with positive entries
		and $\mu_{h} \in \R^{p} $, for all $h=1,\dots, k_{0}$.
		
		\item[Mixture of log transformed zero inflated sparse Poisson counts (SpCount)] 
		Let $\{\ell_{1},\dots, \ell_{p} \}$ be a random permutation of $1,\dots,p$, $r=\lfloor p/k_{0} \rfloor$ and 
		define the set $S_{h}=\{\ell_{(h-1)\times r+1},\dots, \ell_{h\times r}\}$ for all $h=1,\dots,k_{0}$.
		Thus $\{S_{1},\dots, S_{k_{0}}\}$ can be regarded as a random partition of  $\{1,\dots,p\}$ where each partition has $r$ elements.
		Additionally fix $k_{0}$ positive constants $\lambda_{1},\dots, \lambda_{k_{0}}$, and let 
		\vskip-4ex	
		\begin{align*}
			&w_{ij} \mid c_{i}=h \simiid \begin{cases}
				\mathrm{Pois}(\lambda_{h})+ \mn(0,1) \text{ for all } j \in S_{h},\\
				0\text{ with probability 1 for all } j \notin S_{h},
			\end{cases}\\
			& \Pr(c_{i}=h) =\pi_{h} \text{ for all } h=1,\dots, k_{0}.
		\end{align*}
		\vskip-3ex
		where $\mathrm{Pois}(\lambda)$ is the Poisson distribution with mean $\lambda$ and set $y_{ij}=\log(w_{ij}+1)$ for all $j=1,\dots,p$ and $i=1,\dots,n$.
		Thus the observed data $y_{i}$'s are highly non-Gaussian within each cluster.
	\end{description}
	
	
	
	\section{Additional Simulation Studies}
	\label{sm sec:simstudy}
	
	\subsection{Illustration of the Degenerate Clustering Behaviour}
	\label{sm subsec:degen_clustering}

	To show the degenerate clustering behavior discussed in Section \ref{sec:pitfall} we performed two simple simulation experiments under the settings of Corollaries \ref{cor:neg_result2} and  \ref{cor:neg_result}.

	In the first experiment, we  generate  data from a five-component mixture model. 
	Specifically, we assumed five well-separated Gaussians with equal proportions. 
	The location vector for the $h$-th component is $\theta_{h} \bone_p$ with $\bone_{p}$ a $p$-dimensional vector of ones, $\theta_h \in \R$ and the values of $\theta_h$ ranging from -10 to +10.  Each mixture component has identity covariance matrix.  We fix $n=10$ and $p=20$.  
	The left panel of Figure \ref{fig:simpleboxplots} displays the distribution of the posterior median number of clusters in 100 replicates for a standard DP location mixture with hyperparameter specification satisfying Corollary~\ref{cor:neg_result2} and proposed Lamb. For the DPM, we use the implementation in the \texttt{BNPmix} package \citeplatex{bnpmix}.
	Despite coming from a five-component mixture model, the data are grouped into a single cluster for most of the simulation replicates under the DPM specification, consistent with the limiting behavior described by Corollary \ref{cor:neg_result2}.
	
	In the second experiment,  we assume  a  single $p$-variate normal distribution with mean zero, and identity covariance. 
	As before, we fix $n=10$ and $p=20$.  
	The results obtained assuming a DP mixture with the  hyperparameter specification satisfying Corollary \ref{coroll2} and the proposed Lamb are reported in the right panel of Figure \ref{fig:simpleboxplots}.  
	These results clearly show that the limiting behavior described by Corollary \ref{coroll2} is evident already for the moderate $p=20$.
	Notably, the proposed Lamb avoids these pitfalls and is associated to  a median number of clusters that is centered around the true values.

	\begin{figure}[h] 
		\centering
		\includegraphics[width=0.9\linewidth]{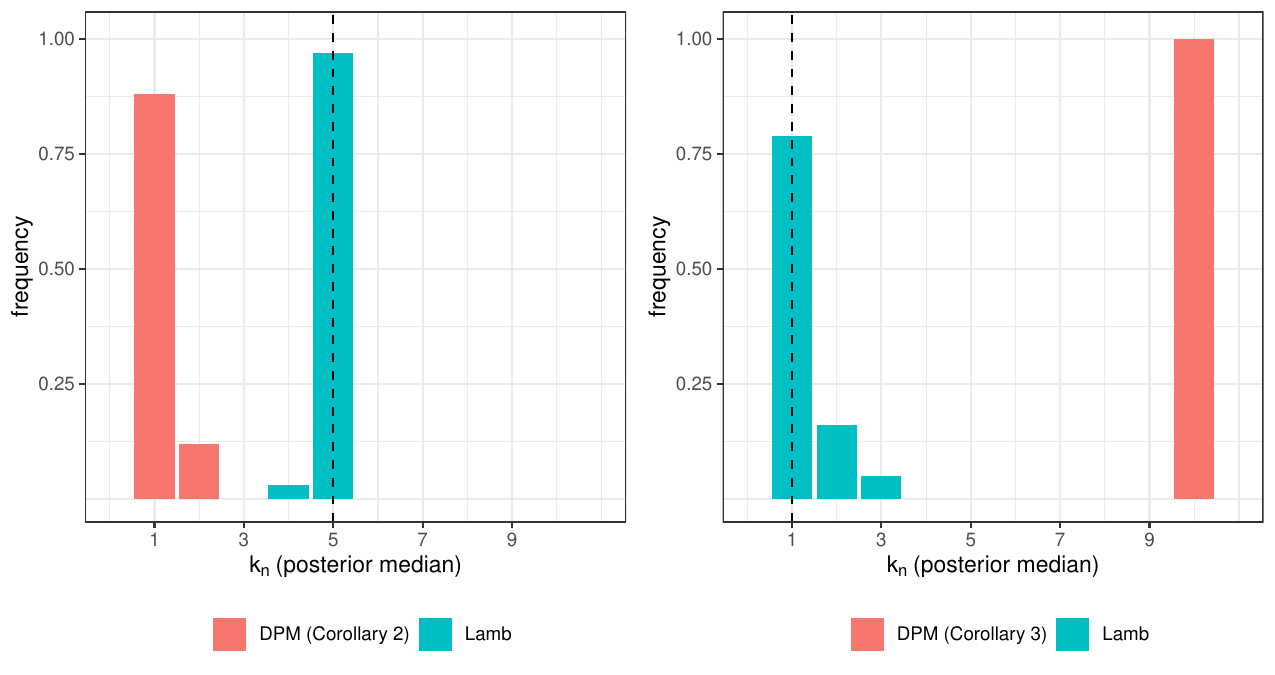}
		\caption{Empirical distribution of the posterior median number of clusters in 100 replicates under the first (left) and second (right) simulation experiment. 
			DPM hyperparameter specification 		satisfies Corollary \ref{cor:neg_result2} (left)  and Corollary \ref{cor:neg_result} (right). Vertical dashed lines represent the true number of clusters.}
		\label{fig:simpleboxplots}
	\end{figure}
	

	\vspace*{-.25in}
	\subsection{Recovering the Latent Space}
	\label{subsec:simulation_latent}

	To empirically illustrate  the robustness of  assumptions \ref{ass0} and \ref{ass2} used to prove the theory of Section \ref{sec:properties}, we perform a simple simulation study. 
	These conditions ensure that the data contain increasing information on the latent factors as $p$ increases.  
	Increasing $p$ means that we observe additional $y_{ij}$ variables for each subject. 
	Each of these variables can have very small correlation with the latent factor $\eta_{0i}$ and there will still be a build-up of information.  
	
	To see this, we generate random $y_{i}$ for $i=1, \dots, 4$ and $p\in\{20,200,1000\}$. 
	Data are generated as 
	$y_{i}=\Lambda_{0} \eta_{0i} +\epsilon_{i}$ where the factor loadings $\Lambda_0$'s are generated according to
	\vskip-1.5ex
	\begin{equation*}
		\lambda_{0jh}\simiid \pi \delta_{\{0\} }+ (1-\pi)\delta_{\{0.5\} },
	\end{equation*}
	\vskip-1ex \noindent
	where $\delta_{\{a\} }$ denotes a Dirac's delta mass at value $a$. 
	The true latent factors are simulated as $\eta_{0i,j}\sim\mn(i+j-1,0.05^{2})$ where $\eta_{0i}=(\eta_{0i1},\dots,\eta_{0id})\trans$.
	We consider two error distributions ensuring low signal-to-noise ratio, and specifically
	$\epsilon_{ij}\sim\mn(0, 25)$  and  $\epsilon_{ij}\sim t_{3}$ where $t_{3}$ denotes a central $t$ distribution with 3 degrees of freedom.
	We set $\pi=0.2$ and the latent dimension $d=2$.
	
	To examine the level of recovery, for the $m$-th MCMC iteration, we regress the true factors  with their current value in the $m$-th iteration. Specifically we stack all $\eta_{0 ij}$ across $i=1,\dots,n$ in a vector and use it as response variable, while using as predictor the vector containing all $\eta_{ij}^{(m)}$ of the $m$-th iteration. 
	We do this for each iteration after the burn-in.
	Clearly, the latent factors are non-identifiable due to the well known rotational ambiguity and thus  they can be learned only up to some non-singular matrix multiplication. 
	Hence, to quantify the accuracy in recovering  the latent space, we consider the  coefficient of determination $R^{2}$ of each fitted regression which is invariant of such identifiability issues. 
	Figure \ref{fig:increasinginfo} reports the results. 
	For both error distributions under consideration, as $p$ grows the posterior distributions of the coefficients of determination concentrate near one implying that with more variables we improve on the learning of the latent space even with low signal-to-noise ratios.
	
	\begin{figure}[H]
		\centering
		\includegraphics[width=0.8\linewidth]{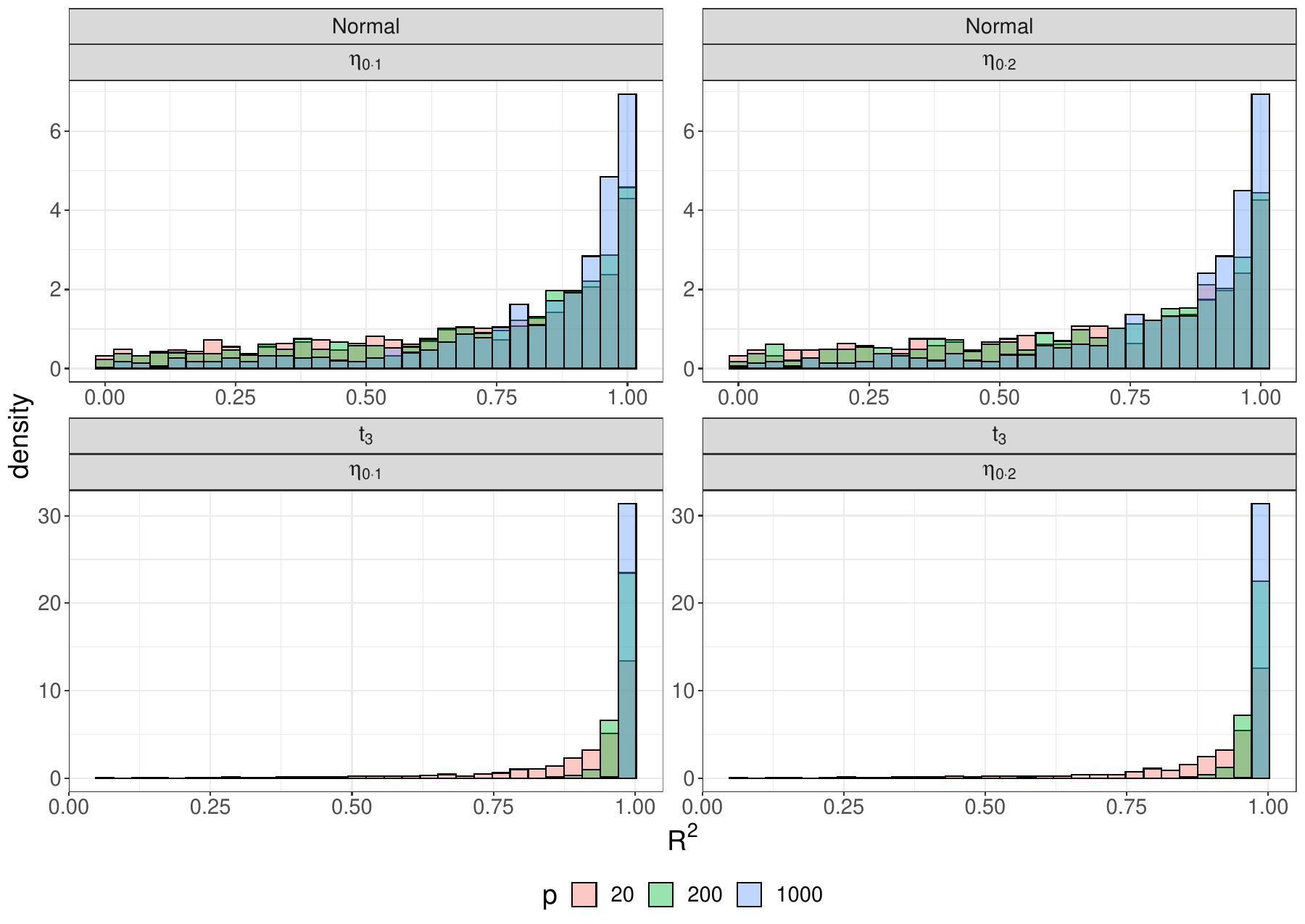}
		\vspace*{-.2in}
		\caption{Posterior distributions of the  coefficient of determination $R^{2}$ of the linear regressions of the true latent factors on the associated posterior samples
			for two error distributions ($\mn(0,25)$, in the upper quadrants and $t_{3}$, in the bottom quadrants). 
			The different dimensions $p$ are denoted by the different colors.}
		\label{fig:increasinginfo}
	\end{figure}
	
	\subsection{Small Sample Studies}
	\label{sm subsec:smallsamp}
	\vspace*{-.1in}	
	In this section, we do additional simulation studies.
	We consider the same setups considered in Section \ref{sec:simulation} of the main paper but take the sample size $n=500$.
	The true number of clusters is fixed to $k_{0}\in\{10,15\}$. 
	The results depicted in 	Figure \ref{sm fig:comparison}  are overall consistent with those reported in Section \ref{sec:simulation}.
	\vspace*{-.4in}	
	\begin{figure}[H]
		\centering
		\includegraphics[width=.5\textwidth]{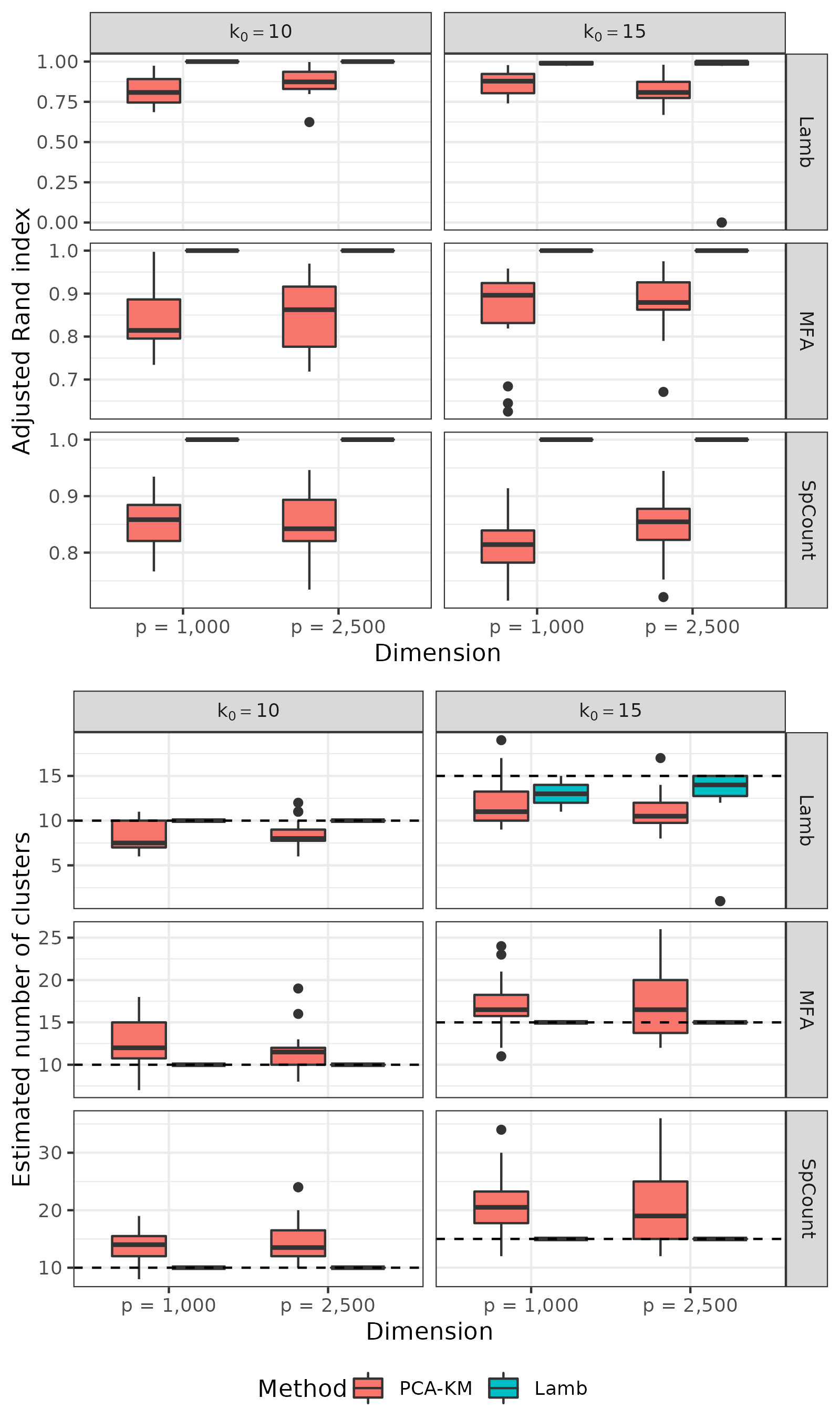}
		\vspace*{-.2in}
		\caption{{Comparison between our proposed Lamb and the two-stage PCA-KM approach:}
			Distributions of the adjusted Rand indices (upper plot) and estimated number of clusters (lower plot) in 20 replicated experiments. 
			Horizontal dashed lines denote the true number of clusters. 
			The {true simulation} scenarios, reported in each row, are labeled as Lamb for the model of Section \ref{sec:lamb}, MFA for mixture of factor analyzers and SpCount for the $\log$ transformed zero inflated sparse Poisson counts.}
		\label{sm fig:comparison}
	\end{figure}

	\vspace*{-.5in}
	\subsection{Figures Associated to Section \ref{sec:simulation}}
	\label{subsec:simulations_sm}			
		\vspace*{-.1in}	
	Figures \ref{sm fig: k10_p1000}-\ref{sm fig: k25_p2500} report the UMAP \citeplatex{umap} plots of the simulated datasets of Section \ref{sec:simulation},
	corresponding to the replicate with median adjusted Rand index \citeplatex{rand1971objective}. 
	In each figure, the upper and lower panels show the true clustering and the estimated clustering obtained by the Lamb model, respectively.
	Each figure's caption specifies the true number of clusters ($k_{0}$) and the dimension ($p$).

	\begin{figure}
		\centering
		\includegraphics[height=.225\textheight] {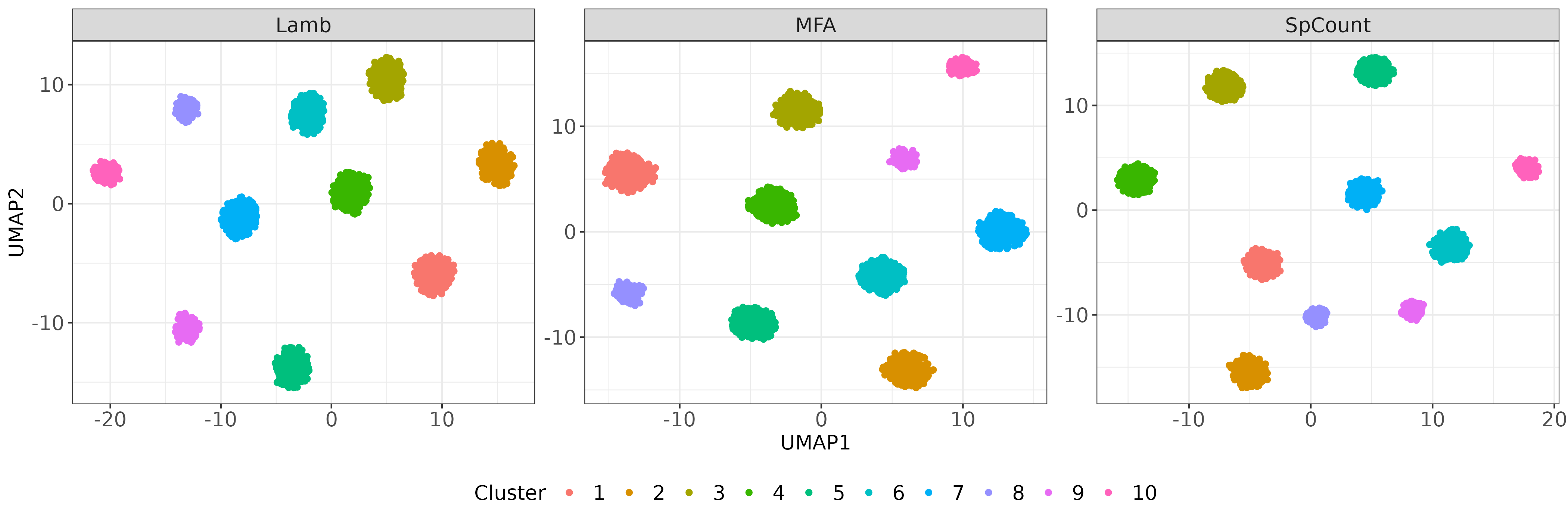}
		\includegraphics[height=.225\textheight] {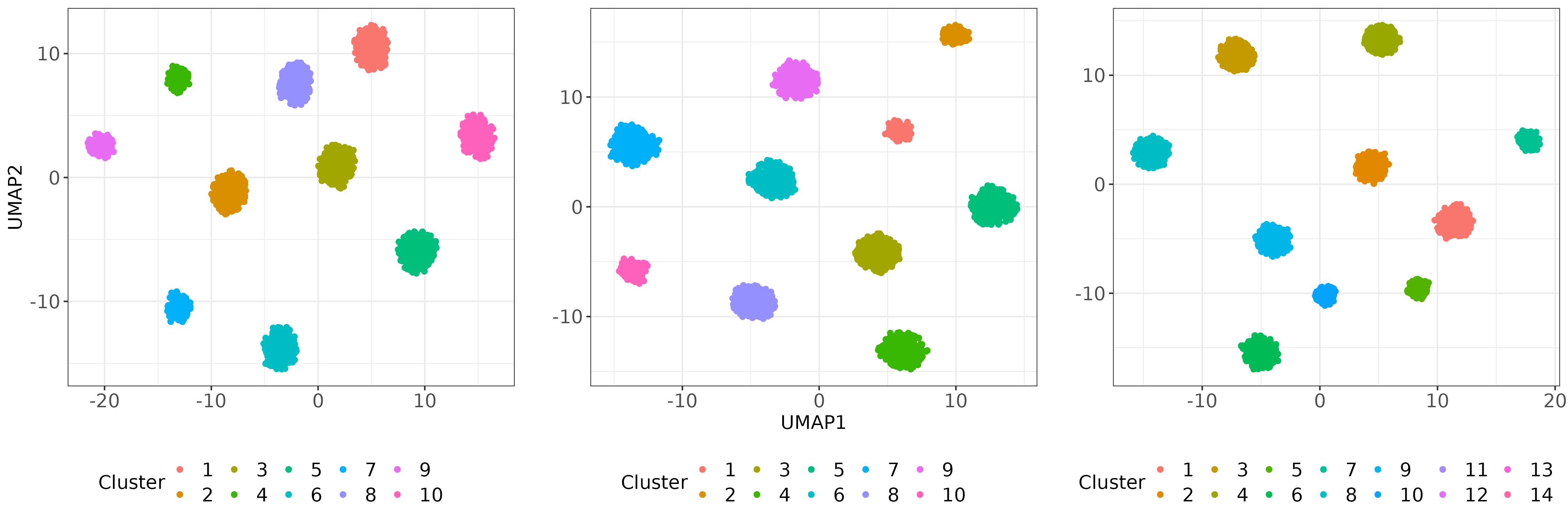}
		\caption{$k_{0}=10$, $p=1,000$.}
		\label{sm fig: k10_p1000}
	%
		\includegraphics[height=.225\textheight] {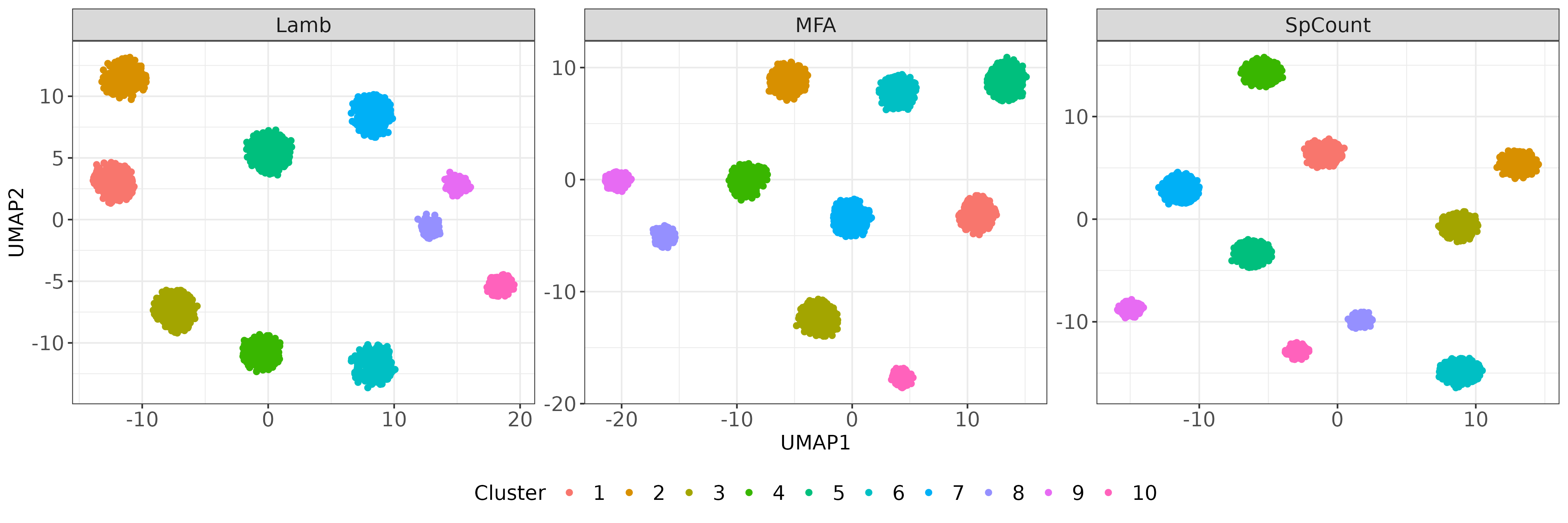}
		\includegraphics[height=.225\textheight] {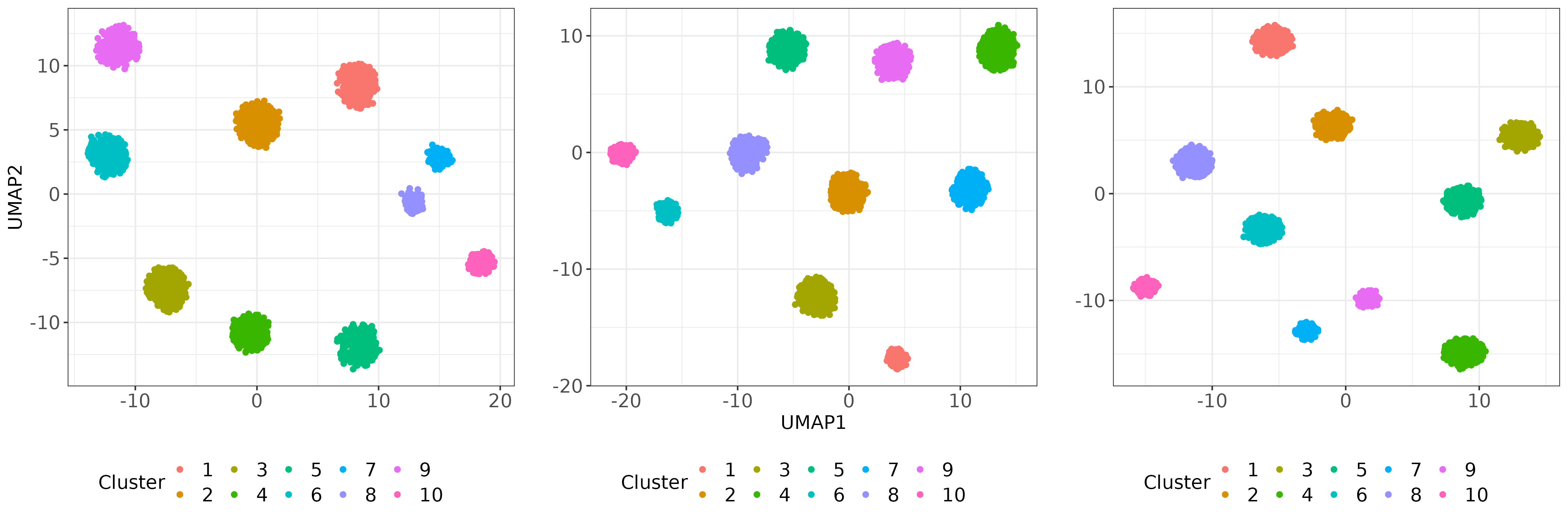}
		\caption{$k_{0}=10$, $p=2,500$
		}
		\label{sm fig: k10_p2500}
	\end{figure}
	
	\begin{figure}
		\includegraphics[height=.225\textheight] {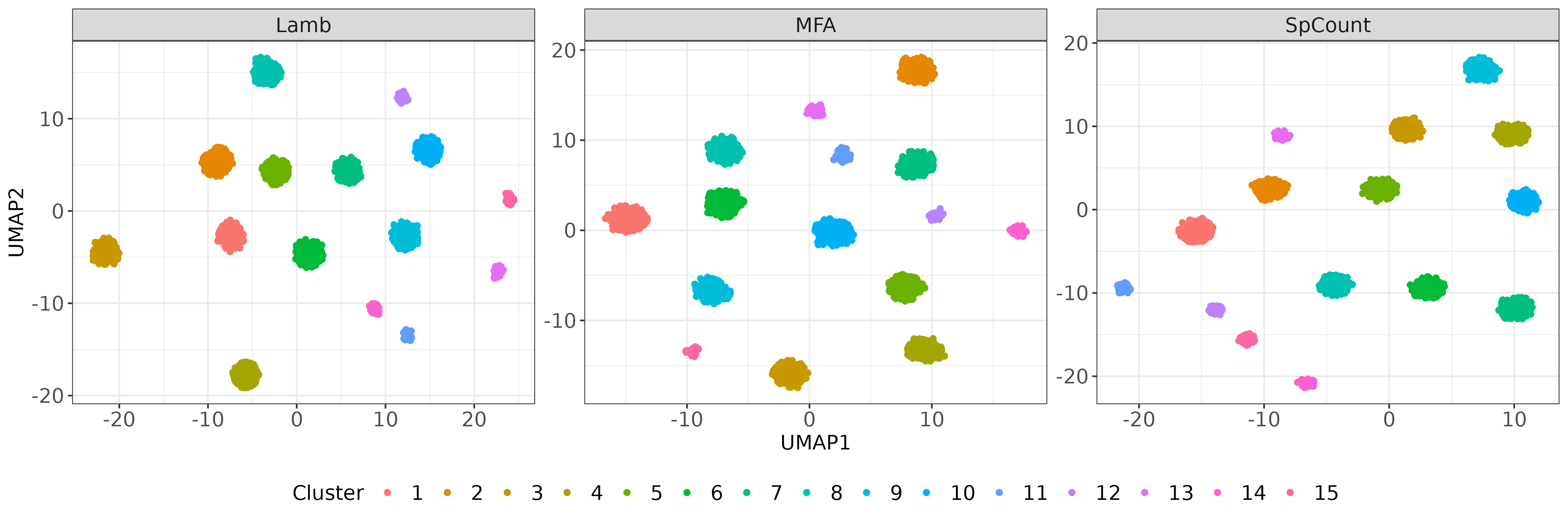}
		\includegraphics[height=.225\textheight] {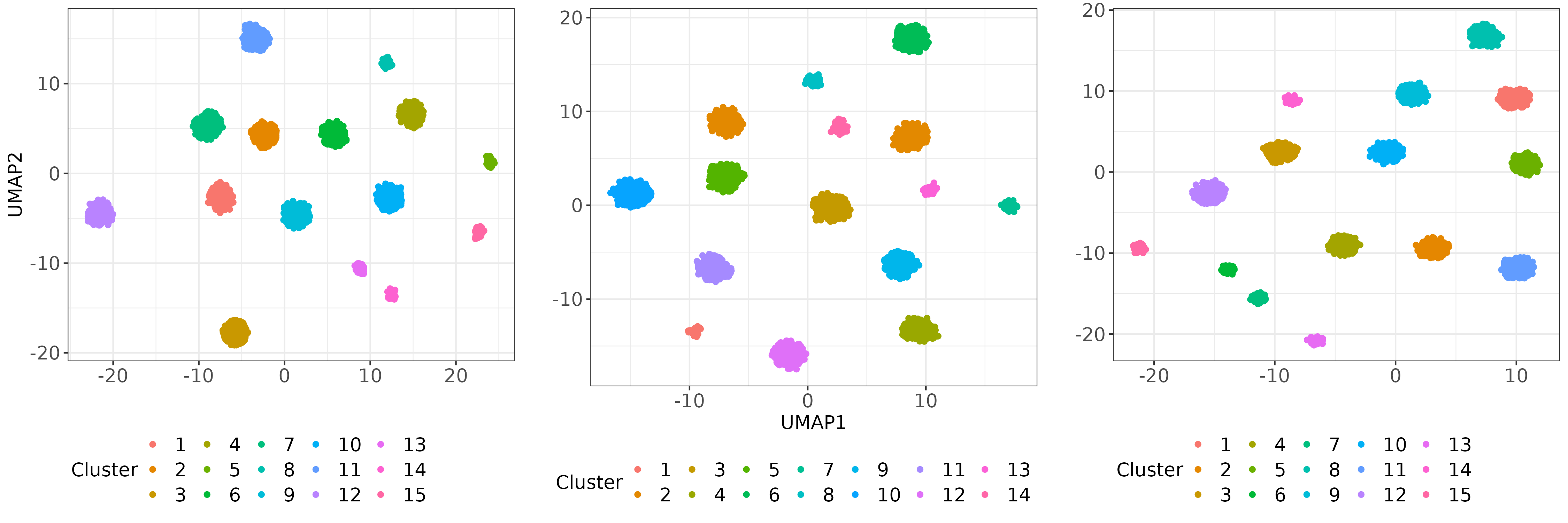}
		\caption{$k_{0}=15$, $p=1,000$ }
		\label{sm fig: k15_p1000}
	%
		\includegraphics[height=.225\textheight] {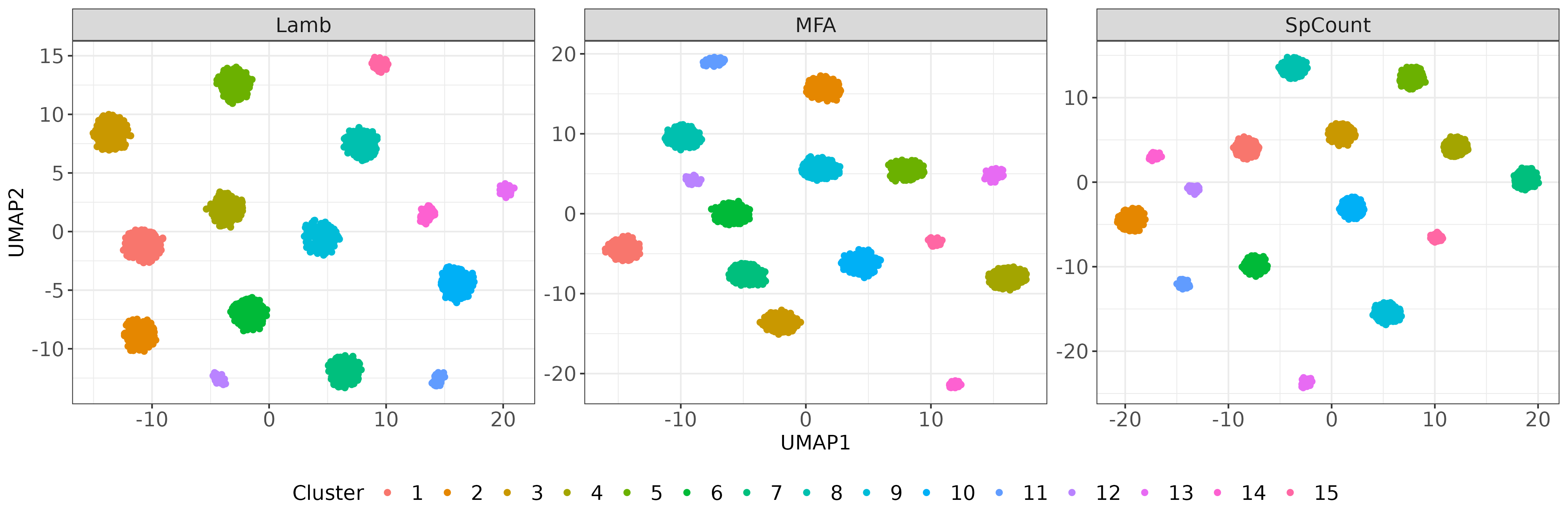}
		\includegraphics[height=.225\textheight] {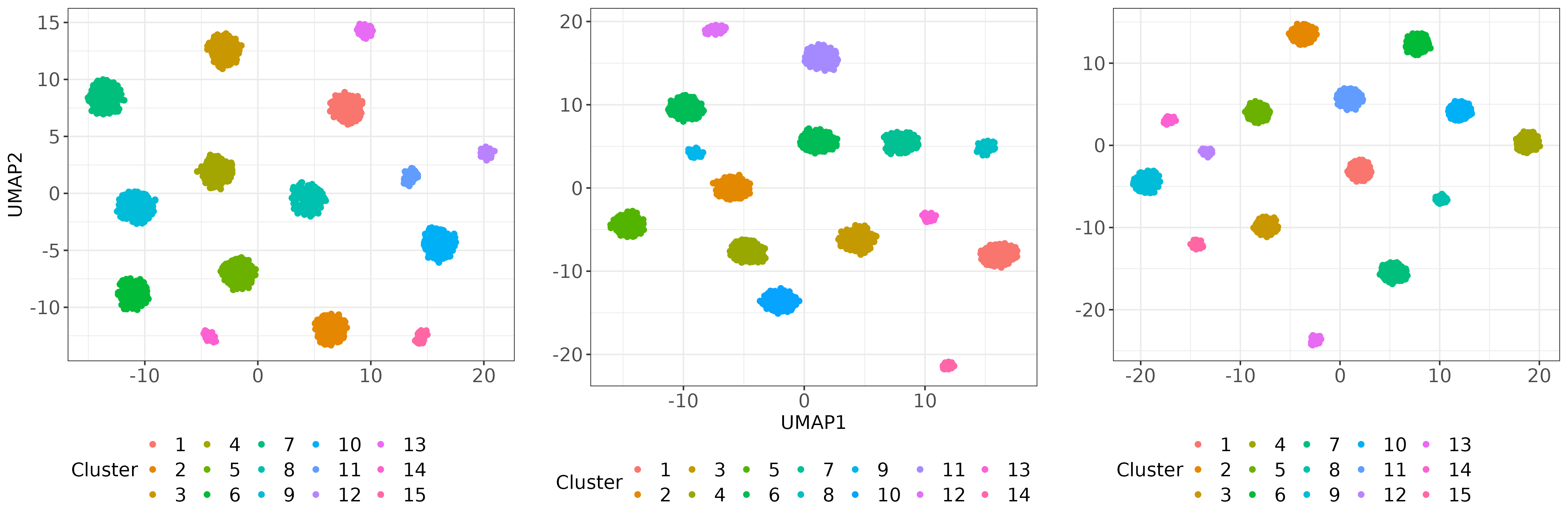}
		\caption{$k_{0}=15$, $p=2,500$ }
		\label{sm fig: k15_p2500}
	\end{figure}
	
	\begin{figure}
		\includegraphics[height=.225\textheight] {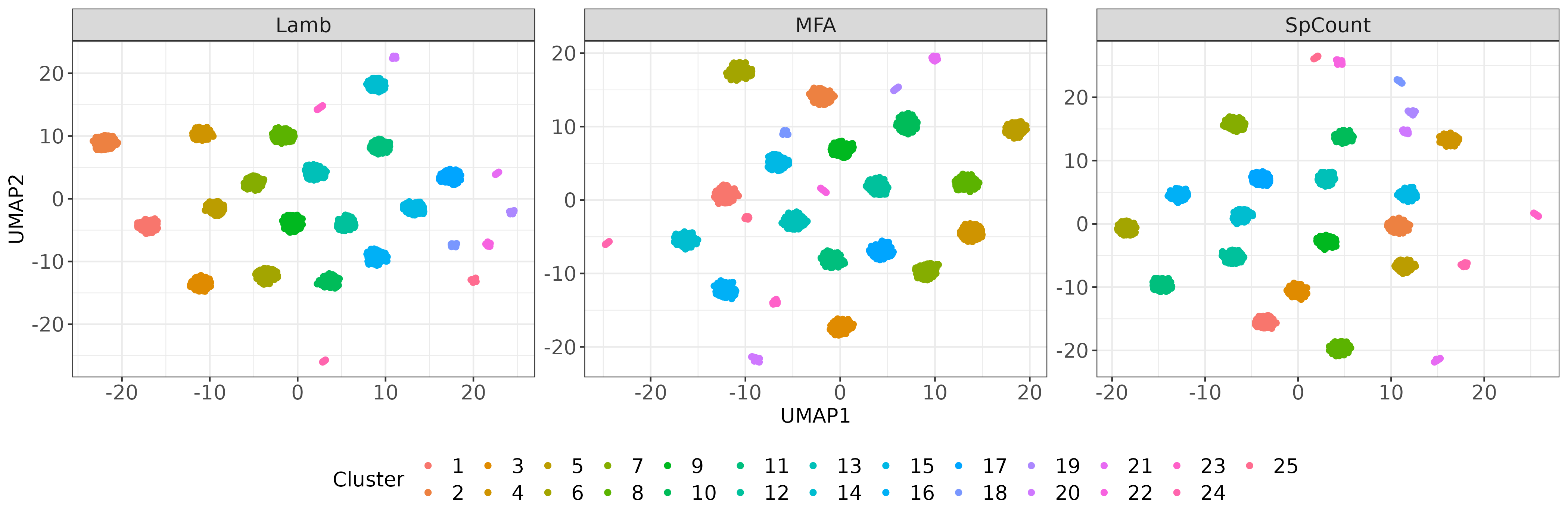}
		\includegraphics[height=.225\textheight] {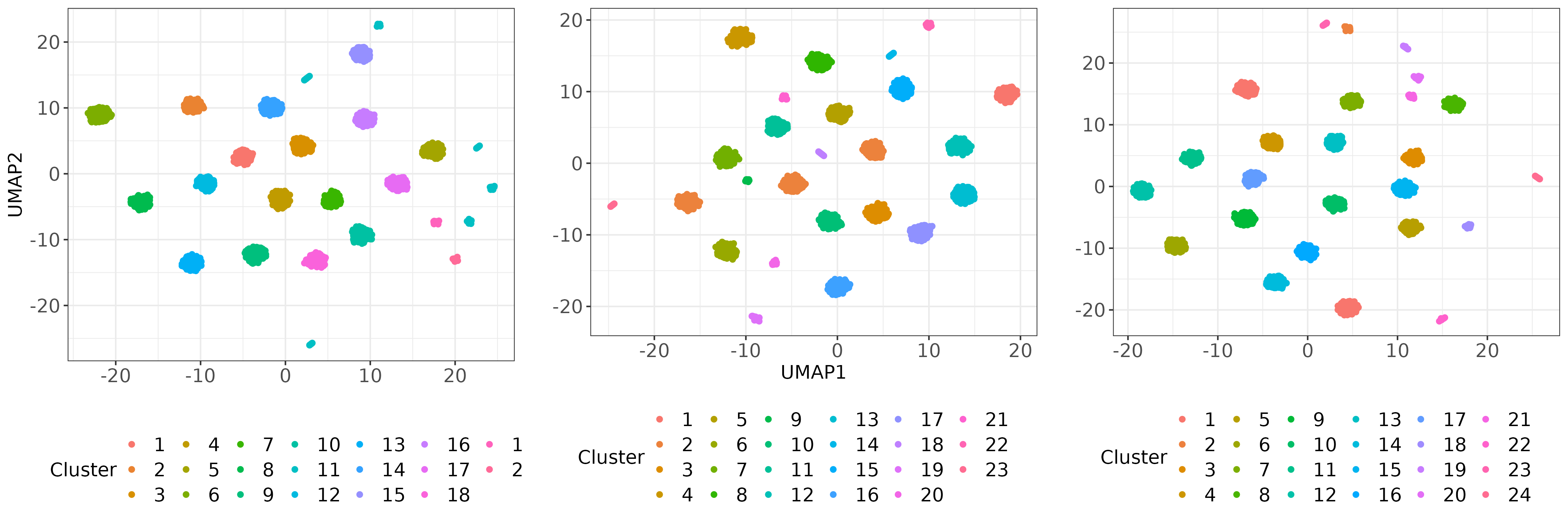}
		\caption{$k_{0}=25$, $p=1,000$ }
		\label{sm fig: k25_p1000}
	%
	%
		\centering
		\includegraphics[height=.225\textheight] {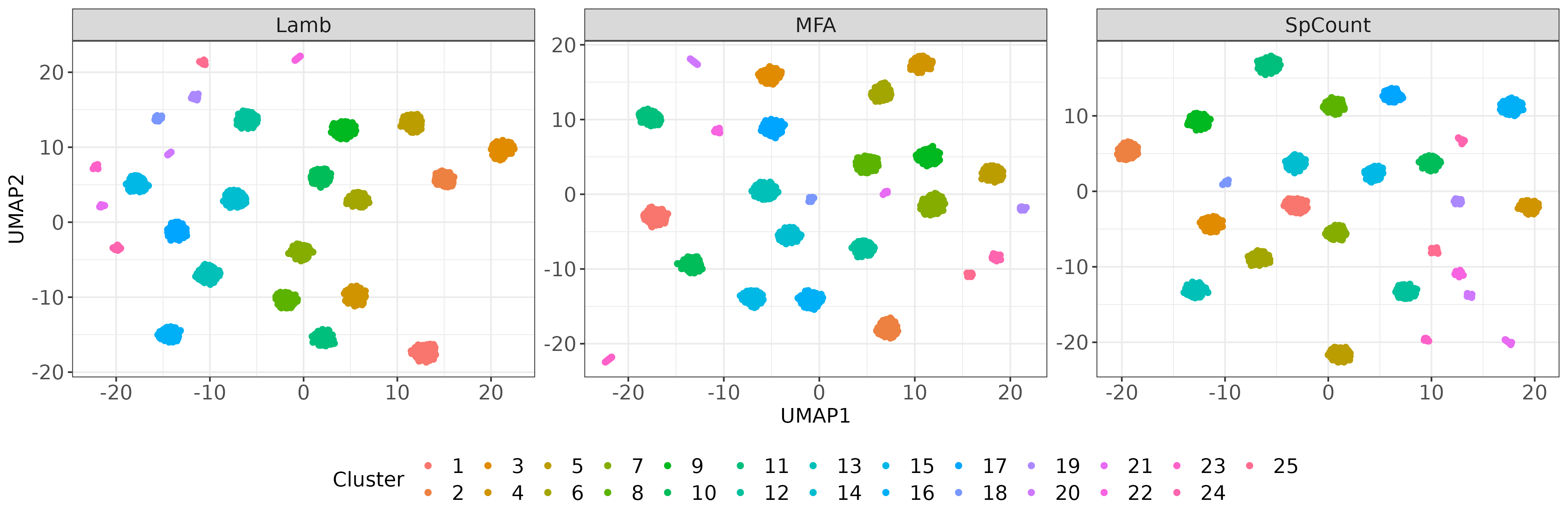}
		\includegraphics[height=.225\textheight] {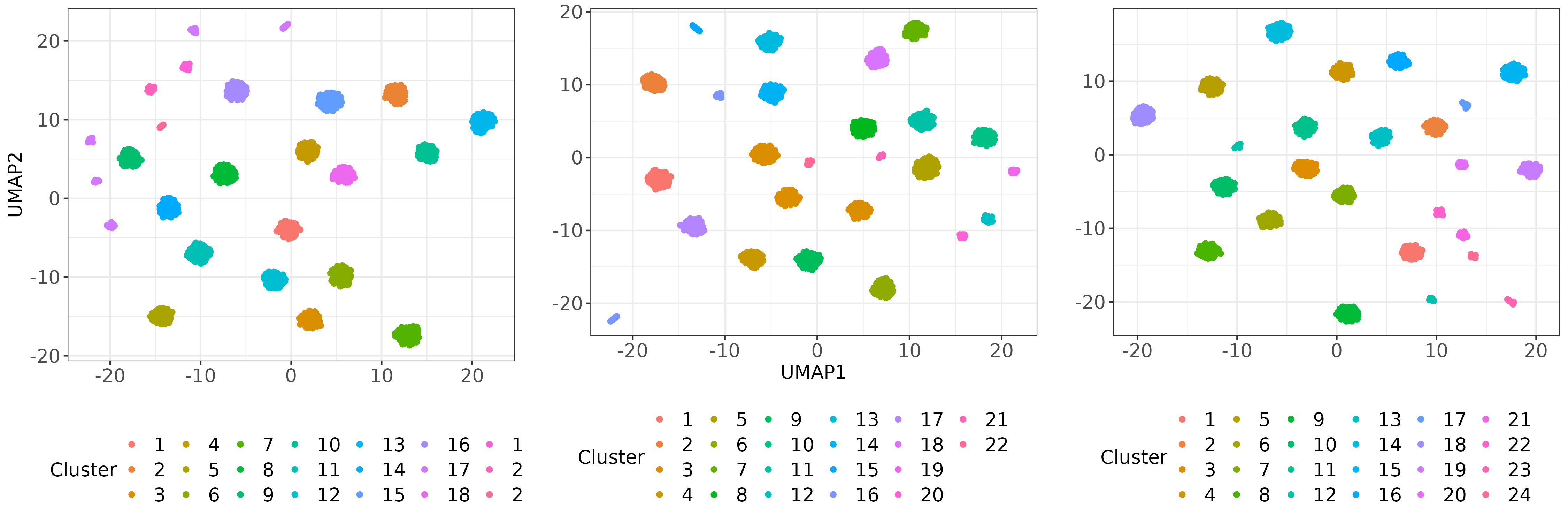}
		\caption{$k_{0}=25$, $p=2,500$ }
		\label{sm fig: k25_p2500}
	\end{figure}	

	\newpage

	\section{MCMC Convergence Diagnostics in the Cell Line Application}
	\label{sm sec:mcmc_diagnose}
	
	In this section, we provide convergence diagnostics of the MCMC sampler discussed
	in Section \ref{sec:posteriorcomputation}. 
	Note that most of the variables that are sampled are latent objects and not identifiable.
	Hence we compute the $\log$-likelihood of $y_{1:n}\mid \Lambda,\eta, \Sigma$ across the MCMC samples.
	On these $\log$-likelihoods, 
	we show traceplots and Geweke convergence diagnostics \citeplatex{geweke_mcmc}
	as implemented in the \texttt{coda} \texttt{R} package \citeplatex{coda_package}.
	The results are shown in Figure \ref{fig:gewekeplot} and they indicate evidence towards good mixing.

	\begin{figure}[ht]
		\centering
		\includegraphics[trim={0cm .5cm 1cm, 2cm},clip, height=0.35\textheight]{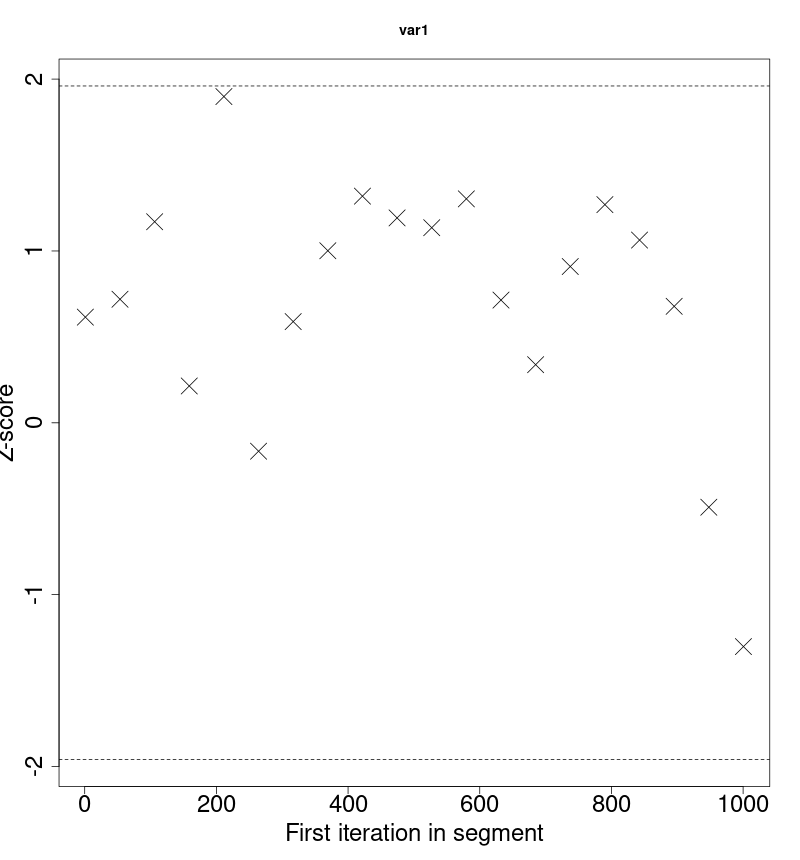}~~~
		\includegraphics[trim={0cm .5cm 1cm, 2cm},clip, height=0.35\textheight]{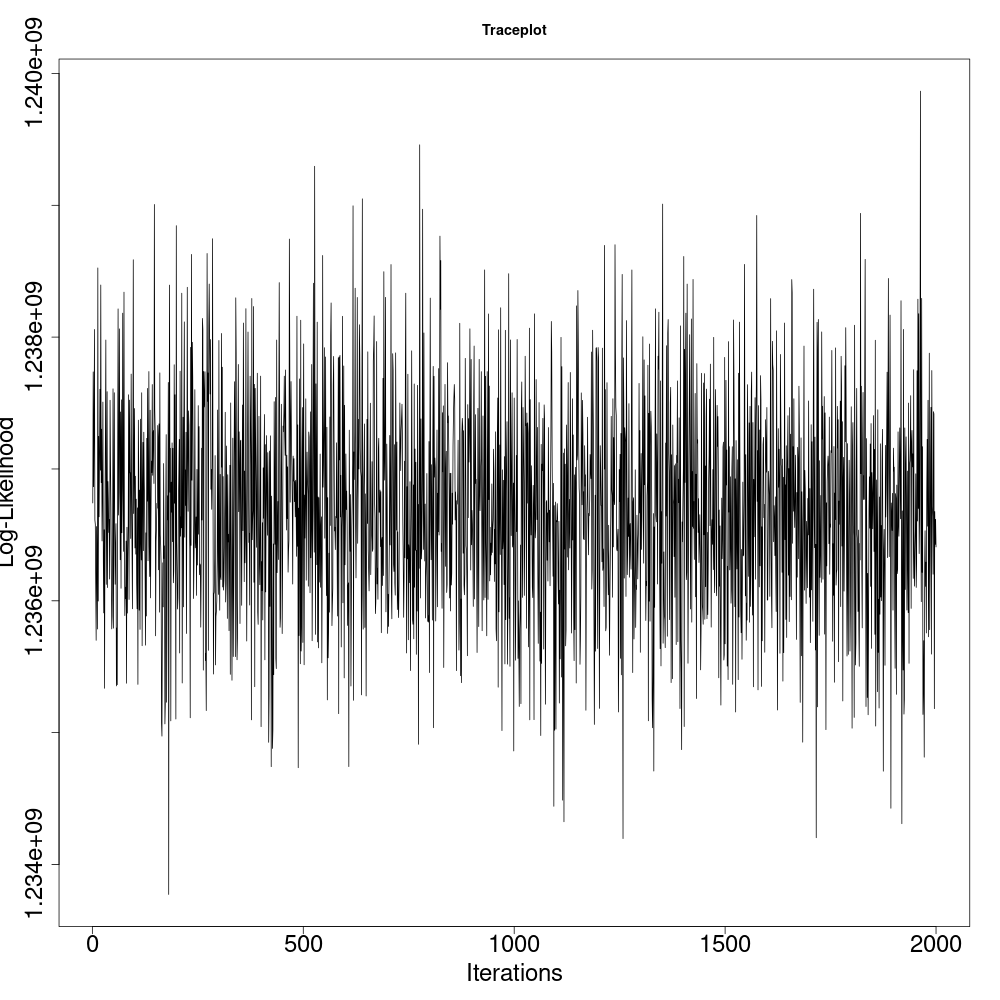}
		\caption{MCMC Convergence Diagnostics in the Cell Line Application: The joint $\log$-likelihoods of the $y_{1:n}\mid \Lambda,\eta, \Sigma$ are first calculated across the MCMC iterations.
			The Geweke convergence diagnostic on the $\log$-likelihoods is shown in the left panel and
			their traceplot in the right panel.}
		\label{fig:gewekeplot}
	\end{figure}

\bibliographystylelatex{natbib}
\bibliographylatex{refs}
\end{document}